\let\oldvec\vec 
\documentclass[runningheads,envcountsect,envcountsame,final]{llncs}
\let\vec\oldvec 

\pdfoutput=1
\usepackage[T1]{fontenc}
\usepackage{graphicx}
\usepackage{tikz}
\usetikzlibrary{calc}
\usetikzlibrary{positioning}
\usetikzlibrary{arrows,automata}
\usetikzlibrary{patterns,patterns.meta}
\usepackage{xspace}
\usepackage{xcolor}
\usepackage{amsmath}
\usepackage{thmtools}
\usepackage{amssymb}
\usepackage{mathrsfs}
\usepackage{ifdraft}
\usepackage{wrapfig}
\usepackage{algpseudocode}
\usepackage{subcaption}
\usepackage{verbatim}
\usepackage{booktabs}
\usepackage{multirow}
\usepackage{paralist}
\usepackage{cite}
\usepackage{lipsum}
\usepackage{amsbsy}
\usepackage{bm}
\usepackage{times}
\usepackage{style/orcidlink}

\RequirePackage{centernot}
\RequirePackage[Symbol]{upgreek}

\RequirePackage{etoolbox}
\RequirePackage{mathtools}

\let\originalleft\left
\let\originalright\right
\renewcommand{\left}{\mathopen{}\mathclose\bgroup\originalleft}
\renewcommand{\right}{\aftergroup\egroup\originalright}

\newcommand{\pvec}[1]{\vec{#1}\mkern2mu\vphantom{#1}}

\newcommand{\Nat}{\mathbb{N}}

\newcommand{\Procs}{\ensuremath{\mathcal{P}}}
\definecolor{roleColor}{rgb}{0.1, 0.3, 0.1}\newcommand{\roleCol}[1]{{\color{roleColor}#1}}\newcommand{\roleFmt}[1]{\boldsymbol{\roleCol{\mathtt{#1}}}}\newcommand{\roleP}[1][]{\ifempty{#1}{{\color{roleColor}\roleFmt{p}}}{{\color{roleColor}\roleFmt{p}_{#1}}}}
\newcommand{\roleVar}[1][]{\ifempty{#1}{{\color{roleColor}\roleFmt{x}}}{{\color{roleColor}\roleFmt{x}_{#1}}}}

\newcommand{\procA}{{\color{roleColor}\roleFmt{p}}}
\newcommand{\procB}{{\color{roleColor}\roleFmt{q}}}
\newcommand{\procC}{{\color{roleColor}\roleFmt{r}}}
\newcommand{\procD}{{\color{roleColor}\roleFmt{s}}}

\newcommand{\run}{\rho}

\newcommand{\Alphabet}{\Sigma}

\newcommand{\val}{\ensuremath{m}}
\newcommand{\MsgVals}{\ensuremath{\mathcal{V}}}

\newcommand{\msgO}{\textcolor{orange}{o}}
\newcommand{\msgB}{\textcolor{blue}{b}}
\newcommand{\msgM}{\textcolor{magenta}{m}}

\newcommand{\CSM}[1]{\ensuremath{\{\!\!\{#1_\procA\}\!\!\}_{\procA \in \Procs}}}
\newcommand{\CSMl}[1]{\ensuremath{\{\!\!\{{#1}\}\!\!\}_{\procA \in \Procs}}}

\newcommand{\emptystring}{\varepsilon}

\newcommand{\set}[1]{\{#1\}}
\newcommand{\lang}{\mathcal{L}}
\newcommand{\interswaplang}{\mathcal{C}^{\interswap}}
\newcommand{\SyncToAsync}{\ensuremath{\operatorname{\texttt{\upshape{split}}}}}
\newcommand{\channels}{\ensuremath{\mathsf{Chan}}}
\newcommand{\channel}[2]{\ensuremath{#1,#2}}

\newcommand{\trace}{\texttt{\upshape{trace}}}

\newcommand{\GG}{\mathbf{G}}
\newcommand{\getMu}{\mathit{get\mu}}
\newcommand{\getMuG}{\getMu_\GG}

\newcommand{\semglobal}{\ensuremath{\mathsf{GAut}}}
\newcommand{\semglobalsync}{\ensuremath{\mathsf{GAut}}}
\newcommand{\projerasuresymb}{\downarrow}
\newcommand{\projerasure}[2]{\ensuremath{\semglobal(#1)\negmedspace\!\projerasuresymb_{#2}}}
\newcommand{\projerasuretrans}{\ensuremath{\delta_\projerasuresymb}}
\newcommand{\AlphSync}{\ensuremath{Σ_{\mathit{sync}}}}
\newcommand{\AlphAsync}{\ensuremath{Σ_{\mathit{async}}}}

\newcommand{\restrict}[2]{{#1}|_{#2}}
\newcommand{\lfp}{\mathrm{lfp}}
\newcommand{\fin}{\operatorname{fin}}

\newcommand{\infi}{\operatorname{inf}}

\newcommand{\interswap}{\ensuremath{\sim}}

\def \ifempty#1{\def\temp{#1} \ifx\temp\empty }

\newcommand{\snd}[3]{#1\kern 0.13em{\triangleright}\kern 0.08em#2\kern 0.08em{\operatorname{!}}\kern 0.08em#3}
\newcommand{\rcv}[3]{#2\kern 0.08em{\triangleleft}\kern 0.13em#1\kern 0.08em{\operatorname{?}}\kern 0.08em#3}

\newcommand{\ssnd}[2]{#1\kern 0.08em{!}\kern 0.08em#2}
\newcommand{\srcv}[2]{#1\kern 0.08em{?}\kern 0.08em#2}
\newcommand{\msgFromTo}[3]{#1\kern 0.12em{\to}\kern 0.12em#2\kern 0.12em {:} \kern 0.12em #3}

\newcommand{\preforder}{\ensuremath{\leq}}

\newcommand{\proofend}{\hfill$\qed$}
\newcommand{\myparagraph}[1]{\smallskip\noindent\textbf{#1}}

\newcommand{\proj}{{\ensuremath{\upharpoonright}}}

\newcommand{\wproj}{{\ensuremath{\Downarrow}}}

\newcommand{\merge}{\sqcap}
\def\mmerge{\mathrel{\ThisStyle{\stretchrel*{\ooalign{\raise0.2\LMex\hbox{$\SavedStyle\sqcap$}\cr \raise-0.2\LMex\hbox{$\SavedStyle\sqcap$}}}{\sqcap}}}}
\def\mmmerge{\mathrel{\ThisStyle{\stretchrel*{\ooalign{\raise0.6\LMex\hbox{$\SavedStyle\sqcap$}\cr \raise0.2\LMex\hbox{$\SavedStyle\sqcap$}\cr \raise-0.2\LMex\hbox{$\SavedStyle\sqcap$}}}{\sqcap}}}}

\newcommand{\blockedset}{\ensuremath{\mathcal{B}}}
\newcommand{\semavail}{\ensuremath{M}}

\newcommand{\semavaildef}[3]{\ensuremath{\semavail^{#1, #2}_{(#3\ldots)}}}

\newcommand{\powersetof}[1]{2^{#1}}

\newcommand{\union}{\cup}
\newcommand{\inters}{\cap}
\newcommand{\Union}{\bigcup}
\newcommand{\Inters}{\bigcap}
\newcommand{\dunion}{\uplus}

\DeclarePairedDelimiter\card{\lvert}{\rvert}

\providecommand{\implies}{\Rightarrow}

\providecommand{\coloneq}{\mathrel{\mathop:}=} \providecommand{\Coloneqq}{\mathrel{\mathop{::}}=} \newcommand{\is}{\coloneq}

 \def\st{\colon}

\def\grammOr{\hspace{3pt}\mid\hspace{3pt}}
\def\grammIs{\Coloneqq}

\begingroup
\catcode`\|=\active \gdef\@grammar@bar{\catcode`\|=\active \def|{\grammOr}}
\endgroup

\newcommand{\gramm}[1]{\begingroup
  \def\is{\grammIs}\@grammar@bar #1\endgroup }

\newenvironment{grammar}{\begin{equation*}\def\is{& \grammIs }\@grammar@bar \aligned }
{\endaligned \end{equation*}\aftergroup\ignorespaces }

\newcommand{\hole}{\hbox{-}}

\newcommand{\subsetcons}[2]{\ensuremath{\mathscr{C}(#1,#2)}}
\newcommand{\subsetproj}[2]{\ensuremath{\mathscr{C}(#1,#2)}}
\newcommand{\subsetprojCSM}{\CSMl{\subsetproj{\GG}{\procA}}}

\newcommand{\transAnnoFunc}{\ensuremath{\operatorname{tr-orig}}}
\newcommand{\transAnnoFunDest}{\ensuremath{\operatorname{tr-dest}}}

\newcommand{\globcomplocal}[3]{\ensuremath{\operatorname{R}^#1_{#2}(#3)}}

\newcommand{\CSMhole}{\mathcal{A}[\cdot]_\procA}
\newcommand{\CSMholew}[1]{\mathcal{A}[#1]_\procA}

\renewcommand\fbox{\fcolorbox{red}{white}}
 \usepackage{newunicodechar}
\newunicodechar{∃}{\ensuremath{\exists}}
\newunicodechar{∀}{\ensuremath{\forall}}
\newunicodechar{θ}{\ensuremath{\theta}}
\newunicodechar{τ}{\ensuremath{\tau}}
\newunicodechar{φ}{\ensuremath{\varphi}}
\newunicodechar{ξ}{\ensuremath{\xi}}
\newunicodechar{ζ}{\ensuremath{\zeta}}
\newunicodechar{ψ}{\ensuremath{\psi}}
\newunicodechar{π}{\ensuremath{\pi}}
\newunicodechar{α}{\ensuremath{\alpha}}
\newunicodechar{β}{\ensuremath{\beta}}
\newunicodechar{γ}{\ensuremath{\gamma}}
\newunicodechar{δ}{\ensuremath{\delta}}
\newunicodechar{ε}{\ensuremath{\varepsilon}}
\newunicodechar{κ}{\ensuremath{\kappa}}
\newunicodechar{λ}{\ensuremath{\lambda}}
\newunicodechar{μ}{\ensuremath{\mu}}
\newunicodechar{ρ}{\ensuremath{\rho}}
\newunicodechar{σ}{\ensuremath{\sigma}}
\newunicodechar{ω}{\ensuremath{\omega}}
\newunicodechar{Γ}{\ensuremath{\Gamma}}
\newunicodechar{Φ}{\ensuremath{\Phi}}
\newunicodechar{Δ}{\ensuremath{\Delta}}
\newunicodechar{Σ}{\ensuremath{\Sigma}}
\newunicodechar{Π}{\ensuremath{\Pi}}
\newunicodechar{∑}{\ensuremath{\Sigma}}
\newunicodechar{∏}{\ensuremath{\Pi}}
\newunicodechar{Θ}{\ensuremath{\Theta}}
\newunicodechar{Ω}{\ensuremath{\Omega}}
\newunicodechar{⇒}{\ensuremath{\Rightarrow}}
\newunicodechar{⇐}{\ensuremath{\Leftarrow}}
\newunicodechar{⇔}{\ensuremath{\Leftrightarrow}}
\newunicodechar{→}{\ensuremath{\rightarrow}}
\newunicodechar{←}{\ensuremath{\leftarrow}}
\newunicodechar{↔}{\ensuremath{\leftrightarrow}}
\newunicodechar{¬}{\ensuremath{\neg}}
\newunicodechar{∧}{\ensuremath{\land}}
\newunicodechar{∨}{\ensuremath{\lor}}
\newunicodechar{≠}{\ensuremath{\neq}}
\newunicodechar{≡}{\ensuremath{\equiv}}
\newunicodechar{∼}{\ensuremath{\sim}}
\newunicodechar{≈}{\ensuremath{\approx}}
\newunicodechar{≥}{\ensuremath{\geq}}
\newunicodechar{≤}{\ensuremath{\leq}}
\newunicodechar{≫}{\ensuremath{\gg}}
\newunicodechar{≪}{\ensuremath{\ll}}
\newunicodechar{∅}{\ensuremath{\emptyset}}
\newunicodechar{⊆}{\ensuremath{\subseteq}}
\newunicodechar{⊂}{\ensuremath{\subset}}
\newunicodechar{∩}{\ensuremath{\cap}}
\newunicodechar{⋂}{\ensuremath{\cap}}
\newunicodechar{∪}{\ensuremath{\cup}}
\newunicodechar{⋃}{\ensuremath{\cup}}
\newunicodechar{⊎}{\ensuremath{\uplus}}
\newunicodechar{∈}{\ensuremath{\in}}
\newunicodechar{∉}{\ensuremath{\not\in}}
\newunicodechar{⊤}{\ensuremath{\top}}
\newunicodechar{⊥}{\ensuremath{\bot}}
\newunicodechar{₀}{\ensuremath{_0}}
\newunicodechar{₁}{\ensuremath{_1}}
\newunicodechar{₂}{\ensuremath{_2}}
\newunicodechar{₃}{\ensuremath{_3}}
\newunicodechar{₄}{\ensuremath{_4}}
\newunicodechar{₅}{\ensuremath{_5}}
\newunicodechar{₆}{\ensuremath{_6}}
\newunicodechar{₇}{\ensuremath{_7}}
\newunicodechar{₈}{\ensuremath{_8}}
\newunicodechar{₉}{\ensuremath{_9}}
\newunicodechar{⁰}{\ensuremath{^0}}
\newunicodechar{¹}{\ensuremath{^1}}
\newunicodechar{²}{\ensuremath{^2}}
\newunicodechar{³}{\ensuremath{^3}}
\newunicodechar{⁴}{\ensuremath{^4}}
\newunicodechar{⁵}{\ensuremath{^5}}
\newunicodechar{⁶}{\ensuremath{^6}}
\newunicodechar{⁷}{\ensuremath{^7}}
\newunicodechar{⁸}{\ensuremath{^8}}
\newunicodechar{⁹}{\ensuremath{^9}}
\newunicodechar{𝔹}{\ensuremath{\mathbb{B}}}
\newunicodechar{ℝ}{\ensuremath{\mathbb{R}}}
\newunicodechar{ℕ}{\ensuremath{\mathbb{N}}}
\newunicodechar{ℂ}{\ensuremath{\mathbb{C}}}
\newunicodechar{ℚ}{\ensuremath{\mathbb{Q}}}
\newunicodechar{𝕋}{\ensuremath{\mathbb{T}}}
\newunicodechar{𝕏}{\ensuremath{\mathbb{X}}}
\newunicodechar{ℤ}{\ensuremath{\mathbb{Z}}}
\newunicodechar{✓}{\ding{51}}
\newunicodechar{✗}{\ding{55}}
\newunicodechar{◊}{\ensuremath{\lozenge}}
\newunicodechar{□}{\ensuremath{\square}}
\newunicodechar{𝓐}{\ensuremath{\mathcal{A}}}
\newunicodechar{𝓑}{\ensuremath{\mathcal{B}}}
\newunicodechar{𝓒}{\ensuremath{\mathcal{C}}}
\newunicodechar{𝓓}{\ensuremath{\mathcal{D}}}
\newunicodechar{𝓔}{\ensuremath{\mathcal{E}}}
\newunicodechar{𝓕}{\ensuremath{\mathcal{F}}}
\newunicodechar{𝓖}{\ensuremath{\mathcal{G}}}
\newunicodechar{𝓗}{\ensuremath{\mathcal{H}}}
\newunicodechar{𝓘}{\ensuremath{\mathcal{I}}}
\newunicodechar{𝓙}{\ensuremath{\mathcal{J}}}
\newunicodechar{𝓚}{\ensuremath{\mathcal{K}}}
\newunicodechar{𝓛}{\ensuremath{\mathcal{L}}}
\newunicodechar{𝓜}{\ensuremath{\mathcal{M}}}
\newunicodechar{𝓝}{\ensuremath{\mathcal{N}}}
\newunicodechar{𝓞}{\ensuremath{\mathcal{O}}}
\newunicodechar{𝓟}{\ensuremath{\mathcal{P}}}
\newunicodechar{𝓠}{\ensuremath{\mathcal{Q}}}
\newunicodechar{𝓡}{\ensuremath{\mathcal{R}}}
\newunicodechar{𝓢}{\ensuremath{\mathcal{S}}}
\newunicodechar{𝓣}{\ensuremath{\mathcal{T}}}
\newunicodechar{𝓤}{\ensuremath{\mathcal{U}}}
\newunicodechar{𝓥}{\ensuremath{\mathcal{V}}}
\newunicodechar{𝓦}{\ensuremath{\mathcal{W}}}
\newunicodechar{𝓧}{\ensuremath{\mathcal{X}}}
\newunicodechar{𝓨}{\ensuremath{\mathcal{Y}}}
\newunicodechar{𝓩}{\ensuremath{\mathcal{Z}}}
\newunicodechar{…}{\ensuremath{\ldots}}
\newunicodechar{∗}{\ensuremath{\ast}}
\newunicodechar{⊢}{\ensuremath{\vdash}}
\newunicodechar{⊧}{\ensuremath{\models}}
\newunicodechar{′}{\ensuremath{'}}
\newunicodechar{″}{\ensuremath{''}}
\newunicodechar{‴}{\ensuremath{'''}}
\newunicodechar{∥}{\ensuremath{\|}}
\newunicodechar{⊕}{\ensuremath{\oplus}}
\newunicodechar{⁺}{\ensuremath{^+}}
\newunicodechar{⊇}{\ensuremath{\supseteq}}
\newunicodechar{∘}{\ensuremath{\circ}}
\newunicodechar{∙}{\ensuremath{\cdot}}
\newunicodechar{⋅}{\ensuremath{\cdot}}
\newunicodechar{≈}{\ensuremath{\approx}}
\newunicodechar{×}{\ensuremath{\times}}
\newunicodechar{∞}{\ensuremath{\infty}}
\newunicodechar{⊑}{\ensuremath{\sqsubseteq}}
 
\ifoptionfinal
{\usepackage[disable]{todonotes}}
{\usepackage{todonotes}}

\usepackage{hyperref}
\usepackage[capitalise]{cleveref}
\usepackage[shortlabels]{enumitem}

\crefformat{section}{\S#2#1#3} \crefmultiformat{section}{\S#2#1#3}{ and~\S#2#1#3}{, \S#2#1#3}{, and~\S#2#1#3} \crefrangeformat{section}{\S#3#1#4 to~\S#5#2#6}  

\newtoggle{techrep}
\toggletrue{techrep}
\iftoggle{techrep}{
    \newcommand{\appendixRef}[1]{\cref{#1}}
}{
    \newcommand{\appendixRef}[1]{the extended version~\cite{li2024subtyping}}
}

\begin{document}
\title{Deciding Subtyping for \\Asynchronous Multiparty Sessions} \ifoptionfinal
{\authorrunning{E.\ Li, F.\ Stutz, and T.\ Wies}}
{\titlerunning{Deciding Subtyping for Asynchronous Multiparty Sessions}}

\author{Authors withheld}
\author{
Elaine Li\thanks{corresponding author}\inst{1}\orcidlink{0000-0003-0173-4498} \and
Felix Stutz\inst{2}\orcidlink{0000-0003-3638-4096} \and
Thomas Wies\inst{1}\orcidlink{0000-0003-4051-5968}
}
\institute{New York University, New York, USA \email{efl9013@nyu.edu, wies@cs.nyu.edu} \and
Max Planck Institute for Software Systems, Kaiserslautern, Germany \email{fstutz@mpi-sws.org}}

\maketitle              \begin{abstract}
Multiparty session types (MSTs) are a type-based approach to verifying communication protocols, represented as global types in the framework. 
We present a precise subtyping relation for asynchronous MSTs with communicating state machines (CSMs) as implementation model.
We address two problems: when can a local implementation safely substitute another, and when does an arbitrary CSM implement a global type? 
We define safety with respect to a given global type, in terms of subprotocol fidelity and deadlock freedom. 
Our implementation model subsumes existing work which considers local types with restricted choice. 
We exploit the connection between MST subtyping and refinement to formulate concise conditions that are directly checkable on the candidate implementations, and use them to show that both problems are decidable in polynomial time. 

\end{abstract}
\section{Introduction}
\label{sec:intro}

Multiparty session types (MSTs)~\cite{DBLP:conf/popl/HondaYC08} are a type-based approach to verifying communication protocols. In MST frameworks, a communication protocol is expressed as a \emph{global type}, which describes the interactions of all protocol participants from a birds-eye view. 
The key property of interest in MST frameworks is \emph{implementability}, which asks whether there exists a collection of local implementations, one per protocol participant, that is deadlock-free and produces the same set of behaviors described by the global type.
The latter property is known as \emph{protocol fidelity}. 
Given an implementable global type, the \emph{synthesis} problem asks to compute such a collection. 
To solve implementability and synthesis, MST frameworks are often equipped with a \emph{projection operator}, which is a partial map from global types to a collection of local implementations.
Projection operators compute a correct implementation for a given global type if one~exists.

However, projection operators only compute one candidate out of many possible implementations for a given global type, which narrows the usability of MST frameworks. As we demonstrate below, substituting this candidate can in some cases achieve an exponential reduction in the size of the local implementation. 
Furthermore, applications may sometimes require that an implementation produce only a subset of the global type's specified behaviors.
We refer to this property as \emph{subprotocol fidelity}. 
For example, a general client-server protocol may customize the set of requests it handles to the specific devices it runs on. Subtyping reintroduces this flexibility into MST frameworks, by characterizing when an implementation can replace another while preserving desirable correctness guarantees.

Formally, a subtyping relation is a reflexive and transitive relation that respects Liskov and Wing's substitution principle~\cite{DBLP:journals/toplas/LiskovW94}: $T'$ is a subtype of $T$ when $T'$ can be \emph{safely} used in any context that expects a term of type $T$.
While implementability for MSTs was originally defined on syntactic local types\cite{DBLP:conf/concur/Honda93,DBLP:conf/popl/HondaYC08}, other implementation models have since been investigated, including communicating session automata~\cite{DBLP:conf/esop/DenielouY12} and behavioral contracts~\cite{DBLP:journals/toplas/CastagnaGP09}. 
We motivate our work with the observation that a subtyping relation is only as powerful as its notion of safety, and the expressivity of its underlying implementation model. 
Existing subtyping relations adopt a notion of safety that is agnostic to a global specification.
For example, \cite{DBLP:journal/mscs/BarbaneraD15,DBLP:journals/mscs/BernardiH16} define safety as the successful completion of a single role in binary sessions, \cite{DBLP:conf/cav/LangeY19} defines safety as eventual reception and progress of all roles in multiparty sessions, and \cite{DBLP:journals/jlamp/GhilezanJPSY19} defines safety as the termination of all roles in multiparty sessions.
As a result, these subtyping relations eagerly reject subtypes that are viable for the specific global type at hand. 
In addition, existing implementation models are restricted to local types with \emph{directed choice} for branching, or equivalent representations thereof~\cite{DBLP:journals/ssm/BravettiZ21}, which prohibit a role from sending messages to or receiving messages from different participants in a choice.  This restrictiveness undermines the flexibility that subtyping is fundamentally designed to provide. 

We present a subtyping relation that extends prior work along both dimensions. 
We define a stronger notion of safety with respect to a given global type: a substitution is safe if in all \emph{well-behaved} contexts, the resulting implementation satisfies both deadlock freedom and subprotocol fidelity. 
We assume an implementation model of unrestricted communicating state machines (CSMs)~\cite{DBLP:journals/jacm/BrandZ83} communicating via FIFO channels, which subsumes implementation models in prior work~\cite{DBLP:conf/cav/LangeY19, DBLP:journals/jlamp/GhilezanJPSY19, DBLP:conf/ppopp/CutnerYV22}. 
We demonstrate that this generalization renders existing subtyping relations which are precise for a restrictive implementation model incomplete.
As a result of both extensions, our subtyping relation requires reasoning about available messages~\cite{DBLP:conf/concur/MajumdarMSZ21} for completeness, a novel feature that is absent from existing subtyping relations.

Our result applies to global types with \emph{sender-driven} choice, which generalize global types from their original formulation with directed choice~\cite{DBLP:conf/popl/HondaYC08}, and borrows insights from recent work on a sound and complete projection operator for this class of global types~\cite{DBLP:conf/cav/LiSWZ23}.

\myparagraph{Contributions.}
\newcommand{\ProblemOne}{\emph{Protocol Verification}\xspace}
\newcommand{\ProblemTwo}{\emph{Protocol Refinement}\xspace}
In this paper, we present the first precise subtyping relation that guarantees deadlock freedom and subprotocol fidelity with respect to a global type, and that assumes an unrestricted, asynchronous CSM implementation model. 
We solve the \ProblemOne problem and the \ProblemTwo problem with respect to global type $\GG$ and a set of roles $\Procs$: 
\begin{enumerate}
	\item \ProblemOne: Given a CSM $\mathcal{A}$, does $\mathcal{A}$ implement $\GG$? 
\item \ProblemTwo: Let $\procA$ be a role and let $B$ be a safe implementation for $\procA$ in any well-behaved context for $\GG$. Given $A$, can $A$ safely replace $B$ in any well-behaved context for $\GG$? 
\end{enumerate}
We exploit the connection between MST subtyping and CSM refinement to formulate concise conditions that are directly checkable on candidate state machines. Using this characterization, we show that both problems are decidable in polynomial time.

 \section{Motivation}
\label{sec:motivation}
We first showcase that sound and complete projection operators can yield local implementations that are exponential in the size of its global type, but can be reduced to constant size by subtyping. 
We then demonstrate the restrictiveness of existing subtyping relations both in terms of their notion of safety and their implementation model.

\myparagraph{Subset projection with exponentially many states.}
We first construct a family of implementable global types $\GG_n$ for $n \in \Nat$ such that $\GG_n$ has size linear in $n$ and the deterministic finite state machine for $\procB$ that recognizes
the projection of the global language onto $\procB$'s alphabet $\Alphabet_\procB$, denoted  $\lang(\GG_n)\wproj_{\Alphabet_\procB}$, has size exponential in $n$.

The construction of the $\GG_n$'s builds on the regular expression $(a^* (ab^*)^n a)^*$, which can only be recognized by a deterministic finite state machine that grows exponentially with~$n$ \cite[Thm.\,11]{DBLP:journals/jalc/EllulKSW05}.

First, we construct the part for $(ab^*)^{i} a$ recursively.
In global types, $\msgFromTo{\procA}{\procB}{\val}$ denotes role $\procA$ sending a message $\val$ to role $\procB$, $+$ denotes choice, $\mu t$ binds a recursion variable $t$ that can be used in the continuation, and $0$ denotes termination.
\[\small
	G_{i} \is
    \msgFromTo{\procA}{\procB}{a}. \,
	μ t_{3,i}. +
	\begin{cases}
		\msgFromTo{\procA}{\procC}{m_{3}}. \,
		\msgFromTo{\procA}{\procB}{b}. \, t_{3,i} \, \\
		\msgFromTo{\procA}{\procC}{n_{3}}. \,
		G_{i-1}\end{cases}
	\;
	\text{for } i > 0
\quad \text{ and } \quad
  G_0 \is
    \msgFromTo{\procA}{\procB}{a}. \,
    t_1
\]
Here, each $G_{i}$ for $i > 0$ generates $(ab^*)$ and $G_0$ adds the last $a$.
Role $\procA$'s choice to send either $m_3$ or $n_3$ to $\procC$ respectively encodes the choice to continue iterating $b$'s or to stop in $b^*$; $\procB$ however, is not involved in this exchange and thus $\procB$'s local language is isomorphic to $(ab^*)^{i}a$. 

Next, we define some scaffolding $G(\hole)$ for the outermost Kleene Star and the first~$a^*$:
\[
  \small
	G(G') \is
	μ t_1. \,
	+
\begin{cases}
		\msgFromTo{\procA}{\procC}{m_1}. \, μ t_2.
		+
		\begin{cases}
			\msgFromTo{\procA}{\procC}{m_2}. \,
			\msgFromTo{\procA}{\procB}{a}. \, t_2 \\
			\msgFromTo{\procA}{\procC}{n_2}. \,
G'
		\end{cases}
		\\
		\msgFromTo{\procA}{\procC}{n_1}. \, 0
	\end{cases}
	\enspace. 
\]
We combine both to obtain the family
$\GG_n \is G(G_n)$.

As $\GG_n$ is implementable, the subset projection~\cite{DBLP:conf/cav/LiSWZ23} for each role is defined. 
One feature of the implementations computed by this projection operator is \emph{local language preservation}, meaning that the language recognized by the local implementation is precisely the projection of the global language onto its alphabet, e.g.\
$\lang(\GG_n)\wproj_{\Alphabet_\procB}$
for role~$\procB$ with alphabet $\Alphabet_\procB$.
In this case,
because $\lang(\GG_n)\wproj_{\Alphabet_\procB}$ can only be recognized by a deterministic finite state machine with size exponential in $n$, the corresponding local language preserving implementation also has size exponential in $n$.

However, not all implementations need to satisfy local language preservation. 
Consider the type $\mu t. (\msgFromTo{\procA}{\procB}{\msgO}. \, t + \msgFromTo{\procA}{\procB}{\msgB}.\, 0)$. The projection of the global language onto~$\procB$ limits $\procB$ to only receiving a sequence of $\msgO$ messages terminated by a $\msgB$ message.
However, an implementation for $\procB$ can rely on $\procA$ to send correct sequences of messages, and instead accept any message that it receives. 
A similar pattern arises in the family $\GG_n$, where the exponentially-sized implementation for role $\procB$ can simply be substituted with an automaton that allows to receive any message from $\procA$. 

\myparagraph{The restrictiveness of existing MST subtyping relations.} 
Consider the two implementations for role $\procA$, represented as finite state machines $A$ and $B$ in \cref{fig:motivation-A,fig:motivation-A'}. 
State machine $A$ embodies the idea of input covariance~\cite{DBLP:journals/acta/GayH05} by adding receive actions, namely \fbox{$\rcv{\procB}{\procA}{\val}$}, which denotes role $\procA$ receiving a message $\val$ from role $\procB$. 
But is it the case that $A$ is a subtype of $B$? 
A preliminary answer based on prior work~\cite{DBLP:conf/tacas/LangeY16, DBLP:journals/jlamp/GhilezanJPSY19} is \emph{no}, for the reason that $A$ falls outside of the implementation models considered in these works: the initial state in $A$ contains outgoing receive transitions from two distinct senders, $\procB$ and $\procC$, and one of the final states contains an outgoing transition. Thus, there exists no local type representation of $A$. 

\begin{figure}[t]
\hfill
\begin{subfigure}[t]{.4\textwidth}
\centering
\resizebox{.9\textwidth}{!}{
	\begin{tikzpicture}[node distance=1cm and 2cm,>=stealth', line width=0.25mm]
		\node[draw, circle, minimum size=0.7cm, initial left=, initial text =](t0){};
		\node[draw, circle, accepting, minimum size = 0.7cm, right=2cm of t0](t1){};	
		\node[draw, circle, accepting, minimum size = 0.7cm, right=2cm of t1](t2){};	
		\path[->](t0) edge node[above] {$\rcv{\procB}{\procA}{\val}$} (t1);
		\path[->](t1) edge node[above] {$\rcv{\procC}{\procA}{\val}$} (t2);
		\path[->](t0) edge node[below] {\fbox{$\rcv{\procC}{\procA}{\val}$}} (t1);
	\end{tikzpicture}
	}
	\caption{$A$}
	\label{fig:motivation-A'}
\end{subfigure}
\hfill
\begin{subfigure}[t]{0.4\textwidth}
\centering
\resizebox{.9\textwidth}{!}{
	\begin{tikzpicture}[node distance=1cm and 2cm,>=stealth', line width=0.25mm]
		\node[draw, circle, minimum size=0.7cm, initial left=, initial text =](t0){};
		\node[draw, circle, accepting, minimum size = 0.7cm, right=2cm of t0](t1){};	
		\node[draw, circle, accepting, minimum size = 0.7cm, right=2cm of t1](t2){};	
		\path[->](t0) edge node[above] {$\rcv{\procB}{\procA}{\val}$} (t1);
		\path[->](t1) edge node[above] {$\rcv{\procC}{\procA}{\val}$} (t2);
		\path[->](t0) edge node[below] {\phantom{\fbox{$\rcv{\procC}{\procA}{\val}$}}} (t1);
	\end{tikzpicture}
}
	\caption{$B$}
	\label{fig:motivation-A}
\end{subfigure}
\hspace{.7cm}
\caption{Two state machines for role $\procB$}
\end{figure}
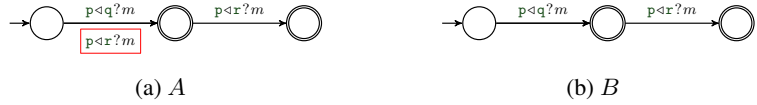

As a first step, let us generalize the implementation model to machines with arbitrary finite state control, and revisit the question. It turns out that the answer now depends on what protocol role $\procA$, alongside the other roles in the context, is following. 
Consider the two global types
\[\small
	\GG_1 \is 
	\msgFromTo{\procB}{\procA}{\val}. \,
	\msgFromTo{\procC}{\procA}{\val}. \,
	0 \qquad\text{and}\qquad
	\GG_2 \is 
	\msgFromTo{\procB}{\procA}{\val}. \,
	0 \,\enspace.
\]
We observe that $A$ is a subtype of $B$ under the context of $\GG_2$, but not under the context of $\GG_1$. 
Suppose that roles $\procB$ and $\procC$ are both following $\GG_1$, and thus both roles send a message $\val$ to $\procA$. Under asynchrony, the two messages can arrive in $\procA$'s channel in any order; this holds even in a synchronous setting. 
Therefore, there exists an execution trace in which $\procA$ takes the transition labeled \fbox{$\rcv{\procC}{\procA}{\val}$} in $A$ and first receives from $\procC$. Role $\procA$ then finds itself in a final state with a pending message from $\procB$ that it is unable to receive, thus causing a deadlock in the CSM. 
On the other hand, if $\procB$ were following $\GG_2$, the addition of the receive transition \fbox{$\rcv{\procC}{\procA}{\val}$} is safe because it is never enabled, and thus $A$ can safely compose with any context following $\GG_2$ without violating protocol fidelity and deadlock freedom.

 \section{Preliminaries}
\label{sec:prelim}

We restate relevant definitions from~\cite{DBLP:conf/cav/LiSWZ23}.

\paragraph{Words.}
Let $\Alphabet$ be a finite alphabet.
$\Alphabet^*$ denotes the set of finite words over $\Alphabet$, $\Alphabet^\omega$ the set of infinite words, and $\Alphabet^\infty\negthinspace$ their union $\Alphabet^* \cup \Alphabet^\omega$.
A word $u \in \Alphabet^*$ is a \emph{prefix} of word $v \in \Alphabet^\infty$, denoted $u \leq v$, if there exists $w \in \Alphabet^\infty$ with $u \cdot w = v$.

\paragraph{Message Alphabet.}
Let $\Procs$ be a set of roles and $\MsgVals$ be a set of messages. We define the set of \emph{synchronous events} $\AlphSync \is \set{ \msgFromTo{\procA}{\procB}{\val} \mid \procA,\procB ∈ \Procs \text{ and } \val ∈ \MsgVals}$ where $\msgFromTo{\procA}{\procB}{\val}$ denotes that message $\val$ is sent by $\procA$ to $\procB$ atomically.
This is split for \emph{asynchronous events}.
For a role $\procA \in\Procs$, we define the alphabet
    $\Alphabet_{\procA,!} = \set{\snd{\procA}{\procB}{\val} \mid \procB \in \Procs,\; \val \in \MsgVals }$ of \emph{send} events
and the alphabet
    $\Alphabet_{\procA,?} = \set{\rcv{\procB}{\procA}{\val} \mid \procB \in \Procs,\; \val \in \MsgVals }$ of \emph{receive} events.
The event $\snd{\procA}{\procB}{\val}$ denotes role $\procA$ sending a message $\val$ to $\procB$,
and $\rcv{\procB}{\procA}{\val}$ denotes role~$\procA$ receiving a message $\val$ from $\procB$.
We write $\Alphabet_{\procA} = \Alphabet_{\procA,!} \union \Alphabet_{\procA,?}$,
$\Alphabet_! = \Union_{\procA \in \Procs} \Alphabet_{\procA,!}$, and
$\Alphabet_? = \Union_{\procA \in \Procs} \Alphabet_{\procA,?}$.
Finally, $\AlphAsync = \Alphabet_! \cup \Alphabet_?$.
We say that $\procA$ is \emph{active} in $x \in \AlphAsync$ if $x \in \Alphabet_{\procA}$.
For each role $\procA\in \Procs$, we define a homomorphism~$\wproj_{\Alphabet_\procA}$, where $x \wproj_{\Alphabet_\procA} = x$ if $x \in \Alphabet_\procA$ and $\emptystring$ otherwise.
We fix $\Procs$ and~$\MsgVals$ in the rest of the paper.

\paragraph{Global Types -- Syntax.}
Global types for MSTs \cite{DBLP:conf/concur/MajumdarMSZ21} are defined by the grammar:
    \begin{grammar}
     G \is
       0
     | \sum_{i ∈ I} \msgFromTo{\procA}{\procB_{i}}{\val_i.G_i}
     | \mu t. \; G
     | t
    \end{grammar}\\[-3ex]
where $\procA, \procB_i$ range over $\Procs$, $\val_i$ over $\MsgVals$, and $t$ over a set of recursion variables. 

We require each branch of a choice to be distinct: 
$∀ i,j ∈ I.\, i≠j ⇒ (\procB_{i},\val_i) ≠ (\procB_{j},\val_j)$, 
the sender and receiver of an event to be distinct: 
$\procA ≠ \procB_i$ for each $i \in I$, 
and recursion to be guarded:
in $μ t. \, G$, there is at least one message between $μt$ and each $t$ in $G$.
We omit $\sum$ for singleton choices. 
When working with a protocol described by a global type, we use~$\GG$ to refer to the top-level type, and $G$ to refer to its subterms.

We use the extended definition of global types from \cite{DBLP:conf/concur/MajumdarMSZ21} featuring \emph{sender-driven choice}. 
This definition subsumes classical MSTs that only allow \emph{directed choice} \cite{DBLP:conf/popl/HondaYC08}.
We focus on communication primitives and omit features like delegation or parametrization, and refer the reader to \cref{sec:related} for a discussion of different MST frameworks.

\paragraph{Global Types -- Semantics.}
As a basis for the semantics of a global type $\GG$, we construct a finite state machine
$
\semglobalsync(\GG) = (Q_{\GG}, \AlphSync, δ_{\GG}, q_{0, \GG}, F_{\GG})
$
where
\begin{itemize}
	\item $Q_{\GG}$ is the set of all syntactic subterms in $\GG$ together with the term $0$,
	\item $δ_{\GG}$ consists of the transitions $(\sum_{i ∈ I} \msgFromTo{\procA}{\procB_i}{\val_i.G_i}, \msgFromTo{\procA}{\procB_i}{\val_i}, G_i)$ for each $i ∈ I$,
	as well as $(μ t. G', ε, G')$ and $(t, ε, μ t. G')$ for each subterm~$\mu t.G'$, \item $q_{0, \GG} = \GG$ and
	$F_{\GG} = \set{0}$.
\end{itemize}
We define a homomorphism 
$\SyncToAsync$ onto the asynchronous alphabet:
\[
\SyncToAsync(\msgFromTo{\procA}{\procB}{\val})
\is
\snd{\procA}{\procB}{\val}. \,
\rcv{\procA}{\procB}{\val}\enspace.
\]
The semantics $\lang(\GG)$ of a global type $\GG$ is given by
$\interswaplang(\SyncToAsync(\lang(\semglobalsync(\GG))))$ where $\interswaplang$ is the closure under the indistinguishability relation $\interswap$ \cite{DBLP:conf/concur/MajumdarMSZ21}.
Two events are independent if they are not related by the \emph{happened-before} relation \cite{DBLP:journals/cacm/Lamport78}. For instance, any two send events from distinct senders are independent.
Two words are indistinguishable if one can be reordered into the other by repeatedly swapping consecutive independent events. The full definition can be found in \appendixRef{app:indist-rel}.

We call a state $q_G \in Q_{\GG}$ a \emph{send-originating} state, denoted $q_G \in Q_{\GG,!}$ for role $\procA$ if there exists a transition $q_G \xrightarrow{\msgFromTo{\procA}{\procB}{\val}} q_{G'} \in \delta_\GG$, and a \emph{receive-originating} state, denoted $q_G \in Q_{\GG,?}$ for $\procA$ if there exists a transition $q_G \xrightarrow{\msgFromTo{\procB}{\procA}{\val}} q_{G'} \in \delta_\GG$. 
We omit mention of role $\procA$ when clear from context.

\paragraph{Communicating State Machine \cite{DBLP:journals/jacm/BrandZ83}.}
\label{def:csm-formalisation}
$\mathcal{A} = \CSM{A}$ is a CSM over $\Procs$ and~$\MsgVals$ if
${A}_\procA = (Q_\procA, \Alphabet_\procA, \delta_\procA, q_{0, \procA}, F_\procA)$
is a deterministic finite state machine
over~$\Alphabet_\procA$ for every $\procA\in\Procs$.
Let 
$\prod_{\procA \in \Procs} Q_\procA$ 
denote the set of global states and
\mbox{$\channels = \set{(\channel{\procA}{\procB}) \mid \procA,\procB\in \Procs, \procA\neq \procB}$}
denote the set of channels. 
A~\emph{configuration} of $\mathcal{A}$ is a pair $(\vec{s}, \xi)$, where $\vec{s}\,$ is a global state and
$\xi : \channels \rightarrow \MsgVals^*$
is a mapping from each channel to a sequence of messages.
We use $\vec{s}_\procA$ to denote the state of $\procA$ in $\vec{s}$.
The CSM transition relation, denoted $\rightarrow$, is defined as~follows.
\begin{itemize}
	\item
	$(\vec{s},\xi) \xrightarrow{\snd{\procA}{\procB}{\val}} (\pvec{s}',\xi')$ if
	$(\vec{s}_\procA, \snd{\procA}{\procB}{\val}, \pvec{s}'_\procA)\in\delta_\procA$,
	$\vec{s}_\procC = \pvec{s}'_\procC$ for every role $\procC \neq \procA$,
	$\xi'(\channel{\procA}{\procB}) =  \xi(\channel{\procA}{\procB})\cdot\val$ and $\xi'(c) = \xi(c)$ for every other channel $c\in \channels$.
	
	\item
	$(\vec{s},\xi) \xrightarrow{\rcv{\procA}{\procB}{\val}} (\pvec{s}',\xi')$ if
	$(\vec{s}_\procB, \rcv{\procA}{\procB}{\val}, \pvec{s}'_\procB)\in\delta_\procB$,
	$\vec{s}_\procC = \pvec{s}'_\procC$ for every role $\procC \neq \procB$,
	$\xi(\channel{\procA}{\procB}) = \val\cdot \xi'(\channel{\procA}{\procB})$
	and $\xi'(c) = \xi(c)$ for every other channel $c\in \channels$.
\end{itemize}
In the initial configuration $(\vec{s}_0, \xi_0)$, each role's state in $\vec{s}_0$ is the initial state $q_{0,\procA}$ of $A_\procA$, and $\xi_0$ maps each channel to $\emptystring$.
A configuration $(\vec{s}, \xi)$ is said to be \emph{final} iff $\vec{s}_\procA$ is final for every $\procA$ and $\xi$ maps each channel to~$\emptystring$.
Runs and traces are defined in the expected way. 
A run is \emph{maximal} if either it is finite and ends in a final configuration, or it is infinite. 
The language $\lang(\mathcal{A})$ of the CSM $\mathcal{A}$ is defined as the set of maximal traces.
A~configuration $(\vec{s}, \xi)$ is a \emph{deadlock} if it is not final and has no outgoing transitions.
A~CSM is \emph{deadlock-free} if no reachable configuration is a deadlock.

\begin{definition}[Implementability]
\label{def:implementability}
We say that a CSM $\CSM{A}$ implements a global type $\GG$ if the following two properties hold: 
\begin{inparaenum}[(i)]
	\item \label{def:implementability-protocol-fidelity}
	\emph{protocol fidelity:} $\lang(\CSM{A}) = \lang(\GG)$, and
	\item \label{def:implementability-deadlock-freedom}
	\emph{deadlock freedom:} $\CSM{A}$ is deadlock-free.
\end{inparaenum}
A global type $\GG$ is \emph{implementable} if there exists a CSM that implements it. 
\end{definition}

One candidate implementation for global types can be computed directly from $\semglobal(\GG)$, by removing actions unrelated to each role and determinizing the result. 
The following two definitions define this candidate implementation in two steps. 
\begin{definition}[Projection by Erasure~\cite{DBLP:conf/cav/LiSWZ23}]
	\label{def:projection-by-erasure}
	Let $\GG$ be some global type
	with its state machine
	$
	\semglobalsync(\GG) =
	(Q_{\GG},
	\AlphSync,
	\delta_{\GG},
	q_{0, \GG},
	F_{\GG})
	$.
	For each role \mbox{$\procA \in \Procs$}, we define the state machine
	$
	\projerasure{\GG}{\procA} \,=
	(Q_{\GG},
	\Alphabet_\procA \dunion \set{\emptystring},
	\projerasuretrans,
	q_{0, \GG},
	F_{\GG})
	$
	where
	$\projerasuretrans \is
	\set{q \xrightarrow{\SyncToAsync(a) \wproj_{\Alphabet_\procA}} q'
		\mid q \xrightarrow{a} q' \in \delta_{\GG}}$. 
	By definition of $\SyncToAsync(\hole)$, it holds that $\SyncToAsync(a) \wproj_{\Alphabet_\procA} \in \Alphabet_\procA \dunion \set{\emptystring}$.
\end{definition}

We determinize $\projerasure{\GG}{\procA}$ via a standard subset construction~\cite{DBLP:books/daglib/0086373} to obtain a deterministic local state machine for $\procA$. 
Note that the construction ensures that $Q_\procA$ only contains subsets of $Q_{\GG}$ whose states are reachable via the same traces.
\begin{definition}[Subset Construction~\cite{DBLP:conf/cav/LiSWZ23}]
	\label{def:subset-construction}
	Let $\GG$ be a global type and $\procA$ be a role. 
	Then, the \emph{subset construction} for $\procA$ is defined as
	\[
	\subsetcons{\GG}{\procA} =
	\bigl(
	Q_{\procA},
	\Alphabet_\procA,
	\delta_{\procA},
	s_{0, \procA},
	F_{\procA}
	\bigr)
	\text{ where }
	\]
	\begin{itemize}
\item $ \delta(s, a) \is
		\set{q' \in Q_{\GG}
			\mid
			\exists q \in s,
			q \xrightarrow{a} \xrightarrow{\emptystring}\mathrel{\vphantom{\to}^*} q' \in \projerasuretrans
		},
		$
		for
		every
		$s \subseteq Q_{\GG}$ and
		$a \in \Alphabet_{\procA}$,\item $s_{0, \procA} \is
		\set{q \in Q_{\GG} \mid
			q_{0, \GG} \xrightarrow{\emptystring} \mathrel{\vphantom{\to}^*} q \in \projerasuretrans}$,
		\item $Q_{\procA} \is \lfp_{\set{s_{0,\procA}}}^\subseteq \lambda Q.\, Q \union \set{ \delta(s,a) \mid s \in Q \land a \in \Alphabet_{\procA}} \setminus \set{\emptyset}$,
		
\item $\delta_{\procA} \is \restrict{\delta}{Q_{\procA} \times \Alphabet_{\procA}}$, and
		\item $F_{\procA} \is
		\set{s \in Q_{\procA}
			\mid s \inters F_{\GG} \neq \emptyset}$.
	\end{itemize}
\end{definition}

Li et al.~\cite{DBLP:conf/cav/LiSWZ23} showed that if $\GG$ is implementable, then $\CSMl{\subsetcons{\GG}{\procA}}$ implements~$\GG$ and satisfies the following property:

\begin{definition}
	Let $\GG$ be a global type.
	We call an implementation $\CSM{A}$ \emph{local language preserving} with respect to $\GG$ if $\lang(A_\procA) = \lang(\GG)\wproj_{\Alphabet_\procA}$ for all $\procA \in \Procs$.
\end{definition}

For the remainder of the paper, we fix a global type $\GG$ that we assume is implementable.

 \section{Deciding \emph{\ProblemOne}}
\label{sec:first-problem}

\ProblemOne asks: Given a CSM $\mathcal{A}$, does $\mathcal{A}$ implement $\GG$? 
For two CSMs $\mathcal{A}$ and $\mathcal{B}$, we say that $\mathcal{A}$ refines $\mathcal{B}$ if and only if every trace in $\mathcal{A}$ is a trace in $\mathcal{B}$, and a trace in $\mathcal{A}$ terminates maximally in $\mathcal{A}$ if and only if it terminates maximally in $\mathcal{B}$. If $\mathcal{A}$ and $\mathcal{B}$ refine each other, we say that they are equivalent.
Further, in the case that $\mathcal{B}$ is deadlock-free, one can simplify the condition to the following: every trace in $\mathcal{A}$ is a trace in $\mathcal{B}$, and if a trace terminates in $\mathcal{A}$, then it terminates in $\mathcal{B}$ and is maximal in $\mathcal{A}$. 

We can recast \ProblemOne in terms of CSM refinement using the fact that $\subsetprojCSM$ is an implementation for $\GG$. 
Therefore, the question amounts to asking whether $\mathcal{A}$ and $\subsetprojCSM$ are equivalent. 

\newcommand{\CharacterizationOne}{\emph{$C_1$}\xspace}
\newcommand{\CharacterizationTwo}{\emph{$C_2$}\xspace}

Our goal is then to present a characterization \CharacterizationOne that satisfies the following: 
\begin{theorem}
	\label{thm:equivalence-one}
	Let $\GG$ be an implementable global type and $\mathcal{A}$ be a CSM. 
	Then, $\subsetprojCSM$ and $\mathcal{A}$ are equivalent if and only if \CharacterizationOne is satisfied. 
\end{theorem}

We motivate our characterization for \ProblemOne using a series of examples. 
Consider the following simple global type $\GG_1$:
\[
\small
	\GG_1 \is 
	+ \;
	\begin{cases}
		\msgFromTo{\procA}{\procB}{\msgB}. \,
		\msgFromTo{\procB}{\procA}{\msgB}. \,
		0 \,
		\\
		\msgFromTo{\procA}{\procB}{\msgM}. \,
		\msgFromTo{\procB}{\procA}{\msgM}. \,
		0 \,
	\end{cases}
\]
This global type is trivially implementable; the subset construction for role $\procB$ obtained by the projection operator in \cite{DBLP:conf/cav/LiSWZ23} is depicted in \cref{fig:G1-subset-projection-p}. Clearly, in any CSM implementing $\GG_1$, the subset construction can be replaced with the more compact state machine~$A_1$, shown in \cref{fig:G1-alternative-1}.

\begin{figure}[t]
\hfill
\begin{subfigure}[b]{.3\textwidth}
\centering
\resizebox{0.99\textwidth}{!}{
	\begin{tikzpicture}[node distance=1cm and 2cm,>=stealth', line width=0.25mm]
		\node[draw, circle, minimum size=0.5cm, initial left=, initial text =](q0){};
		\node[draw, circle, minimum size = 0.5cm, below right=0.4cm and 1.5cm of q0](q1){};
		\node[draw, circle, minimum size = 0.5cm, above right=0.4cm and 1.5cm of q0](q2){};
		\node[draw, circle, minimum size = 0.5cm, accepting, right=1.5cm of q1](q3){};
		\node[draw, circle, minimum size = 0.5cm, accepting, right= 1.5cm of q2](q4){};
		\path[->](q0) edge node[sloped, pos=0.5, below] {$\snd{\procA}{\procB}{\msgM}$} (q1);
		\path[->](q0) edge node[sloped, pos=0.5, above] {$\snd{\procA}{\procB}{\msgB}$} (q2);
		\path[->](q1) edge node[below] {$\rcv{\procB}{\procA}{\msgM}$} (q3);
		\path[->](q2) edge node[above] {$\rcv{\procB}{\procA}{\msgB}$} (q4);
	\end{tikzpicture}
}
	\caption{$\subsetproj{\GG_1}{\procA}$	\label{fig:G1-subset-projection-p}}
\end{subfigure}\hfill \begin{minipage}[b]{.3\textwidth}
\begin{subfigure}[b]{\textwidth}
\centering
\resizebox{0.9\textwidth}{!}{
	\begin{tikzpicture}[node distance=1cm and 2cm,>=stealth', line width=0.25mm]
		\node[draw, circle, minimum size=0.5cm, initial left=, initial text =](q0){};
		\node[draw, circle, minimum size = 0.5cm, right=1.5cm of q0](q1){};
		\node[draw, circle, minimum size = 0.5cm, accepting, right=1.5cm of q1](qf){};
		\path[->](q0) edge node[below] {$\snd{\procA}{\procB}{\msgM}$} (q1);
		\path[->](q0) edge node[above] {$\snd{\procA}{\procB}{\msgB}$} (q1);
		\path[->](q1) edge node[below] {$\rcv{\procB}{\procA}{\msgM}$} (qf);
		\path[->](q1) edge node[above] {$\rcv{\procB}{\procA}{\msgB}$} (qf);
\end{tikzpicture}
}
	\caption{$A_1$ 	\label{fig:G1-alternative-1}
}
\end{subfigure}\\\\
\begin{subfigure}[b]{\textwidth}
\centering
\resizebox{0.9\textwidth}{!}{
	\begin{tikzpicture}[node distance=1cm and 2cm,>=stealth', line width=0.25mm]
		\node[draw, circle, minimum size=0.5cm, initial left=, initial text =](q0){};
		\node[draw, circle, minimum size = 0.5cm, accepting, right=1.5cm of q0](q1){};
\path[->](q0) edge node[below] {$\snd{\procA}{\procB}{\msgM}$} (q1);
		\path[->](q0) edge node[above] {$\snd{\procA}{\procB}{\msgB}$} (q1);
		\path[->](q1) edge [loop right] node {$\rcv{\_}{\procA}{\_}$} (q1);
	\end{tikzpicture}
}
	\caption{$A_2$ 	\label{fig:G1-alternative-2}}
      \end{subfigure}
\end{minipage}
\hfill
\begin{subfigure}[b]{0.3\textwidth}
	\centering
	\resizebox{.9\textwidth}{!}{
		\begin{tikzpicture}[node distance=1cm and 2cm,>=stealth', line width=0.25mm]
			\node[draw, circle, minimum size=0.5cm, initial left=, initial text =](q0){};
			\node[draw, circle, minimum size=0.5cm, above right=1cm and 1.5cm of q0](qfake){};
			\node[draw, circle, minimum size=0.5cm, right=1.5cm of qfake](qfake2){};
			\node[draw, circle, minimum size = 0.5cm, accepting, right=1.5cm of q0](q1){};
\path[->](q0) edge node[above] {$\snd{\procA}{\procB}{\msgB}$} (q1);			
			\path[->](q0) edge node[below] {$\snd{\procA}{\procB}{\msgM}$} (q1);
			\path[->](q0) edge node[sloped, pos=0.5, above] {$\rcv{\procB}{\procA}{\msgO}$} (qfake);
			\path[->](qfake) edge node[above] {$\snd{\procA}{\procB}{\msgO}$} (qfake2);
			\path[->](q1) edge [loop right] node {$\rcv{\_}{\procA}{\_}$} (q1);
\end{tikzpicture}
	}
	\caption{$A_3$ 	\label{fig:G1-alternative-3}} \end{subfigure}
	\caption{Subset construction of $\GG_1$ onto $\procA$ and three alternative implementations} \end{figure}
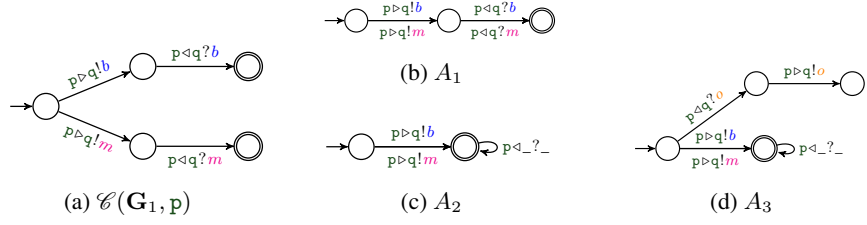

For a local state machine in a CSM, control flow is determined by both the local transition relation and the global channel state. 
However, in some cases, the local information is redundant: the role's channel contents alone are enough to enforce that it produces the correct behaviors. 
In the example above, after $\procA$ chooses to send $\procB$ either $\msgM$ or $\msgB$, $\procB$ will guarantee that the correct message, i.e. the same one, is sent back to~$\procA$.
Role $\procA$'s state machine can rely on its channel contents to follow the protocol -- it does not need separate control states for each message. 
In fact, we can further replace $\procA$'s control states after sending with an accepting universal receive state, as shown in $A_2$ in \cref{fig:G1-alternative-2}.
Finally, we can add send transitions from unreachable states, as shown in $A_3$ in \cref{fig:G1-alternative-3}. 

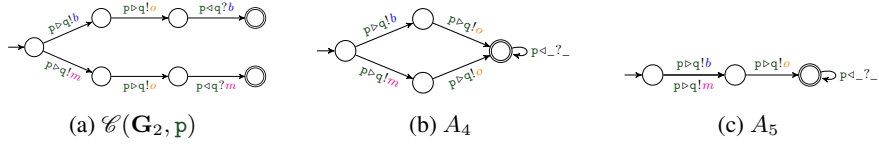
\begin{figure}[t]
\hfill
\begin{subfigure}[b]{.3\textwidth}
\centering
\resizebox{0.99\textwidth}{!}{
	\centering
	\begin{tikzpicture}[node distance=1cm and 2cm,>=stealth', line width=0.25mm]
		\node[draw, circle, minimum size=0.5cm, initial left=, initial text =](q0){};
		\node[draw, circle, minimum size = 0.5cm, below right=0.4cm and 1.4cm of q0](q1){};
		\node[draw, circle, minimum size = 0.5cm, above right=0.4cm and 1.4cm of q0](q2){};
		\node[draw, circle, minimum size = 0.5cm, right=1.5cm of q1](q3){};
		\node[draw, circle, minimum size = 0.5cm, right=1.5cm of q2](q4){};
		\node[draw, circle, minimum size = 0.5cm, accepting, right=1.5cm of q3](q5){};
		\node[draw, circle, minimum size = 0.5cm, accepting, right= 1.5cm of q4](q6){};
		\path[->](q0) edge node[sloped, pos=0.5, below] {$\snd{\procA}{\procB}{\msgM}$} (q1);
		\path[->](q0) edge node[sloped, pos=0.5, above] {$\snd{\procA}{\procB}{\msgB}$} (q2);
		\path[->](q1) edge node[below] {$\snd{\procA}{\procB}{\msgO}$} (q3);
		\path[->](q2) edge node[above] {$\snd{\procA}{\procB}{\msgO}$} (q4);
		\path[->](q3) edge node[below] {$\rcv{\procB}{\procA}{\msgM}$} (q5);
		\path[->](q4) edge node[above] {$\rcv{\procB}{\procA}{\msgB}$} (q6);
\end{tikzpicture}
}
	\caption{$\subsetproj{\GG_2}{\procA}$}
	\label{fig:G2-subset-projection-p}
\end{subfigure}
\hfill
\begin{subfigure}[b]{.3\textwidth}
\centering
\resizebox{0.99\textwidth}{!}{
\begin{tikzpicture}[node distance=1cm and 2cm,>=stealth', line width=0.25mm]
		\node[draw, circle, minimum size=0.5cm, initial left=, initial text =](q0){};
		\node[draw, circle, minimum size = 0.5cm, below right=0.4cm and 1.5cm of q0](q1){};
		\node[draw, circle, minimum size = 0.5cm, above right=0.4cm and 1.5cm of q0](q2){};
		\node[draw, circle, minimum size = 0.5cm, accepting, below right=0.4cm and 1.5cm of q2](q34){};
\path[->](q0) edge node[sloped, pos=0.5, below] {$\snd{\procA}{\procB}{\msgM}$} (q1);
		\path[->](q0) edge node[sloped, pos=0.5, above] {$\snd{\procA}{\procB}{\msgB}$} (q2);
		\path[->](q1) edge node[sloped, pos=0.5, below] {$\snd{\procA}{\procB}{\msgO}$} (q34);
		\path[->](q2) edge node[sloped, pos=0.5, above] {$\snd{\procA}{\procB}{\msgO}$} (q34);
		\path[->](q34) edge [loop right] node {$\rcv{\_}{\procA}{\_}$} (q34);
\end{tikzpicture}
}
	\caption{$A_4$}
	\label{fig:G2-alternative-1}
\end{subfigure}
\hfill
\begin{subfigure}[b]{.3\textwidth}
\centering
\resizebox{0.99\textwidth}{!}{
    \begin{tikzpicture}[node distance=1cm and 2cm,>=stealth', line width=0.25mm]
		\node[draw, circle, minimum size=0.5cm, initial left=, initial text =](q0){};
		\node[draw, circle, minimum size = 0.5cm, right=1.5cm of q0](q12){};
		\node[draw, circle, accepting, minimum size = 0.5cm, right=1.3cm of q12](q34){};
		\path[->](q0) edge node[below] {$\snd{\procA}{\procB}{\msgM}$} (q12);
		\path[->](q0) edge node[above] {$\snd{\procA}{\procB}{\msgB}$} (q12);
		\path[->](q12) edge node[above] {$\snd{\procA}{\procB}{\msgO}$} (q34);
\path[->](q34) edge [loop right] node {$\rcv{\_}{\procA}{\_}$} (q34);
    \end{tikzpicture}
}
	\caption{$A_5$}
	\label{fig:G2-alternative-2}
\end{subfigure}
	\caption{Subset construction of $\GG_2$ onto $\procA$ and two alternative implementations} \end{figure}

Similar patterns arise for send actions. 
Consider the following variation of the first global type, $\GG_2$:
\[\small
\GG_2 \is 
+ \;
\begin{cases}
	\msgFromTo{\procA}{\procB}{\msgB}. \,
	\msgFromTo{\procA}{\procC}{\msgO}. \,
	\msgFromTo{\procB}{\procA}{\msgB}. \, 
	0 \,
	\\
	\msgFromTo{\procA}{\procB}{\msgM}. \,
	\msgFromTo{\procA}{\procC}{\msgO}. \,
	\msgFromTo{\procB}{\procA}{\msgM}. \, 
	0 \,
\end{cases}
\]
The subset construction from \cite{DBLP:conf/cav/LiSWZ23} yields the state machine for $\procA$ shown in \cref{fig:G2-subset-projection-p}. 

Our reasoning above shows that $A_4$, depicted in~\cref{fig:G2-alternative-1}, is a correct alternative implementation for $\procA$. 
Now observe that the pre-states of the two $\snd{\procA}{\procB}{\msgO}$ transitions can be collapsed because their continuations are identical. This yields another correct alternative implementation $A_5$, shown in \cref{fig:G2-alternative-2}. 

Informally, the subset construction takes a ``maximalist'' approach, creating as many distinct states as possible from the global type, and checking whether they are enough to guarantee that the role behaves correctly. 
However, sometimes this maximalism creates redundancy: just because two states are distinct according to the global type does not mean they need to be. In these cases, an implementation has the flexibility to merge certain distinct states together, or add transitions to a state. 
We wish to precisely characterize when such modifications to local state machines preserve protocol fidelity and deadlock freedom. 

Our conditions for \CharacterizationOne are inspired by the Send and Receive Validity conditions that precisely characterize implementability for global types, given in \cite{DBLP:conf/cav/LiSWZ23}. 
We restate the conditions, in addition to relevant definitions, for clarity. 
\begin{definition}[Available messages~\cite{DBLP:conf/concur/MajumdarMSZ21}]
	The set of available messages  is recursively defined on the structure of the global type.
	For completeness, we need to unfold the distinct recursion variables once.
	For this, we define a map $\getMu$ from variable to subterms and write $\getMuG$ for $\getMu(\GG)$:
	
\begin{footnotesize}
		{} \hfill
		$
		\getMu(0) \is [\,]
		$ \hfill $
		\getMu(t) \is [\,]
		$ \hfill $
		\getMu(μt.G) \is [t \mapsto G] ∪ \getMu(G)
		$  \hfill {}
		
		{} \hfill $
		\getMu(\sum_{i ∈ I} \msgFromTo{\procA}{\procB_i}{\val_i.G_i}) \is \bigcup_{i∈I} \getMu(G_i)
		$  \hfill {}
	\end{footnotesize}

	\noindent The function $\semavaildef{\blockedset}{T}{\hole}$ keeps a set of unfolded variables $T$, which is empty initially.	
	\begin{center}
		\begin{minipage}{0.98\textwidth}
			\begin{small}
				\noindent
				$
				\semavaildef{\blockedset}{T}{0} \is ∅ \hfill
				\semavaildef{\blockedset}{T}{μt.G} \is \semavaildef{\blockedset}{T ∪ \set{t}}{G} \hfill
				\semavaildef{\blockedset}{T}{t} \is \begin{cases}
					∅ & \text{if} ~ t ∈ T\\
					\semavaildef{\blockedset}{T ∪ \set{t}}{\getMuG(t)} & \text{if} ~ t ∉ T
				\end{cases}\\[2mm]
				\semavaildef{\blockedset}{T}{\sum_{i ∈ I} \msgFromTo{\procA}{\procB_i}{\val_i.G_i}}
				\is
				\begin{cases}
					\bigcup_{i∈I,m∈\MsgVals} (\semavaildef{\blockedset}{T}{G_i} \setminus \set{ \snd{\procA}{\procB_i}{\val} }) ∪  \set{ \snd{\procA}{\procB_i}{\val_i} } \quad \hfill \text{if} ~ \procA ∉ \blockedset \\
					\bigcup_{i∈I} \semavaildef{\blockedset ∪ \set{ \procB_i }}{T}{G_i} \quad \hfill \text{if} ~ \procA ∈ \blockedset
				\end{cases}
				$
			\end{small}
		\end{minipage}
	\end{center}\smallskip	
	\noindent We write $\semavail^{\blockedset}_{(G \ldots)}$ for $\semavaildef{\blockedset}{\emptyset}{G}$.
	If $\blockedset$ is a singleton set, we omit set notation and write $\semavail^{\procA}_{(G \ldots)}$
	for $\semavail^{\set{\procA}}_{(G \ldots)}$.
\end{definition}

Intuitively, the available messages definition captures all of the messages that can be at the head of their respective channels when a particular role is blocked from taking further transitions. 

For notational convenience, we define the \emph{origin} and \emph{destination} of a transition following~\cite{DBLP:conf/cav/LiSWZ23}, but generalized from the subset construction automaton. 
\begin{definition}[Transition Origin and Destination]
	Let $\GG$ be a global type and let $\projerasuretrans$ be the transition relation of $\projerasure{\GG}{\procA}$.
For $x \in \Alphabet_\procA$ and $s,s' \subseteq Q_\GG$,
	we define the set of \emph{transition origins} $\transAnnoFunc(s \xrightarrow{x} s')$ and \emph{transition destinations} $\transAnnoFunDest(s \xrightarrow{x} s')$ as~follows:
\begin{align*}
		\transAnnoFunc(s \xrightarrow{x} s')
		\is {} &
		\set{G \in s
			\mid
			\exists G' \in s'. \,
			G \xrightarrow{x}\mathrel{\vphantom{\to}^*} G' \in \projerasuretrans} \; \text{ and }\\
		\transAnnoFunDest(s \xrightarrow{x} s')
		\is {} &
		\set{G' \in s'
			\mid
			\exists G \in s. \,
			G \xrightarrow{x}\mathrel{\vphantom{\to}^*} G' \in \projerasuretrans} \enspace.
	\end{align*}
\end{definition}

Li et al.~\cite{DBLP:conf/cav/LiSWZ23} showed that $\GG$ is implementable if and only if the subset construction CSM $\CSMl{\subsetcons{\GG}{\procA}}$ satisfies Send and Receive Validity for each $\subsetcons{\GG}{\procA}$.

\begin{definition}[Send Validity]
	\label{cond:send-state-validity-transition-origins}
	$\subsetcons{\GG}{\procA}$ satisfies \emph{Send Validity} iff every send transition $s  \xrightarrow{x} s' \in \delta_\procA$ is enabled in all states contained in $s$: 
	\[
	\forall s  \xrightarrow{x} s' \in \delta_\procA.
	~x \in \Alphabet_{\procA,!} \implies
	\transAnnoFunc(s  \xrightarrow{x} s') = s \enspace.
	\]
\end{definition}

\begin{definition}[Receive Validity]
	\label{cond:rcv-state-validity}
	$\subsetcons{\GG}{\procA}$ satisfies \emph{Receive Validity} iff no receive transition is enabled in an alternative continuation that originates from the same source state:
	\[
	\begin{array}{l}
		\forall 
		s \xrightarrow{\rcv{\procB_1}{\procA}{\val_1}} s_1,\,
		s \xrightarrow{\rcv{\procB_2}{\procA}{\val_2}} s_2 \in \delta_\procA.\, \\
		\qquad  \procB_1 \neq \procB_2
		\; \implies \;
		\forall~G_2 \in \transAnnoFunDest(s \xrightarrow{\rcv{\procB_2}{\procA}{\val_2}} s_2). \;
		\snd{\procB_1}{\procA}{\val_1} \notin \semavail^{\procA}_{(G_2 \ldots)} \enspace.
	\end{array}
	\]
\end{definition}

We wish to adapt these conditions to define \CharacterizationOne.
However, unlike Send and Receive Validity, which are defined on special state machines, namely the subset construction for each role, the \ProblemOne problem asks whether arbitrary state machines implement the given~$\GG$.

We first present a \emph{state decoration} function which maps local states in an arbitrary deterministic finite state machine to sets of global states in $\GG$. 
Intuitively, state decoration captures all global states that can be reached in the projection by erasure automaton $\projerasure{\GG}{\procB}$ on the same prefixes that reach the present state in the local state machine. 
\begin{definition}[State decoration with respect to $\GG$]
	\label{def:state-decoration}
Let $\procA\in \Procs$ be a role and let $A = (Q, \Alphabet_\procA, s_0, \delta, F)$ be a deterministic finite state machine for $\procA$.
	Let
	$
	\projerasure{\GG}{\procA} \,=
	(Q_{\GG},
	\Alphabet_\procA \dunion \set{\emptystring},
	\projerasuretrans,
	q_{0, \GG},
	F_{\GG})
	$ 
	be $\procA$'s projection by erasure state machine for $\GG$. 
	We define a total function $d_{\GG,A} : Q \rightarrow \powersetof{Q_\GG}$ that maps each state in $A$ to a subset of states in $\projerasure{\GG}{\procA}$ such that:
\[
	d_{\GG,A,\procA}(s) = 
	\{ q \in Q_{\GG} 
	\mid 
	\exists u \in \Alphabet^*_\procA.~
	s_0 \xrightarrow{u}\mathrel{\vphantom{\to}^*} s \in \delta \land 
	q_{0,\GG} \xrightarrow{u}\mathrel{\vphantom{\to}^*} q \in \projerasuretrans 
	\}\enspace.
	\]
	We refer to $d_{\GG,A,\procA}(s)$ as the \emph{decoration set} of $s$, and omit the subscripts $\GG,A,\procA$ when clear from context. 
\end{definition}

\begin{remark}
	Note that the subset construction can be viewed as a special state machine for which the state decoration function is the identity function. 
	In other words, for all $s \in Q_\procA$ where $Q_\procA$ is the set of states of $\subsetcons{\GG}{\procA}$, $d(s) = s$. 
\end{remark}

\newcommand{\sendDecVal}{\emph{Send Decoration Validity}\xspace}
\newcommand{\receiveDecVal}{\emph{Receive Decoration Validity}\xspace}
\newcommand{\localLangIncl}{\emph{Local Language Inclusion}\xspace}
\newcommand{\transitionExhaustive}{\emph{Transition Exhaustivity}\xspace} 
\newcommand{\finalVal}{\emph{Final State Validity}\xspace}
\noindent
We are now equipped to present \CharacterizationOne.  
\begin{definition}[\CharacterizationOne]
	\label{def:characterization-one}
	Let $\GG$ be a global type and $\mathcal{A}$ be a CSM.
\CharacterizationOne is satisfied when for all $\procA \in \Procs$, with $A_\procA = (Q_\procA, \Alphabet_\procA, \delta_\procA, s_{0,\procA}, F_\procA)$ denoting the state machine for $\procA$ in $\mathcal{A}$, the following conditions hold:
	\begin{itemize}
		\item \label{cond:send-decoration-validity} 
		\sendDecVal: every send transition $s \xrightarrow{x} s' \in \delta_\procA$ is enabled in all states decorating $s$: \\$\forall s \xrightarrow{\snd{\procA}{\procB}{\val}} s' \in \delta_\procA.~\transAnnoFunc(d(s)  \xrightarrow{\snd{\procA}{\procB}{\val}} d(s')) = d(s)$. \\
		\item \label{cond:receive-decoration-validity} 
		\receiveDecVal: no receive transition is enabled in an alternative continuation originating from the same state: \\$\begin{array}{l}
		\forall s \xrightarrow{\rcv{\procB_1}{\procA}{\val_1}} s_1,~s \xrightarrow{x} s_2 \in \delta_\procA.~x \neq \rcv{\procB_1}{\procA}{\_} \implies \\\qquad
		\forall G' \in \transAnnoFunDest(d(s) \xrightarrow{x} d(s_2)). \;
		\snd{\procB_1}{\procA}{\val_1} \notin \semavail^{\procA}_{(G' \ldots)} .
		\end{array}$ \\
\item \label{cond:send-state-exhaustive}
		\transitionExhaustive: every transition that is enabled in some global state decorating $s$ must be an outgoing transition from $s$: \\
		$\forall s \in Q.~\forall G \xrightarrow{x}\mathrel{\vphantom{\to}^*} G' \in \projerasuretrans.~G \in d(s) \implies \exists s' \in Q.~s \xrightarrow{x} s' \in \delta_\procA$.  \\
		\item \label{cond:final-states-final} 
		\finalVal: a reachable state with a non-empty decorating set is final if its decorating set contains a final global state: \\		
		$\forall s \in Q.~d(s) \neq \emptyset \implies (d(s) \cap F_\GG \neq \emptyset \implies s \in F_\procA)$. 
	\end{itemize}
\end{definition}

\noindent
We want to show the following equivalence to prove \cref{thm:equivalence-one}: 
\begin{center}
	\CharacterizationOne $\Leftrightarrow \mathcal{A}$ refines $\subsetprojCSM$ and $\subsetprojCSM$ refines $\mathcal{A}$.
\end{center}

We address soundness (the forward direction) and completeness (the backward direction) in turn. 
Soundness states that \CharacterizationOne is sufficient to show that $\mathcal{A}$ preserves all behaviors of the subset construction, and does not introduce new behaviors. 

We say that a state machine $A$ for role $\procA$ satisfies \localLangIncl if it satisfies $\lang(\GG) \wproj_{\Alphabet_{\procA}} \subseteq \lang(A)$. 
The following lemma, proven in \appendixRef{app:proofs}, establishes that \localLangIncl follows from \transitionExhaustive and \finalVal.
\begin{restatable}{lemma}{impliesLocalLangIncl}
	\label{lm:implies-local-lang-incl}
	Let $A_\procA = (Q_\procA, \Alphabet_\procA, \delta_\procA, s_{0,\procA}, F_\procA)$ denote the state machine for $\procA$ in $\mathcal{A}$. 
	Then, \transitionExhaustive and \finalVal imply
	$\lang(\GG) \wproj_{\Alphabet_{\procA}} \subseteq \lang(A_\procA)$. 
\end{restatable}

The fact that $\mathcal{A}$ preserves behaviors follows immediately from \localLangIncl. 
The fact that $\mathcal{A}$ does not introduce new behaviors, on the other hand, is enforced by \sendDecVal and \receiveDecVal.

In the soundness proof for each of our conditions, we prove refinement via structural induction on traces. 
We show refinement in two steps, first showing that any trace in one CSM is a trace in the other, and then showing that any terminated trace in one CSM is terminated in the other and maximal. 

We recall two definitions from~\cite{DBLP:conf/cav/LiSWZ23} used in the soundness proof. 
\begin{definition}[Intersection sets]
	\label{def:intersection-sets}
	Let $\GG$ be a global type and $\semglobalsync(\GG)$ be the corresponding state machine.
	Let $\procA$ be a role and
	$w \in \AlphAsync^*$ be a word.
	We define the set of possible runs $\globcomplocal{\GG}{\procA}{w}$
	as all maximal runs of $\semglobalsync(\GG)$ that are consistent with $\procA$'s local view of $w$:
	\[
	\globcomplocal{\GG}{\procA}{w}
	\is
	\set{
		\run
		\text{ is a maximal run of }
		\semglobalsync(\GG)
		\mid
		w \wproj_{\Alphabet_\procA} \preforder \SyncToAsync(\trace(\run)) \wproj _{\Alphabet_\procA}
	}
	\enspace .
	\]
	We denote the intersection of the possible run sets for all roles as
	\[
	I(w) 
	\is 
	\Inters_{\procA \in \Procs} \globcomplocal{\GG}{\procA}{w}
	\enspace.
	\]
\end{definition}

\begin{definition}[Unique splitting of a possible run]
	Let $\GG$ be a global type, $\procA$ a role, and $w \in \AlphAsync^*$ a word. Let $\run$ be a possible run in $\globcomplocal{\GG}{\procA}{w}$. 
	We define the longest prefix of $\run$ matching $w$:
	\[
	\alpha'
	\is
	\max
	\set{
		\run'
		\mid
		\run' \leq \run ~\wedge~
		\SyncToAsync(\trace(\run')) \wproj _{\Alphabet_\procA} \preforder w \wproj_{\Alphabet_\procA}
	}
	\enspace .
	\]
	If $\alpha' \neq \run$, we can split $\run$ into
$
	\run = \alpha \cdot G \xrightarrow{l} G' \cdot \beta
	$
where $\alpha' = \alpha \cdot G$, $G'$ denotes the state following $G$, and $\beta$ denotes the suffix of $\run$ following $\alpha \cdot G \cdot G'$.
We call $\alpha \cdot G \xrightarrow{l} G' \cdot \beta$ the unique splitting of $\run$ for $\procA$ matching $w$. 
	We omit the role $\procA$ when obvious from context.
	This splitting is always unique because the maximal prefix of any $\run \in \globcomplocal{\GG}{\procA}{w}$ matching $w$ is unique.
\end{definition}

\begin{lemma}[Soundness of \CharacterizationOne]
	\CharacterizationOne implies that $\mathcal{A}$ and $\subsetprojCSM$ are equivalent. 
\end{lemma}

\begin{proof}
	The proof that \CharacterizationOne implies $\subsetprojCSM$ refines $\mathcal{A}$ depends only on \localLangIncl and can be straightforwardly adapted from \cite[Lemma 4.4]{DBLP:conf/cav/LiSWZ23}.
	We instead focus on showing that \CharacterizationOne implies $\mathcal{A}$ refines $\subsetprojCSM$, which depends on the other two conditions in \CharacterizationOne.
First, we prove that any trace in $\mathcal{A}$ is a trace in $\subsetprojCSM$: 
	
	\noindent \textit{Claim 1: } $\forall~w \in \AlphAsync^\infty.~w$ is a trace in $\mathcal{A}$ implies $w$ is a trace in $\subsetprojCSM$.
	
	We prove the claim by induction for all finite $w$. 
	The infinite case follows from the finite case because $\subsetprojCSM$ is deterministic and all prefixes of $w$ are traces of~$\mathcal{A}$ and, hence, of $\subsetprojCSM$.
The base cases, where $w = \emptystring$, is trivially discharged by the fact that $\emptystring$ is a trace of all CSMs.
	In the inductive step, assume that $w$ is a trace of~$\mathcal{A}$.
	Let $x \in \AlphAsync$ such that $wx$ is a trace of $\mathcal{A}$. 
	We want to show that $wx$ is also a trace of $\subsetprojCSM$. 
	
	From the induction hypothesis, we know that $w$ is a trace of $\subsetprojCSM$. 
	Let~$\xi$ be the channel configuration uniquely determined by $w$.
	Let $(\vec{s},\xi)$ be the $\mathcal{A}$ configuration reached on $w$, and let $(\vec{t},\xi)$ be the $\subsetprojCSM$ configuration reached on~$w$.
	
	Let $\procB$ be the role such that $x \in \Alphabet_\procB$, and let $s$, $t$ denote $\vec{s}_\procB$, $\vec{t}_\procB$ from the respective CSM configurations reached on $w$ for $\mathcal{A}$ and $\subsetprojCSM$.  
	
	To show that $wx$ is a trace of $\subsetprojCSM$, it thus suffices to show that there exists a state $t'$ and a transition $t \xrightarrow{x} t'$ in $\subsetproj{\GG}{\procB}$.

	Since $\subsetprojCSM$ implements $\GG$, all finite traces of $\subsetprojCSM$ are prefixes of $\lang(\GG)$. 
	In other words, $w \in \text{pref}(\lang(\GG))$. 
	Let $\run$ be a run such that $\run \in I(w)$; such a run must exist from~\cite[Lemma 6.3]{DBLP:conf/cav/LiSWZ23}. 
	Let $\alpha \cdot G \xrightarrow{l} G' \cdot \beta$ be the unique splitting of $\run$ for $\procB$ matching $w$. 
From the definition of state decoration, it holds that $G \in d(s)$. 
	From the definition of the subset construction, it holds that $G \in t$.

	We proceed by case analysis on whether $x$ is a send or receive event. 
	\begin{itemize}
		\item Case $x \in \Alphabet_{\procB,!}$. 
		Let $x = \snd{\procB}{\procC}{\val}$. 
By assumption, there exists $s \xrightarrow{\snd{\procB}{\procC}{\val}} s'$ in $A_\procB$. 
		We instantiate \sendDecVal from \CharacterizationOne with $\procB$ and this transition to obtain: 
		\[
		\transAnnoFunc(d(s) \xrightarrow{\snd{\procB}{\procC}{\val}} d(s')) = d(s)
		\enspace .
		\]
		From $G \in d(s)$, it follows that there exists $G' \in Q_\GG$ such that $G \xrightarrow{x} \mathrel{\vphantom{\to}^*} G' \in \projerasuretrans$. 
		Because $G \in t$, the existence of $t'$ such that $t \xrightarrow{\snd{\procB}{\procC}{\val}} t'$ is a transition in $\subsetproj{\GG}{\procA}$ follows immediately from the definition of $\subsetproj{\GG}{\procB}$'s transition relation. 
\item Case $x \in \Alphabet_{\procB,?}$. 
		Let $x = \rcv{\procC}{\procB}{\val}$. 
		
		From the fact that $\run$ is a maximal run in $\GG$ with unique splitting $\alpha \cdot G \xrightarrow{l} G' \cdot \beta$ for $\procB$ matching w, it holds that $w \wproj_{\Alphabet_{\procB}} \cdot
		\, \SyncToAsync(l) \wproj_{\Alphabet_{\procB}} \in \text{pref}(\lang(\GG)) \wproj_{\Alphabet_{\procB}}$.
		From~\cite[Lemma 4.3]{DBLP:conf/cav/LiSWZ23}, $\lang(\GG) \wproj_{\Alphabet_\procB} = \lang(\subsetproj{\GG}{\procB})$. 
		Therefore, there exists a $t''$ such that $t \xrightarrow{\SyncToAsync(l) \wproj_{\Alphabet_\procB}} t''$ is a transition in $\subsetproj{\GG}{\procB}$. 
		From \transitionExhaustive, there likewise exists an $s''$ such that $s \xrightarrow{\SyncToAsync(l) \wproj_{\Alphabet_\procB}} s''$ is a transition in $A_\procB$. 
		
		We now proceed by showing that it must be the case that $\SyncToAsync(l) \wproj_{\Alphabet_{\procB}} = x$.
		The reasoning closely follows that in \cite[Lemma 6.4]{DBLP:conf/cav/LiSWZ23}, which showed that if Receive Validity holds for the subset construction, and some role's subset construction automaton can perform a receive action, then the trace extended with the receive action remains consistent with any global run it was consistent with before.
		We generalize this property in terms of available message sets in the following lemma, whose proof can be found in \appendixRef{app:proofs}.

		\begin{restatable}{lemma} {aboutReceiveDecorationValidity}
			\label{lm:about-receive-decoration-validity}
			Let $\mathcal{A}$ be a CSM, $\procB$ be a role, and $w$, $wx$ be traces of $\mathcal{A}$ such that $x = \rcv{\procC}{\procB}{\val}$. 
			Let $s$ be the state of $\procB$'s state machine in the $\mathcal{A}$ configuration reached on $w$. 
			Let $\run$ be a run that is consistent with $w$, i.e. for all $\procA \in \Procs.~w \wproj_{\Alphabet_{\procA}} \leq \SyncToAsync(\trace(\run)) \wproj_{\Alphabet_{\procA}}$. 
			Let $\alpha \cdot G \xrightarrow{l} G' \cdot \beta$ be the unique splitting of $\run$ for $\procB$ matching~$w$.
			If $\snd{\procC}{\procB}{\val} \notin \semavail^{\procB}_{(G' \ldots)}$, then 
$x = \SyncToAsync(l) \wproj_{\Alphabet_{\procB}}$. 
		\end{restatable}
		
	\noindent 
	We wish to apply \cref{lm:about-receive-decoration-validity} with $\run$ to conclude that $\SyncToAsync(l) \wproj_{\Alphabet_{\procB}} = x$. 
	We satisfy the assumption that $\snd{\procC}{\procB}{\val} \notin \semavail^{\procB}_{(G' \ldots)}$ by instantiating \receiveDecVal with $s \xrightarrow{\rcv{\procC}{\procB}{\val}} s'$, $s \xrightarrow{\SyncToAsync(l) \wproj_{\Alphabet_\procB}} s''$ and $G'$.
The fact that $G' \in \transAnnoFunDest(d(s) \xrightarrow{\SyncToAsync(l) \wproj_{\Alphabet_\procB}} d(s''))$ follows from the fact that $\alpha \cdot G \xrightarrow{l} G' \cdot \beta$ is a run in $\GG$ and the definition of state decoration  (\cref{def:state-decoration}).
	Thus, we conclude from $\SyncToAsync(l) \wproj_{\Alphabet_{\procB}} = x$ that there exists a transition $t \xrightarrow{x} t''$ in $\subsetproj{\GG}{\procB}$. 
\end{itemize}

This concludes our proof that any trace in $\mathcal{A}$ is also a trace of $\subsetprojCSM$.

\noindent \textit{Claim 2: } $\forall~w \in \AlphAsync^*$.~$w$ is terminated in $\mathcal{A} \implies w$ is terminated in $\subsetprojCSM$ and $w$ is maximal in $\mathcal{A}$.

Let $w$ be a terminated trace in $\mathcal{A}$. 
By Claim 1, $w$ is also a trace in $\subsetprojCSM$. 
Let $\xi$ be the channel configuration uniquely determined by $w$.
Let the $\subsetprojCSM$ configuration reached on $w$ be $(\vec{t},\xi)$, and let $(\vec{s},\xi)$ be the $\mathcal{A}$ configuration reached on $w$. 
To see that every terminated trace in $\mathcal{A}$ is also terminated in $\subsetprojCSM$, assume by contradiction that $w$ is not terminated in $\subsetprojCSM$.
Because $\subsetprojCSM$ is deadlock-free, there must exist a role that can take a step in $\subsetprojCSM$. 
Let $\procB$ be this role, and let $x$ be the transition that is enabled from~$\vec{t}_\procB$.
From \localLangIncl and the fact that $\subsetprojCSM$ is deadlock-free, it holds that $x$ is also enabled from $\vec{s}_\procB$. 
We arrive at a contradiction. 
To see that every terminated trace in $\mathcal{A}$ in maximal, from the above we know that $w$ is terminated in $\subsetprojCSM$. From the fact that $\subsetprojCSM$ is deadlock-free, $w$ is maximal in $\subsetprojCSM$: all states in $\vec{t}$ are final and all channels in $\xi$ are empty. 
From \localLangIncl, it follows that all states in $\vec{s}$ are also final, and thus $w$ is maximal in $\mathcal{A}$.
\proofend
\end{proof}

\begin{lemma}[Completeness of \CharacterizationOne]
	\label{lm:completeness-characterization-one}
If $\mathcal{A}$ and $\subsetprojCSM$ are equivalent, then \CharacterizationOne holds. 
\end{lemma}

We show completeness via modus tollens: we assume a violation in \CharacterizationOne and the fact that $\mathcal{A}$ and $\subsetprojCSM$ are equivalent, and prove a contradiction. 
Since \CharacterizationOne is a conjunction of four conditions, we derive a contradiction from the violation of each condition in turn. 
In the interest of proof reuse, we specify which of the two refinement conjuncts we contradict for each condition, and refer the reader to \appendixRef{app:proofs} for the full proofs.

From the negation of \transitionExhaustive and \finalVal, we contradict the fact that $\subsetprojCSM$ refines $\mathcal{A}$. 
\begin{restatable}{lemma}{characterizationOneTransitionExhaustiveFinalValComplete}
	\label{lm:characterization-one-transition-exhaustive-final-val-complete}
If $\mathcal{A}$ violates \transitionExhaustive or \finalVal, then it does not hold that $\subsetprojCSM$ refines $\mathcal{A}$.
\end{restatable}

Unlike the proofs for \transitionExhaustive and \finalVal, the proofs for the remaining two conditions require both refinement conjuncts to prove a contradiction.
Both proofs find a contradiction by obtaining a witness from the violation of \sendDecVal and \receiveDecVal respectively, and showing that the same witness can be used to refute Send and Receive Validity for the subset construction. 

\begin{restatable}{lemma}{characterizationOneSendReceiveDecValComplete}
		\label{lm:characterization-one-send-receive-dec-val-complete}
If $\mathcal{A}$ violates \sendDecVal or \receiveDecVal, then it does not hold that $\mathcal{A}$ and $\subsetprojCSM$ are equivalent.
\end{restatable}

 \section{Deciding \emph{\ProblemTwo}}
We now turn our attention to \ProblemTwo, which asks when an implementation can safely substitute another in all well-behaved contexts with respect to $\GG$. 
Here, we introduce a new notion of refinement with respect to a global type. 
\begin{definition}[Protocol refinement with respect to $\GG$]
	\label{def:protocol-refinement}
	We say that a CSM $\CSM{A}$ refines a CSM $\CSM{B}$ with respect to a global type $\GG$ if the following properties hold: 
	\begin{inparaenum}[(i)]
		\item \label{def:refinement-subprotocol-fidelity}
		\emph{subprotocol fidelity:} $\exists S \subseteq \lang(\semglobal(\GG)).~\lang(\CSM{A}) = \interswaplang(\SyncToAsync(S))$,
		\item \label{def:refinement-language-inclusion}
		\emph{language inclusion:} $\lang(\CSM{A}) \subseteq \lang(\CSM{B})$, and 
		\item \label{def:refinement-deadlock-freedom}
		\emph{deadlock freedom:} $\CSM{A}$ is deadlock-free.
	\end{inparaenum}
\end{definition}
\cref{def:refinement-subprotocol-fidelity}, \emph{subprotocol fidelity}, sets our notion of refinement apart from standard refinement.
We motivate this difference briefly using an example.
Consider the CSM consisting of the subset construction for $\procA$ and $B'_\procB$, depicted in \cref{fig:refinement-def}. 
This CSM recognizes only words of the form 
$(\snd{\procA}{\procB}{\val})^ω$. 
It is nonetheless considered to refine the global type 
$\GG_{loop} \is \mu t.~\procA \xrightarrow{} \procB: \val.~t$ according to the standard notion of refinement, despite the fact that $\procA$'s messages are never received by $\procB$.  
This is because $\lang(\GG_{loop})$, containing only infinite words, is defined in terms of an asymmetric downward closure operator $\preceq_\interswap^\omega$, which allows receives to be infinitely postponed.
We desire a notion of refinement that allows roles to select which runs to follow in a global type, but disallows them from selecting which words to implement among ones that follow the same run. 
More formally, our notion of protocol refinement prohibits selectively implementing words that are equivalent under the indistinguishability relation $\interswap$: any CSM that refines another with respect to a global type has a language that is closed under $\interswap$. 

\begin{figure}[t]
\centering
	\begin{subfigure}[b]{.45\textwidth}
		\centering
			\begin{tikzpicture}[node distance=1cm and 2cm,>=stealth', line width=0.25mm]
				\node[draw, circle, minimum size=0.5cm, initial left=, initial text =](q0){};
				\path[->](q0) edge [loop right] node {$\snd{\procA}{\procB}{\val}$} (q0);
			\end{tikzpicture}
		\caption{State machine $\subsetcons{\GG}{\procA}$} \label{fig:refinement-def-fsm-p}
\end{subfigure}\begin{subfigure}[b]{.45\textwidth}
		\centering
			\begin{tikzpicture}[node distance=1cm and 2cm,>=stealth', line width=0.25mm]
				\node[draw, circle, minimum size=0.5cm, initial left=, initial text =](q0){};
			\end{tikzpicture}
		\caption{State machine $B'_\procB$} \label{fig:refinement-def-fsm-q}
\end{subfigure}\hfill
\caption{CSM violating subprotocol fidelity with respect to $\GG_{loop}$}
\label{fig:refinement-def}
\end{figure}
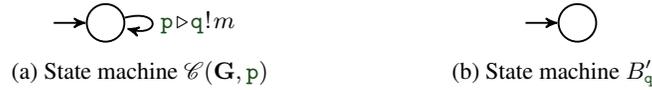

In the remainder of the paper, we refer to refinement with respect to $\GG$, and omit mention of $\GG$ when clear from context. 
Again using the fact that $\subsetprojCSM$ is an implementation for $\GG$, we say that a CSM $\CSM{A}$ refines $\GG$ if it refines $\subsetprojCSM$.

\newcommand{\MonolithicRefinement}{\emph{Monolithic Protocol Refinement}\xspace}
We motivate our formulation of the \ProblemTwo problem by posing the following variation of \ProblemOne, which we call \MonolithicRefinement: 
\begin{center}
	Given an implementable global type $\GG$ and a CSM $\mathcal{A}$, does $\mathcal{A}$ \emph{refine} $\subsetprojCSM$? 
\end{center}

\newcommand{\CharacterizationOnePrime}{$C_1'$\xspace}
This variation asks for a condition, \CharacterizationOnePrime, that satisfies the equivalence: 
\begin{center}
	\CharacterizationOnePrime $\Leftrightarrow$ $\mathcal{A}$ refines $\subsetprojCSM$.
\end{center}

Clearly, \CharacterizationOne is still a sound candidate as equivalence of two CSMs implies bi-directional protocol refinement. 
It is instructive to analyze why the completeness arguments for \CharacterizationOne fail. 
Recall that the completeness proofs for \sendDecVal and \receiveDecVal used the violation of each condition to obtain a local state with a non-empty decoration set, which in turn gives rise to a prefix in $\lang(\GG)$ that must be a trace in the subset construction. This trace is then replayed in the arbitrary CSM, extended in the arbitrary CSM, and then replayed again in the subset construction. 
This sequence of replaying arguments critically relied on both the assumption that $\mathcal{A}$ refines $\subsetprojCSM$, and the assumption that $\subsetprojCSM$ refines $\mathcal{A}$. 

If we cannot assume that $\mathcal{A}$ recognizes every behavior of $\subsetprojCSM$, then the reachable local states of $\mathcal{A}$ are no longer precisely characterized by having a non-empty decoration set.

\begin{figure}[t]
\centering
\begin{subfigure}[b]{.35\textwidth}
    \resizebox{\textwidth}{!}{
        \begin{tikzpicture}[node distance=1cm and 2cm,>=stealth', line width=0.25mm]
            \node[draw, circle, minimum size=0.5cm, initial left=, initial text =](q0){};
            \node[draw, circle, minimum size = 0.5cm, right=1.5cm of q0](q1){};
            \node[draw, circle, minimum size=0.5cm, accepting, right=1.5cm of q1](q2){};
            \path[->](q0) edge node[above] {$\snd{\procA}{\procB}{\msgM}$} (q1);
            \path[->](q1) edge node[above] {$\snd{\procA}{\procC}{\msgM}$} (q2);
        \end{tikzpicture}
    }
    \caption{State machine $\subsetcons{\GG}{\procA}$} \label{fig:refinement-incompleteness-fsm-p}
\end{subfigure}\\
\vspace{2ex}
\begin{subfigure}[b]{.48\textwidth}
	\centering
	\resizebox{\textwidth}{!}{
		\begin{tikzpicture}[node distance=1cm and 2cm,>=stealth', line width=0.25mm]
			\node[draw, circle, minimum size=0.5cm, initial left=, initial text =](q0){};
			\node[draw, circle, minimum size = 0.5cm, right=1.5cm of q0](q1){};
			\node[draw, circle, minimum size = 0.5cm, right=1.5cm of q1](q2){};
			\node[draw, circle, minimum size = 0.5cm, accepting, right=1.5cm of q2](q3){};
\path[->](q0) edge node[above] {$\rcv{\procA}{\procB}{\msgM}$} (q1);
			\path[->](q0) edge node[below] {\fbox{$\rcv{\procC}{\procB}{\msgB}$}} (q1);
			\path[->](q1) edge node[above] {$\rcv{\procC}{\procB}{\msgB}$} (q2);
			\path[->](q1) edge node[below] {$\rcv{\procC}{\procB}{\msgO}$} (q2);
			\path[->](q2) edge node[above] {$\snd{\procB}{\procC}{\msgB}$} (q3);
\end{tikzpicture}
	}
	\caption{State machine $A'_\procB$} \label{fig:refinement-incompleteness-fsm-q}
\end{subfigure}\hfill \begin{subfigure}[b]{.48\textwidth}
\centering
\resizebox{\textwidth}{!}{
	\begin{tikzpicture}[node distance=1cm and 2cm,>=stealth', line width=0.25mm]
		\node[draw, circle, minimum size=0.5cm, initial left=, initial text =](q0){};
		\node[draw, circle, minimum size = 0.5cm, right=1.5cm of q0](q1){};
		\node[draw, circle, minimum size = 0.5cm, right=1.5cm of q1](q2){};
		\node[draw, circle, minimum size = 0.5cm, accepting, right=1.5cm of q2](q3){};
\path[->](q0) edge node[below] {$\snd{\procC}{\procB}{\msgO}$} (q1);
		\path[->](q1) edge node[above] {$\rcv{\procA}{\procC}{\msgM}$} (q2);
		\path[->](q1) edge node[below] {\fbox{$\rcv{\procB}{\procC}{\msgO}$}} (q2);
		\path[->](q2) edge node[sloped, pos=0.5, above] {$\rcv{\procB}{\procC}{\msgB}$} (q3);
		\path[->](q2) edge node[sloped, pos=0.5, below] {$\rcv{\procB}{\procC}{\msgO}$} (q3);
	\end{tikzpicture}
}
\caption{State machine $A'_\procC$} \label{fig:refinement-incompleteness-fsm-r}
\end{subfigure}
\hfill
\caption{Subset construction for $\procA$ and two state machines for $\procB$ and $\procC$\label{fig:refinement-incompleteness} for $\GG'$}
\end{figure}
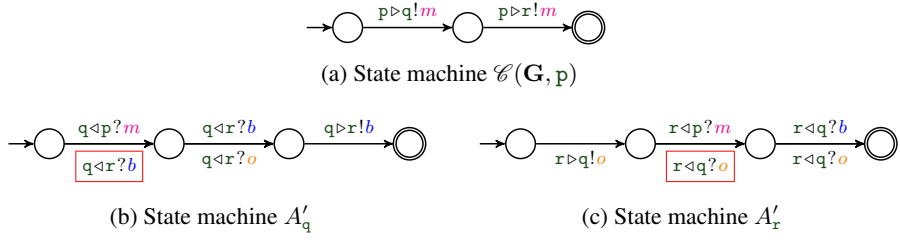
Consider the example global type $\GG'$:
\[\small
\GG' \is 
\msgFromTo{\procA}{\procB}{\msgM}. \,
+ \;
\begin{cases}
\msgFromTo{\procC}{\procB}{\msgB}. \,
\msgFromTo{\procA}{\procC}{\msgM}. \,
+ \;
\begin{cases}
\msgFromTo{\procB}{\procC}{\msgB}. \,
0 \,
\\
\msgFromTo{\procB}{\procC}{\msgO}. \,
0 \,
\end{cases}
\\
\msgFromTo{\procC}{\procB}{\msgO}. \,
\msgFromTo{\procA}{\procC}{\msgM}. \,
+ \;
\begin{cases}
\msgFromTo{\procB}{\procC}{\msgB}. \,
0 \,
\\
\msgFromTo{\procB}{\procC}{\msgO}. \,
0 \,
\end{cases}
\end{cases}
\]
Let the CSM $\mathcal{A}'$ consist of the subset construction automaton for $\procA$, and the state machines $A'_\procB$ and $A'_\procC$, given in
\cref{fig:refinement-incompleteness-fsm-q,fig:refinement-incompleteness-fsm-r}.
The receive transitions highlighted in red are safe despite violating \receiveDecVal, because $\procB$ and $\procC$ coordinate with each other on which runs of $\GG$ they eliminate: $\procC$ chooses to never send a $\msgB$ to $\procB$, thus $\procB$'s highlighted transition is safe, and conversely, $\procB$ never chooses to send $\msgO$ to $\procC$, thus $\procC$'s highlighted transition is safe.
Consequently, $\mathcal{A}'$ refines $\GG'$ despite violating \CharacterizationOne. 

This example shows that any condition \CharacterizationOnePrime that is compositional must sacrifice completeness. In fact, deciding whether an arbitrary CSM $\mathcal{A}$ refines the subset construction $\subsetprojCSM$ for some global type $\GG$ can be shown to be PSPACE-hard via a reduction from the deadlock-freedom problem for 1-safe Petri nets~\cite{DBLP:journals/eik/EsparzaN94}. 
We refer the reader to \appendixRef{app:proofs} for the full construction. 

\begin{restatable}{lemma}{monolithicRefinementHardness}
\label{lm:monolithic-refinement-hardness}
The \MonolithicRefinement problem is PSPACE-hard.
\end{restatable}

Fortunately, we can recover completeness and tractability by only allowing changes to one state machine in $\mathcal{A}$ at a time. 
Next, we formalize the notions of \emph{CSM contexts} and \emph{well-behavedness} with respect to~$\GG$. 
We use $\CSMhole$ to denote a CSM context with a hole for role $\procA \in \Procs$, and $\CSMholew{A}$ to denote the CSM obtained by instantiating the context with state machine $A$ for $\procA$.
We define well-behaved contexts in terms of the canonical implementation $\subsetproj{\GG}{\procA}$. 
\begin{definition}[Well-behaved CSM contexts with respect to $\GG$]
	Let $\CSMhole$ be a CSM context. 
	We say that $\CSMhole$ is well-behaved with respect to $\GG$ if $\CSMholew{\subsetproj{\GG}{\procA}}$ refines $\GG$. 
	We omit $\GG$ when clear from context.
\end{definition}
\ProblemTwo asks to find a \CharacterizationTwo that satisfies the following:
\begin{theorem}
	\label{thm:equivalence-two}
	Let $\GG$ be an implementable global type, $\procA$ be a role, and $A$, $B$ be state machines for role $\procA$ such that for all well-behaved contexts $\CSMhole$, $\CSMholew{B}$ refines $\GG$.
	Then, for all well-behaved contexts $\CSMhole$, $\CSMholew{A}$ refines $\CSMholew{B}$ if and only if \CharacterizationTwo is satisfied.
\end{theorem}

\subsection{\emph{\ProblemTwo} Relative to Subset Construction}

\newcommand{\CharacterizationTwoPrime}{\emph{$C_2'$}\xspace}

As a stepping stone, we first consider the special case of \ProblemTwo when $B$ is the subset construction automaton for role $\procA$.
That is, we present \CharacterizationTwoPrime that satisfies the following equivalence: 
\begin{center}
	\CharacterizationTwoPrime $\Leftrightarrow$ for all well-behaved contexts $\CSMhole$, $\CSMholew{A}$ refines
	$\CSMholew{\subsetproj{\GG}{\procA}}$. 
\end{center}
The relaxation on language equality from \ProblemOne means that state machine $A$ no longer needs to satisfy \localLangIncl, which grants us more flexibility: state machines are now permitted to remove send events. 
Let us revisit our example global type, $\GG_1$: 
\[\small
	\GG_1 \is 
	+ \;
	\begin{cases}
		\msgFromTo{\procA}{\procB}{\msgB}. \,
		\msgFromTo{\procB}{\procA}{\msgB}. \,
		0 \,
		\\
		\msgFromTo{\procA}{\procB}{\msgM}. \,
		\msgFromTo{\procB}{\procA}{\msgM}. \,
		0 \,
	\end{cases}
	\]

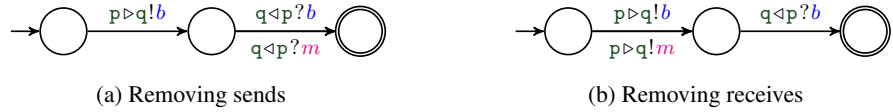
\begin{figure}[t]
\begin{subfigure}[b]{0.45\textwidth}
	\centering
\resizebox{.97\textwidth}{!}{
	\begin{tikzpicture}[node distance=1cm and 2cm,>=stealth', line width=0.25mm]
		\node[draw, circle, minimum size=0.7cm, initial left=, initial text =](q0){};
		\node[draw, circle, minimum size = 0.7cm, right=1.5cm of q0](q1){};
		\node[draw, circle, minimum size = 0.7cm, accepting, right=1.5cm of q1](qf){};
\path[->](q0) edge node[above] {$\snd{\procA}{\procB}{\msgB}$} (q1);
		\path[->](q1) edge node[below] {$\rcv{\procA}{\procB}{\msgM}$} (qf);
		\path[->](q1) edge node[above] {$\rcv{\procA}{\procB}{\msgB}$} (qf);
\end{tikzpicture}
}
    \caption{Removing sends \\\strut}
    \label{fig:prot-ref-candidate-implementation-p}
\end{subfigure}
\hfill
\begin{subfigure}[b]{0.45\textwidth}
	\centering
\resizebox{.97\textwidth}{!}{
	\begin{tikzpicture}[node distance=1cm and 2cm,>=stealth', line width=0.25mm]
		\node[draw, circle, minimum size=0.7cm, initial left=, initial text =](q0){};
		\node[draw, circle, minimum size = 0.7cm, right=1.5cm of q0](q1){};
		\node[draw, circle, minimum size = 0.7cm, accepting, right=1.5cm of q1](qf){};
		\path[->](q0) edge node[below] {$\snd{\procA}{\procB}{\msgM}$} (q1);
		\path[->](q0) edge node[above] {$\snd{\procA}{\procB}{\msgB}$} (q1);
\path[->](q1) edge node[above] {$\rcv{\procA}{\procB}{\msgB}$} (qf);
\end{tikzpicture}
}
    \caption{Removing receives\\\strut} \label{fig:prot-ref-candidate-implementation-p-adding-receives}
\end{subfigure}
\caption{Two candidate implementations for $\procA$} 
\end{figure}

Consider the candidate state machine
for role $\procA$ given in
\cref{fig:prot-ref-candidate-implementation-p}.
The CSM obtained from inserting this state machine into any well-behaved context refines~$\GG$, despite the fact that $\procA$ never sends $\msgM$.
In general, send events can safely be removed from reachable states in a local state machine without violating subprotocol fidelity or deadlock freedom, as long as \emph{not all} of them are~removed.

The same is not true of receive events, on the other hand. 
The state machine
in
\cref{fig:prot-ref-candidate-implementation-p-adding-receives}
is not a safe candidate for $\procA$, because it causes a deadlock in the well-behaved context that consists of the subset construction for every other role.

Our characterization intuitively follows the notion that input types (receive events) are covariant, and output types (send events) are contravariant. However, note that the state machine above cannot be represented in existing works~\cite{DBLP:journals/jlamp/GhilezanJPSY19,DBLP:journals/sefm/BravettiZ19,DBLP:conf/ppopp/CutnerYV22}: their local types support neither states with both outgoing send and receive events, nor states with outgoing send or receive events to/from different~roles.

\newcommand{\sendPreservation}{\emph{Send Preservation}\xspace} 
\newcommand{\receiveExhaustive}{\emph{Receive Exhaustivity}\xspace} 

Our characterization \CharacterizationTwoPrime reuses \sendDecVal, \receiveDecVal and \finalVal from \CharacterizationOne, but splits \transitionExhaustive into a separate condition for send and receive events, to reflect the aforementioned asymmetry between~them.
\begin{definition}[\CharacterizationTwoPrime]
	\label{def:characterization-two-prime}
	Let $\procA \in \Procs$ be a role and let $A = (Q, \Alphabet_{\procA}, s_0, \delta, F)$ be a state machine for $\procA$.
\CharacterizationTwoPrime is satisfied when the following conditions hold in addition to \sendDecVal, \receiveDecVal and \finalVal:
	\begin{itemize}
\item
		\sendPreservation: every state containing a send-originating global state must have at least one outgoing send transition: \\ 
		$\forall s \in Q.~\exists G \in Q_{\GG,!}.~G \in d(t) \implies \exists x \in \Alphabet_{\procA,!},~s' \in Q.~ s \xrightarrow{x} s' \in \delta$.
		\\
		\item \label{cond:receive-state-exhaustive}
		\receiveExhaustive: every receive transition that is enabled in some global state decorating $s$ must be an outgoing transition from $s$: \\
		$\forall s \in Q.~\forall G \xrightarrow{x}\mathrel{\vphantom{\to}^*} G' \in \projerasuretrans.~G \in d(s) \land x \in \Alphabet_{\procA,?} \implies \exists s' \in Q.~s \xrightarrow{x} s' \in \delta$. 
\end{itemize}
\end{definition}

\noindent
We want to show the following equivalence: 
\begin{center}
	\CharacterizationTwoPrime $\Leftrightarrow$ for all well-behaved contexts $\CSMhole$, $\CSMholew{A}$ refines
	$\CSMholew{\subsetproj{\GG}{\procA}}$. 
\end{center}

\noindent
We first prove the soundness of \CharacterizationTwoPrime.
\begin{lemma}[Soundness of \CharacterizationTwoPrime]
	If \CharacterizationTwoPrime holds, then for all well-behaved contexts $\CSMhole$, 
	$\CSMholew{A}$ refines 
	$\CSMholew{\subsetproj{\GG}{\procA}}$.
\end{lemma}
\begin{proof}
	Let $\CSMhole$ be a well-behaved context with respect to $\GG$. 
	Like before, we first prove that any trace in $\CSMholew{A}$ is a trace in $\CSMholew{\subsetproj{\GG}{\procA}}$.
	
	\noindent
	\textit{Claim 1: } $\forall~w \in \AlphAsync^\infty$.~$w$ is a trace in $\CSMholew{A} \implies w$ is a trace in $\CSMholew{\subsetproj{\GG}{\procA}}$. 
	
	The proof of Claim 1 for \CharacterizationTwoPrime differs from that for \CharacterizationOne in only two ways. 
	We discuss the differences in detail below, and avoid repeating the rest of the~proof. 
	\begin{enumerate}
		\item \CharacterizationOne grants that every role's state machine satisfies \sendDecVal and \receiveDecVal, whereas \CharacterizationTwo only guarantees the conditions for role $\procA$. 
		Correspondingly, $\CSMholew{A}$ only differs from $\CSMholew{\subsetproj{\GG}{\procA}}$ in $\procA$'s state machine; all other roles' state machines are identical between the two CSMs. 
		Therefore, the induction step requires a case analysis on the role whose alphabet the event $x$ belongs to.
		In the case that $x \in \Alphabet_{\procB}$ where $\procB \neq \procA$, the induction hypothesis is trivially re-established by the fact that $\procB$'s state machine is identical in both CSMs. 
		In the case that $x \in \Alphabet_{\procA}$, we proceed to reason that $x$ can also be performed by $\subsetproj{\GG}{\procA}$ in the same well-behaved context. 
		
		\item \CharacterizationOne includes \transitionExhaustive, which allows us to conclude that given a run with unique splitting $\alpha \cdot G \xrightarrow{l} G' \cdot \beta$ for $\procA$ matching $w$ and the fact that $G \in s$, there must exist a transition 
		$s \xrightarrow{\SyncToAsync(l) \wproj_{\Alphabet_{\procA}}} s''$ in $\procA$'s state machine. 
		\cref{lm:about-receive-decoration-validity} can then be instantiated directly with $\alpha \cdot G \xrightarrow{l} G' \cdot \beta$ to complete the proof. 
		\CharacterizationTwo, on the other hand, splits \transitionExhaustive into \sendPreservation and \receiveExhaustive, and we can only establish that such a transition exists and reuse the proof in the case that $\SyncToAsync(l) \wproj_{\Alphabet_{\procA}} \in \Alphabet_{\procA,?}$. 
		Since $A$ is permitted to remove send events, if $\SyncToAsync(l) \wproj_{\Alphabet_{\procA}} \in \Alphabet_{\procA,!}$, the transition $s \xrightarrow{\SyncToAsync(l) \wproj_{\Alphabet_{\procA}}} s''$ may not exist at all in $A$. 
However, the existence of a run $\alpha \cdot G \xrightarrow{l} G' \cdot \beta$ where $l$ is a send event for~$\procA$ makes $G$ a send-originating global state in $\procA$'s projection by erasure automaton.
		\sendPreservation thus guarantees that there exists a transition $s \xrightarrow{x'} s'''$ in $A$ such that $x' \in \Alphabet_{\procA,!}$. 
		By \sendDecVal, $x'$ originates from $G$ in the projection by erasure, and we can find another run~$\run'$ such that
		$\alpha' \cdot G \xrightarrow{l'} G'' \cdot \beta'$ is the unique splitting for $\procA$ matching $w$ and $\SyncToAsync(l') \wproj_{\Alphabet_{\procA}} = x'$. 
		We satisfy the assumption that $\snd{\procC}{\procA}{\val} \notin \semavail^{\procA}_{(G'' \ldots)}$ by instantiating \receiveDecVal with $\procA$, $s \xrightarrow{x} s'$, $s \xrightarrow{\SyncToAsync(l') \wproj_{\Alphabet_\procA}} s''$ and $G''$. 
		The fact that $G'' \in \transAnnoFunDest(d_\GG(s) \xrightarrow{\SyncToAsync(l') \wproj_{\Alphabet_\procA}} d_\GG(s''))$ follows from the fact that $\alpha \cdot G \xrightarrow{l'} G'' \cdot \beta'$ is a run in $\GG$ and \cref{def:state-decoration}. 
		Instantiating \cref{lm:about-receive-decoration-validity} with $\run'$, we obtain $\SyncToAsync(l') \wproj_{\Alphabet_{\procA}} = x$, which is a contradiction: $x$ is a receive event and $\SyncToAsync(l') \wproj_{\Alphabet_{\procA}}$ is a send event. 
		Thus, it cannot be the case that $\SyncToAsync(l') \wproj_{\Alphabet_{\procA}} \in \Alphabet_{\procA,!}$. 
	\end{enumerate}
	
	This concludes our proof that any trace in $\CSMholew{A}$ is also a trace in $\CSMholew{\subsetproj{\GG}{\procA}}$.
	
	The following claim completes our soundness proof: 
	
	\noindent
	\textit{Claim 2: } $\forall~w \in \AlphAsync^*$.~$w$ is terminated in $\CSMholew{A} \implies w$ is terminated in $\CSMholew{\subsetproj{\GG}{\procA}}$ and $w$ is maximal in $\CSMholew{A}$. 
	
	The proof of Claim 2 for \CharacterizationOne again relies on \localLangIncl, which is unavailable to \CharacterizationTwoPrime. 
	Instead, we turn to \sendPreservation, \receiveExhaustive and \finalVal to establish this claim.
	Let $w$ be a terminated trace in $\CSMholew{A}$. 
	By Claim 1, it holds that $w$ is a trace in $\CSMholew{\subsetproj{\GG}{\procA}}$. 
	Let $\xi$ be the channel configuration uniquely determined by $w$.
	Let $(\vec{s},\xi)$ be the $\CSMholew{\subsetproj{\GG}{\procA}}$ configuration reached on $w$, and let $(\vec{t},\xi)$ be the $\CSMholew{A}$ configuration reached on $w$. 
	To see that $w$ is terminated in $\CSMholew{\subsetproj{\GG}{\procA}}$, suppose by contradiction that $w$ is not terminated in $\CSMholew{\subsetproj{\GG}{\procA}}$.
	Because $\CSMholew{\subsetproj{\GG}{\procA}}$ is deadlock-free, and because the state machines for all non-$\procA$ roles are identical between the two CSMs, it must be the case that $\procA$ witnesses the non-termination of $w$, in other words, $\subsetproj{\GG}{\procA}$ can take a transition that $A$ cannot. 
	Let $\vec{s}_\procA \xrightarrow{x} s'$ be the transition that $\procA$ can take from $\vec{s}_\procA$. 
	Let $G$ be a state in $\vec{s}_\procA$; such a state is guaranteed to exist by the fact that no reachable states in the subset construction are empty. 
	Then, in the projection by erasure automaton, the initial state reaches $G$ on~$w \wproj_{\Alphabet_{\procA}}$.
	By the fact that $w$ is a trace of $\CSMholew{A}$, it holds that $s_0$ reaches $\vec{s}_\procA$ on $w \wproj_{\Alphabet_{\procA}}$ in~$A$.
	By the definition of state decoration, $G \in d(\vec{t}_\procA)$. 
	\begin{itemize}
		\item If $x \in \Alphabet_!$, it follows that $G$ is a send-originating global state. 
		By \sendPreservation, for any state in $A$ that contains at least one send-originating global state, of which $\vec{t}_\procA$ is one, there exists a transition $\vec{t}_\procA \xrightarrow{x'} t'$ such that $x' \in \Alphabet_{\procA,!}$. 
		Because send transitions in a CSM are always enabled, role $\procA$ can take this transition in $\CSMholew{A}$. 
		We reach a contradiction to the fact that $w$ is terminated in $\CSMholew{A}$. 
		\item If $x \in \Alphabet_?$, it follows that $G$ is a receive-originating global state.
		From \receiveExhaustive, any receive event that originates from any global state in $d(\vec{t}_\procA)$ must also originate from $\vec{t}_\procA$. 
		Therefore, there must exist $t'$ such that $\vec{t}_\procA \xrightarrow{x} t'$ is a transition in $B'_\procA$. 
		Because the channel configuration is identical in both CSMs, role $\procA$ can take this transition in $\CSMholew{A}$.
		We again reach a contradiction to the fact that $w$ is terminated in $\CSMholew{A}$. 
	\end{itemize}
	To see that $w$ is maximal in $\CSMholew{A}$, observe that for all roles $\procB \neq \procA$, $\vec{s}_\procB = \vec{t}_\procB$. Thus, it remains to show that $\vec{t}_\procA$ is a final state in $A$. 
	Because $\vec{s}_\procA$ is a final state, by the definition of the subset construction there exists a global state $G \in \vec{s}_\procA$ such that the projection erasure automaton reaches $G$ on $w \wproj_{\Alphabet_{\procA}}$ and $G$ is a final state. 
	Because $A$ reaches $\vec{t}_\procA$ on $w \wproj_{\Alphabet_{\procA}}$, by~\cref{def:state-decoration} it holds that $G \in d(\vec{t}_\procA)$. 
	By \finalVal, it holds that $\vec{t}_\procA$ is a final state in $A$.
	This concludes our proof that any terminated trace in $\CSMholew{A}$ is also a terminated trace in $\CSMholew{\subsetproj{\GG}{\procA}}$, and is maximal in $\CSMholew{A}$. 
	
	Together, Claim 1 and 2 establish that $\CSMholew{A}$ satisfies language inclusion (\cref{def:refinement-language-inclusion}) and deadlock freedom (\cref{def:refinement-deadlock-freedom}). 
	It remains to show that $\CSMholew{A}$ satisfies subprotocol fidelity (\cref{def:refinement-subprotocol-fidelity}). 
	This follows immediately from \cite[Lemma 22]{DBLP:conf/concur/MajumdarMSZ21}, which states that all CSM languages are closed under the indistinguishability relation $\interswap$. 
	\proofend
\end{proof}

\begin{lemma}[Completeness of \CharacterizationTwoPrime]
	\label{lm:completeness-characterization-two-prime}
	If for all well-behaved contexts $\CSMhole$, 
	$\CSMholew{A}$ refines 
	$\CSMholew{\subsetproj{\GG}{\procA}}$, then \CharacterizationTwoPrime holds.
\end{lemma}

\newcommand{\subsetprojholew}[1]{\mathscr{C}(\GG)[#1]_\procA}

As before, we prove the modus tollens of this implication, which states that if \CharacterizationTwoPrime does not hold, then there exists a well-behaved context $\CSMholew{\cdot}$ such that $\CSMholew{A}$ does not protocol-refine
$\CSMholew{\subsetproj{\GG}{\procA}}$. 

We first turn our attention to finding a well-behaved witness context $\CSMholew{\cdot}$ such that we can refute subprotocol fidelity, language inclusion, or deadlock freedom. 
It turns out that the context consisting of the subset construction automaton for every other role is a suitable witness. We denote this context by $\subsetprojholew{\cdot}$ and note that it is trivially well-behaved because $\subsetprojholew{\subsetproj{\GG}{\procA}} = \subsetprojCSM$.

Recall from the completeness arguments for \CharacterizationOne that we obtained a violating state in some state machine $A$ with a non-empty decoration set from the negation of each condition in \CharacterizationOne. 
From this state's decoration set we obtained a witness global state $G$, and in turn a run $\alpha \cdot G$ in $\GG$, and from the assumption that $\subsetprojCSM$ refines $\mathcal{A}$, we argued that $\SyncToAsync(\trace(\alpha \cdot G))$ is a trace in $\mathcal{A}$. 
We then showed that $A$ is in the violating state in the $\mathcal{A}$ configuration reached on $\SyncToAsync(\trace(\alpha \cdot G))$, and from there we used each violated condition to find a contradiction.

The completeness proof for \CharacterizationTwoPrime cannot similarly use the fact that $\subsetprojCSM$ refines $\subsetprojholew{A}$. 
Instead, we must separately establish that every state with a non-empty decoration set can be reached on a trace shared by both $\subsetprojholew{A}$ and $\subsetprojCSM$. 
The following lemma achieves~this: 

\begin{lemma}
	\label{lm:hack-for-subtype}
Let $A$ be a state machine for $\procA$ and $s$ be a state in $A$.
	Let $G \in d(s)$, and let $u \in \Alphabet^*_\procA$ be a word such that
	$s_0 \xrightarrow{u} \mathrel{\vphantom{\to}^*} s$ in $A$.
	Then, there exists a run $\alpha \cdot G$ of $\semglobalsync(\GG)$ such that 
	$\SyncToAsync(\trace(\alpha \cdot G)) \wproj_{\Alphabet_\procA} = u$, 
	$\SyncToAsync(\trace(\alpha \cdot G))$ is a trace in $\subsetprojholew{A}$ and in the CSM configuration reached on $\SyncToAsync(\trace(\alpha \cdot G))$, $A$ is in state $s$. 
\end{lemma}

With \cref{lm:hack-for-subtype} replacing the assumption that $\subsetprojCSM$ refines $\subsetprojholew{A}$, we can reuse the construction in~\cref{lm:characterization-one-send-receive-dec-val-complete} to obtain a word that is a trace in $\subsetprojholew{A}$ but not a trace in $\subsetprojCSM$, thus evidencing the necessity of \sendDecVal and \receiveDecVal. The proof of \cref{lm:stepping-stone-send-receive-dec-val-complete} proceeds identically to that of~\cref{lm:characterization-one-send-receive-dec-val-complete} and is thus omitted. 

\begin{lemma}
	\label{lm:stepping-stone-send-receive-dec-val-complete}
If $A$ violates \sendDecVal or \receiveDecVal, then it does not hold that for all well-behaved contexts $\CSMhole$, $\CSMholew{A}$ refines  $\subsetprojholew{A}$. 
\end{lemma}

We also use \cref{lm:hack-for-subtype} to show the necessity of \sendPreservation, \receiveExhaustive and \finalVal. 
As a starting point, let $A$, $s$, $u$ and $\alpha \cdot G$ be obtained from \cref{lm:hack-for-subtype} and the violation of \sendPreservation. 
To show the necessity of \sendPreservation, we consider the largest extension $v$ of $u$ in $\subsetprojholew{A}$. 
In the case that $u$ is terminated in $\subsetprojholew{A}$, we refute deadlock freedom from the fact that $u$ is not maximal: $G \in s$ is a send-originating state, and final states in $\semglobal(\GG)$ do not contain outgoing transitions. 
If $v \neq u$, there exists a run $\alpha \cdot G \xrightarrow{\procA \xrightarrow{} \procB: \val} G' \cdot \beta$ such that 
$\SyncToAsync(\trace(\alpha \cdot G \xrightarrow{\procA \xrightarrow{} \procB: \val} G' \cdot \beta) \wproj_{\Alphabet_{\procA}} = v \wproj_{\Alphabet_{\procA}}$. 
By subprotocol fidelity, $\SyncToAsync(\trace(\alpha \cdot G \xrightarrow{\procA \xrightarrow{} \procB: \val} G' \cdot \beta))$ is a trace in $\subsetprojholew{A}$. 
Consequently, $\SyncToAsync(\trace(\alpha \cdot G \xrightarrow{\procA \xrightarrow{} \procB: \val} G' \cdot \beta)) \wproj_{\Alphabet_{\procA}}$ is a prefix in $A$. 
We find a contradiction from the fact that $A$ is deterministic and there is no outgoing transition labeled $\snd{\procA}{\procB}{\val}$ from $s$. 
Similar arguments can be used to show the necessity of \receiveExhaustive. 
Finally, for \finalVal, in the case that $s$ is non-final in $A$ but contains a final state in $\semglobal(\GG)$, we can instantiate \cref{lm:hack-for-subtype} with this final state and show that $u$ evidences a deadlock. 

\begin{lemma}
	\label{lm:stepping-stone-completeness-rest}
If $A$ violates \sendPreservation, \receiveExhaustive or \finalVal, then it does not hold that for all well-behaved contexts $\CSMhole$, $\CSMholew{A}$ refines  $\subsetprojholew{A}$. 
\end{lemma}

 \subsection{\emph{\ProblemTwo} (General Case)}

Equipped with the solution to a special case, we are ready to revisit the general case of \ProblemTwo, which asks to find a \CharacterizationTwo that satisfies the following: 
\begin{center}
	\CharacterizationTwo $\Leftrightarrow$ for all well-behaved contexts $\CSMhole$, 
	$\CSMholew{A}$ refines
	$\CSMholew{B}$. 
\end{center}

Critical to the former problems is the fact that the state decoration function precisely captures those states in a local state machine that are reachable in some CSM execution, under some assumptions on the context: a state is reachable if and only if its decoration set is non-empty. 
This allows the conditions in \CharacterizationOne and \CharacterizationTwoPrime to precisely characterize the reachable local states. 

The second problem generalizes the subset projection to an arbitrary state machine $B$, and asks whether a candidate state machine $A$ (the subtype) refines $B$ (the supertype) in any well-behaved context.
Unfortunately, we cannot simply decorate the subtype with the supertype's states, because not all states in the supertype are reachable. 
Instead, we need to restrict the set of states in the supertype to those that themselves have non-empty decoration sets with respect to $\GG$. 

In the remainder of this section, let $\procA \in \Procs$ be a role, let $B = (Q_B, \Alphabet_\procA, t_0, \delta_B, F_B)$ denote the supertype state machine for $\procA$, and let $A = (Q_A, \Alphabet_\procA, s_0, \delta_A, F_A)$ denote the subtype state machine for $\procA$. 
We modify our state decoration function in~\cref{def:state-decoration} to map states of $A$ to subsets of states in $B$ that themselves have non-empty decoration sets with respect to $\GG$. 
\begin{definition}[State decoration with respect to a supertype]
	\label{def:state-decoration-supertype}
	Let $\GG$ be a global type. 
	Let $\procA \in \Procs$ be a role, and let $B = (Q_B, \Alphabet_\procA, t_0, \delta_B, F_B)$ and $A = (Q_A, \Alphabet_\procA, s_0, \delta_A, F_A)$ be two deterministic finite state machines for $\procA$.
	We define a total function $d_{\GG,B,A} : Q' \rightarrow \powersetof{Q}$ that maps each state in $A$ to a subset of states in $B$ such that:
\[
	d_{\GG,B,A}(s) = 
	\{ t \in Q_B
	\mid 
	\exists u \in \Alphabet_\procA^*.~ 
	s_0 \xrightarrow{u}\mathrel{\vphantom{\to}^*} s \in \delta_A \land 
	t_0 \xrightarrow{u}\mathrel{\vphantom{\to}^*} t \in \delta_B \land d(t) \neq \emptyset
	\}
	\]
\end{definition}
We again omit the subscripts $\GG$ and $A$ when clear from context, but retain the subscript $B$ to distinguish $d_{B}$ from $d$ in~\cref{def:state-decoration}.

We likewise require a generalization of $\transAnnoFunc$ and $\transAnnoFunDest$ to be defined in terms of~$B$, instead of the projection by erasure automaton for $\procA$.\begin{definition}[Transition origin and destination with respect to a supertype]
	Let $\GG$ be a global type, and let $B = (Q_B, \Alphabet_\procA, t_0, \delta_B, F_B)$ be a state machine. 
For $x \in \Alphabet_\procA$ and $s,s' \subseteq Q_B$,
	we define the set of \emph{transition origins} $\transAnnoFunc(s \xrightarrow{x} s')$ and \emph{transition destinations} $\transAnnoFunDest(s \xrightarrow{x} s')$ as follows:
\begin{align*}
		\transAnnoFunc_B(s \xrightarrow{x} s')
		\is {} &
		\set{t \in s
			\mid
			\exists t' \in s'. \,
			t \xrightarrow{x}\mathrel{\vphantom{\to}^*} t' \in \delta_B} \; \text{ and }\\
		\transAnnoFunDest_B(s \xrightarrow{x} s')
		\is {} &
		\set{t' \in s'
			\mid
			\exists t \in s. \,
			t \xrightarrow{x}\mathrel{\vphantom{\to}^*} t' \in \delta_B} \enspace.
	\end{align*}
\end{definition}

We present \CharacterizationTwo in terms of the newly defined decoration function $d_B$.

\newcommand{\sendDecSubVal}{\emph{Send Decoration Subtype Validity}\xspace}
\newcommand{\receiveDecSubVal}{\emph{Receive Decoration Subtype Validity}\xspace}
\newcommand{\sendPreservationSub}{\emph{Send Subtype Preservation}\xspace} 
\newcommand{\receiveExhaustiveSub}{\emph{Receive Subtype Exhaustivity}\xspace} 
\newcommand{\finalSubVal}{\emph{Final State Subtype Validity}\xspace}

\begin{definition}[\CharacterizationTwo]
	\label{def:characterization-three}
	Let $\GG$ be a global type, $\procA \in \Procs$ be a role, and $B = (Q_B, \Alphabet_\procA, t_0, \delta_B, F_B)$ and $A = (Q_A, \Alphabet_\procA, s_0, \delta_A, F_A)$ be two deterministic state machines for $\procA$.
\CharacterizationTwo is the conjunction of the following conditions:\begin{itemize}
		\item 
		\sendDecSubVal: every send transition $s \xrightarrow{x} s' \in \delta_A$ must be enabled in all states of $B$ decorating $s$: \\
		$\forall s \xrightarrow{\snd{\procA}{\procB}{\val}} s' \in \delta_A.~\transAnnoFunc_B(d_{B}(s)  \xrightarrow{\snd{\procA}{\procB}{\val}} d_{B}(s')) = d_B(s)$. \\
		\item 
		\receiveDecSubVal: no receive transition is enabled in an alternative continuation originating from the same state: \\
		$\begin{array}{l}\forall s \xrightarrow{\rcv{\procB_1}{\procA}{\val_1}} s_1,~s \xrightarrow{x} s_2 \in \delta_A.~x \neq \rcv{\procB_1}{\procA}{\_} \implies \\ \qquad
		\forall G \in \underset{t \in d_{B}(s_2)} \bigcup \{d(t) \mid t \in \transAnnoFunDest_B(d_B(s) \xrightarrow{x} d_B(s_2))\}.~ 
		\snd{\procB_1}{\procA}{\val_1} \notin \semavail^{\procA}_{(G \ldots)} .
                \end{array}$ \\
		\item 
		\sendPreservationSub: every state decorated by a send-originating global state must have at least one outgoing send transition: \\ 
		$\forall s \in Q_A.~(\underset{t \in d_{B}(s)} \bigcup d(t) \cap Q_{\GG,!} \neq \emptyset \,) \implies
		\exists x \in \Alphabet_{\procA,!},\,s' \in Q_A.~ s \xrightarrow{x} s' \in \delta_A$. \\
		\item 
		\receiveExhaustiveSub: every receive transition that is enabled in some global state decorating $s$ must be an outgoing transition from $s$: \\
		$\forall s \in Q_A.~\forall G \xrightarrow{x}\mathrel{\vphantom{\to}^*} G' \in \projerasuretrans.~G \in \underset{t \in d_B(s)} \bigcup d(t) \implies 
\exists s' \in Q_A.~s \xrightarrow{x} s' \in \delta_A$. \\
		\item \label{cond:final-states-final} 
		\finalVal: a reachable state is final if its decorating set contains a final global state: \\		
		$\forall s \in Q_A.~\underset{t \in d_{B}(s)} \bigcup d(t) \neq \emptyset \implies 
		(\underset{t \in d_{B}(s)} \bigcup d(t) \cap F_\GG \neq \emptyset \,) \implies s \in F_A$.
	\end{itemize}
\end{definition}

\newcommand{\subCSM}{$\CSMholew{A}$\xspace}
\newcommand{\superCSM}{$\CSMholew{B}$\xspace}

\noindent
We want to show the following equivalence to prove \cref{thm:equivalence-two}:
\begin{center}
	\CharacterizationTwo $\Leftrightarrow$ for all well-behaved contexts $\CSMhole$, \subCSM refines \superCSM. 
\end{center}

\begin{restatable}[Soundness of \CharacterizationTwo]{lemma}{soundnessCharacterizationTwo}
	\label{lm:soundness-characterization-two}
	If \CharacterizationTwo holds, then for all well-behaved contexts~$\CSMhole$, \subCSM refines
	\superCSM. 
\end{restatable}

Predictably, the proof of soundness is directly adapted from the proof for \CharacterizationTwoPrime by applying suitable ``liftings'', and can be found in~\appendixRef{app:proofs}.

\begin{lemma}[Completeness of \CharacterizationTwo]
  If for all well-behaved contexts $\CSMhole$, 
	$\CSMholew{A}$ refines
	$\CSMholew{B}$, then \CharacterizationTwo holds. 
\end{lemma}

Again, we prove the modus tollens of this implication, and we again are required to find a witness well-behaved context $\CSMhole$, such that \subCSM does not refine \superCSM under the assumption of the negation of \CharacterizationTwo. 
In the special case where $B$ is the subset construction automaton, we observed that any state in $A$ with a non-empty decoration set with respect to $\GG$ is reachable by the CSM consisting of $A$ and the subset construction context, denoted $\subsetprojholew{A}$. 
We were therefore able to use $\subsetprojholew{\cdot}$ as the witness well-behaved context. 
A similar characterization is true in the general case: a~state in~$A$ is reachable by $\subsetprojholew{A}$ if it has a non-empty decoration set with respect to~$B$.
This in turn depends on the fact that we only label states in $A$ with states in $B$ that themselves have non-empty decorating sets with respect to $\GG$. 
The following lemma lifts~\cref{lm:hack-for-subtype} to the general problem setting: 
\begin{lemma}
	\label{lm:hack-for-subtype-supertype}
	Let $A, B$ be two state machines for $\procA$, such that for all well-behaved contexts $\CSMhole$, \superCSM refines $\GG$. 
	Let $s$ be a state in $A$, and let $t$ be a state in $B$ such that $t \in d_B(s)$.  
	Let $u \in \Alphabet^*_\procA$ be a word such that
	$s_0 \xrightarrow{u} \mathrel{\vphantom{\to}^*} s$ in $A$.
	Then, there exists a run $\alpha \cdot G$ of $\semglobalsync(\GG)$ such that 
	$\SyncToAsync(\trace(\alpha \cdot G)) \wproj_{\Alphabet_\procA} = u$, 
	$\SyncToAsync(\trace(\alpha \cdot G))$ is a trace in both $\subsetprojholew{A}$ and $\subsetprojholew{B}$ and in the CSM configuration reached on $\SyncToAsync(\trace(\alpha \cdot G))$, $A$ is in state $s$. 
\end{lemma}
\begin{proof}
	From the fact that $t \in d_B(s)$ and~the definition of state decoration (\cref{def:state-decoration-supertype}), it holds that
	$d(t) \neq \emptyset$ and
	$t_0 \xrightarrow{u}\mathrel{\vphantom{\to}^*} t \in \delta_B$. 
	Let $G \in d(t)$. 
	We apply \cref{lm:hack-for-subtype} to obtain a run $\alpha \cdot G$ such that $\SyncToAsync(\trace(\alpha \cdot G)) \wproj_{\Alphabet_\procA} = u$, $\SyncToAsync(\trace(\alpha \cdot G))$ is a trace in $\subsetprojholew{B}$ and in the $\subsetprojholew{B}$ configuration reached on $\SyncToAsync(\trace(\alpha \cdot G))$, $B$ is in state $t$. 
	Because $s_0 \xrightarrow{u}\mathrel{\vphantom{\to}^*} s \in \delta_A$, and all non-$\procA$ state machines are identical from $\subsetprojholew{B}$ to $\subsetprojholew{A}$, it is clear that 
	$\SyncToAsync(\trace(\alpha \cdot G))$ is also a trace of $\subsetprojholew{A}$ and in the CSM configuration reached on $\SyncToAsync(\trace(\alpha \cdot G))$, $A$ is in state~$s$.
	\proofend
\end{proof}

\newcommand{\subsubsetprojCSM}{$\subsetprojholew{A}$\xspace}
\newcommand{\supersubsetprojCSM}{$\subsetprojholew{B}$\xspace}

Having found our witness well-behaved context $\subsetprojholew{\cdot}$, established \cref{lm:hack-for-subtype-supertype} to replace \cref{lm:hack-for-subtype}, and observed that the violation of each condition in \CharacterizationTwo likewise yields a state with a non-empty decoration set with respect to $B$, completeness then amounts to showing the existence of a $w \in \AlphAsync^*$ such that $w$ refutes subprotocol fidelity, language inclusion, or deadlock freedom.
Recall that the proofs for the necessity of \sendPreservation, \receiveExhaustive and \finalVal in the case where $B$ is the subset construction constructed a trace that refuted either subprotocol fidelity or deadlock freedom. 
These two properties are identical across both formulations of the problem, and therefore the construction can be wholly reused to show the necessity of \sendPreservationSub, \receiveExhaustiveSub and \finalSubVal. 

\begin{lemma}
If \subCSM violates \sendDecSubVal or \receiveDecSubVal, then it does not hold that for all well-behaved contexts $\CSMhole$, \subCSM refines \superCSM. 
\end{lemma}

The proofs for the necessity of \sendDecVal and \receiveDecVal, on the other hand, construct a word that is a trace in $\CSMholew{A}$ but not a trace in $\subsetprojholew{A}$. 
In the general case, we can show that the same construction is a trace in $\CSMholew{A}$ but not a trace in $\CSMholew{B}$. 
We omit the proofs to avoid redundancy.

\begin{lemma}
If $\CSM{A}$ violates \sendPreservationSub, \receiveExhaustiveSub, or \finalSubVal, then it does not hold that for all well-behaved contexts $\CSMhole$,
	\subCSM refines \superCSM. 
\end{lemma}

 \section{Complexity Analysis}

We complete our discussion with a complexity analysis of the two considered problems, building on the characterizations established in \cref{thm:equivalence-one} and \cref{thm:equivalence-two}.

For the \ProblemOne problem, let $m$ be the size of $\mathcal{A}$ and $n$ the size of~$\GG$. Moreover, let $A_\procA$ be the local implementation of some role $\procA$ in $\mathcal{A}$. Observe that the sets $d_\GG(s)$ for each state $s$ of $A_\procA$ as well as the sets $\semavail^{\procA}_{(G' \ldots)}$ for each subterm $G'$ of~$\GG$ are at most of size $n$. It is then easy to see that \CharacterizationOne can be checked in time polynomial in $n$ and~$m$, provided that the sets $d_\GG(s)$ and $\semavail^{\procA}_{(G' \ldots)}$ are also computable in polynomial~time.

To see this for the sets $\semavail^{\procA}_{(G' \ldots)}$, observe that the definition expands each occurrence of a recursion variable in $\GG$ at most once. So the traversal takes time $O(n^2)$. For each traversed event $\msgFromTo{\procA}{\procB}{\val}$ in $\GG$, we need to perform a constant number of lookup, insertion, and deletion operations on a set of size at most $n$, which takes time $O(\log n)$. The time for computing $\semavail^{\procA}_{(G' \ldots)}$ is thus in $O(n^2 \log n)$.

Similarly, observe that the function $d_\GG$ can be computed for the local implementation of each role $A_\procA \in \Procs$ using a simple fixpoint loop. Each set $d_\GG(s)$ can be represented as a bit vector of size $n$, making all set operations constant time. The loop inserts at most $n$ subterms of $\GG$ into each $d_\GG(s)$, which takes time $O(mn)$ for all insertions. Moreover, for each $G$ inserted into a set $d_\GG(s)$ and each transition $s \xrightarrow{x}\mathrel{\vphantom{\to}} s'$ in~$A_\procA$, we need to compute the set $\{G' \mid G \xrightarrow{x}\mathrel{\vphantom{\to}^*} G' \in \projerasuretrans \}$ which is then added to $d_\GG(s')$. Computing these sets takes time $O(mn)$ for each $G$ and $s$.

Following analogous reasoning, we can also establish that \CharacterizationTwo is checkable in polynomial time.

\begin{theorem}
  The \ProblemOne and \ProblemTwo problems are decidable in polynomial time.
\end{theorem}

 \section{Related Work}
\label{sec:related}
Session types were first introduced in binary form by Honda in 1993~\cite{DBLP:conf/concur/Honda93}. 
Binary session types describe interactions between two participants, and communication safety of binary sessions amounts to channel duality. 
Binary session types were generalized to multiparty session types -- describing interactions between more than two participants -- by Honda, Yoshida and Carbone in 2008~\cite{DBLP:conf/popl/HondaYC08}, and the corresponding notion of safety was generalized from duality to multiparty consistency. 
Binary session types were inspired by and enjoy a close connection to linear logic~\cite{DBLP:journals/tcs/Girard87, DBLP:journals/jfp/Wadler14, DBLP:journals/mscs/CairesPT16}.
Horne generalizes this connection to multiparty session types and non-commutative extensions of linear logic~\cite{DBLP:conf/concur/Horne20}. The connection between multiparty session types and logic is also explored in \cite{DBLP:conf/concur/CarboneLMSW16, DBLP:conf/forte/CairesP16, DBLP:journals/acta/CarboneMSY17}. 
MSTs have since been extensively studied and widely adopted in practical programming languages; we refer the reader to \cite{DBLP:conf/sfm/CoppoDPY15} for a comprehensive survey. 

\myparagraph{Session type syntax.}
Session type frameworks have enjoyed various extensions since their inception. 
In particular, the choice operator for both global and local types has received considerable attention over the years. 
MSTs were originally introduced as global types, with a \emph{directed} choice operator that restricted a sender to sending different messages to the same recipient. \cite{DBLP:journals/corr/abs-1203-0780} and \cite{DBLP:conf/concur/MajumdarMSZ21} relax this restriction to \emph{sender-driven choice}, which allows a sender to send different messages to different recipients, and increases the expressivity of global types. Our paper targets global types with sender-driven choice. 
For local types, a direct comparison can be drawn to the $\pi$-calculus, for which \emph{mixed choice} was shown to be strictly more expressive than \emph{separate choice}~\cite{DBLP:journals/mscs/Palamidessi03}. 
Mixed choices allow both send and receive actions, whereas separate choices consist purely of either sends or receives. 
\cite{DBLP:conf/cav/LiSWZ23} showed that any global type with sender-driven choice can be implemented by a CSM with only separate choice. 
Mixed choice for binary local types was investigated in~\cite{DBLP:journals/tcs/CasalMV22}, although~\cite{DBLP:journals/corr/abs-2209-06819} later showed that this variant falls short of the full expressive power of mixed choice $\pi$-calculus, and instead can only express separate choice $\pi$-calculus. 
Other communication primitives have also been studied, such as channel delegation~\cite{DBLP:conf/popl/HondaYC08, DBLP:conf/esop/HondaVK98, DBLP:journals/tcs/CastellaniDGH20}, dependent predicates~\cite{DBLP:conf/ppdp/ToninhoCP11,DBLP:/conf/ppdp/ToninhoCP21}, parametrization~\cite{DBLP:journals/corr/abs-1208-6483, DBLP:journals/scp/CharalambidesDA16} and data refinement~\cite{DBLP:journals/pacmpl/ZhouFHNY20}.

\myparagraph{Session type semantics.} 
MSTs were introduced in~\cite{DBLP:conf/popl/HondaYC08} with a process algebra semantics. The connection to CSMs was established in~\cite{DBLP:conf/esop/DenielouY12}, which defines a class of CSMs whose state machines can be represented as local types, called \emph{Communicating Session Automata} (CSA). CSAs inherit from the local types they represent restrictions on choice discussed above, ``tree-like'' restrictions on the structure (see~\cite{DBLP:conf/ecoop/Stutz23} for a characterization), and restrictions on outgoing transitions from final states. The CSM implementation model in our work assumes none of the above restrictions, and is thus true to its name.

\myparagraph{Session subtyping.} 
Session subtyping was first introduced by~\cite{DBLP:journals/acta/GayH05} in the context of the $\pi$-calculus, which was in turn inspired by Pierce and Sangiorgi's work on subtyping for channel endpoints~\cite{DBLP:journals/mscs/PierceS96}. 
The session types literature distinguishes between two notions of subtyping based on the network assumptions of the framework: \emph{synchronous} and \emph{asynchronous subtyping}.
Both notions respect Liskov and Wing's substitution principle~\cite{DBLP:journals/toplas/LiskovW94}, but differ in the guarantees provided. 
We discuss each in turn. 

Synchronous subtyping follows the notions of covariance and contravariance introduced by \cite{DBLP:journals/acta/GayH05}, and checks that a subtype contains fewer sends and more receives than its supertype. 
For binary synchronous session types, Lange and Yoshida~\cite{DBLP:conf/tacas/LangeY16} show that subtyping can be decided in quadratic time via model checking of a characteristic formulae in the modal $\mu$-calculus. 
For multiparty synchronous session types, Ghilezan et al.~\cite{DBLP:journals/jlamp/GhilezanJPSY19} present a precise subtyping relation that is universally quantified over all contexts, and restricts the local type syntax to directed choice.
As mentioned in \cref{sec:intro}, \cite{DBLP:journals/jlamp/GhilezanJPSY19}, their subtyping relation is incomplete when generalized to asynchronous multiparty sessions with directed choice.
As discussed in \cref{sec:motivation}, their subtyping relation is further incomplete when generalized to asynchronous multiparty sessions with mixed choice, due to the ``peculiarity [...] that, apart from a pair of inactive session types, only inputs and outputs from/to a same participant can be related''~\cite{DBLP:journals/jlamp/GhilezanJPSY19}.
The complexity of the subtyping relation in~\cite{DBLP:journals/jlamp/GhilezanJPSY19} is not mentioned.

Unlike subtyping relations for synchronous sessions which preserve language inclusion, subtyping relations for asynchronous sessions instead focus on deadlock-free optimizations that permute roles' local order of send and receive actions, also called \emph{asynchronous message reordering}, or AMR~\cite{DBLP:conf/ppopp/CutnerYV22}.
First proposed for binary sessions by Mostrous and Yoshida~\cite{DBLP:conf/tlca/MostrousY09}, and for multiparty sessions by Mostrous et al.~\cite{DBLP:conf/esop/MostrousYH09}, this notion of subtyping does not satisfy subprotocol fidelity in general; indeed, in some cases, the set of behaviors recognized by a supertype is entirely disjoint from that of its subtype~\cite{DBLP:journals/lmcs/BravettiCLYZ21}. 
Asynchronous subtyping was shown to be undecidable for both binary and multiparty session types~\cite{DBLP:conf/fossacs/LangeY17,DBLP:journals/tcs/BravettiCZ18}.
Existing works are thus either restricted to binary protocols \cite{DBLP:conf/fossacs/LangeY17, DBLP:journals/lmcs/BravettiCLYZ21, DBLP:journals/tcs/BravettiCZ18, DBLP:conf/coordination/BacchianiBLZ21}, prohibit non-deterministic choice involving multiple receivers \cite{DBLP:journals/pacmpl/GhilezanPPSY21,DBLP:conf/fossacs/BravettiLZ21}, or make strong fairness assumptions on the network \cite{DBLP:conf/fossacs/BravettiLZ21}.

The connection between session subtyping and behavioral contract refinement has been studied only in the context of binary session types, and is thus out of scope of our work. We refer the reader to~\cite{DBLP:journals/jlamp/GhilezanJPSY19} for a survey.

\subsubsection*{Acknowledgements}
The authors thank Damien Zufferey for discussions and feedback.
This work is funded in parts by the National Science Foundation under grant CCF-2304758.
Felix Stutz was supported by the Deutsche Forschungsgemeinschaft project 389792660 TRR 248—CPEC.

\phantomsection\label{paper-last-page}

\ifoptionfinal
{}
{
    \clearpage
}

\bibliographystyle{splncs04}

\begin{thebibliography}{10}
\providecommand{\url}[1]{\texttt{#1}}
\providecommand{\urlprefix}{URL }
\providecommand{\doi}[1]{https://doi.org/#1}

\bibitem{DBLP:conf/coordination/BacchianiBLZ21}
Bacchiani, L., Bravetti, M., Lange, J., Zavattaro, G.: A session subtyping
  tool. In: Damiani, F., Dardha, O. (eds.) Coordination Models and Languages -
  23rd {IFIP} {WG} 6.1 International Conference, {COORDINATION} 2021, Held as
  Part of the 16th International Federated Conference on Distributed Computing
  Techniques, DisCoTec 2021, Valletta, Malta, June 14-18, 2021, Proceedings.
  Lecture Notes in Computer Science, vol. 12717, pp. 90--105. Springer (2021).
  \doi{10.1007/978-3-030-78142-2\_6},
  \url{https://doi.org/10.1007/978-3-030-78142-2\_6}

\bibitem{DBLP:journal/mscs/BarbaneraD15}
Barbanera, F., De'Liguoro, U.: Sub-behaviour relations for session-based
  client/server systems. Mathematical Structures in Computer Science
  \textbf{25}(6),  1339–1381 (2015). \doi{10.1017/S096012951400005X}

\bibitem{DBLP:journals/mscs/BernardiH16}
Bernardi, G.T., Hennessy, M.: Modelling session types using contracts. Math.
  Struct. Comput. Sci.  \textbf{26}(3),  510--560 (2016).
  \doi{10.1017/S0960129514000243},
  \url{https://doi.org/10.1017/S0960129514000243}

\bibitem{DBLP:journals/jacm/BrandZ83}
Brand, D., Zafiropulo, P.: On communicating finite-state machines. J. {ACM}
  \textbf{30}(2),  323--342 (1983). \doi{10.1145/322374.322380},
  \url{https://doi.org/10.1145/322374.322380}

\bibitem{DBLP:journals/lmcs/BravettiCLYZ21}
Bravetti, M., Carbone, M., Lange, J., Yoshida, N., Zavattaro, G.: A sound
  algorithm for asynchronous session subtyping and its implementation. Log.
  Methods Comput. Sci.  \textbf{17}(1) (2021),
  \url{https://lmcs.episciences.org/7238}

\bibitem{DBLP:journals/tcs/BravettiCZ18}
Bravetti, M., Carbone, M., Zavattaro, G.: On the boundary between decidability
  and undecidability of asynchronous session subtyping. Theor. Comput. Sci.
  \textbf{722},  19--51 (2018). \doi{10.1016/j.tcs.2018.02.010},
  \url{https://doi.org/10.1016/j.tcs.2018.02.010}

\bibitem{DBLP:conf/fossacs/BravettiLZ21}
Bravetti, M., Lange, J., Zavattaro, G.: Fair refinement for asynchronous
  session types. In: Kiefer, S., Tasson, C. (eds.) Foundations of Software
  Science and Computation Structures - 24th International Conference, {FOSSACS}
  2021, Held as Part of the European Joint Conferences on Theory and Practice
  of Software, {ETAPS} 2021, Luxembourg City, Luxembourg, March 27 - April 1,
  2021, Proceedings. Lecture Notes in Computer Science, vol. 12650, pp.
  144--163. Springer (2021). \doi{10.1007/978-3-030-71995-1\_8},
  \url{https://doi.org/10.1007/978-3-030-71995-1\_8}

\bibitem{DBLP:journals/sefm/BravettiZ19}
Bravetti, M., Zavattaro, G.: Relating session types and behavioural contracts:
  The asynchronous case. In: {\"O}lveczky, P.C., Sala{\"u}n, G. (eds.) Software
  Engineering and Formal Methods. pp. 29--47. Springer International
  Publishing, Cham (2019)

\bibitem{DBLP:journals/ssm/BravettiZ21}
Bravetti, M., Zavattaro, G.: Asynchronous session subtyping as communicating
  automata refinement. Softw. Syst. Model.  \textbf{20}(2),  311–333 (apr
  2021). \doi{10.1007/s10270-020-00838-x},
  \url{https://doi.org/10.1007/s10270-020-00838-x}

\bibitem{DBLP:conf/forte/CairesP16}
Caires, L., P{\'{e}}rez, J.A.: Multiparty session types within a canonical
  binary theory, and beyond. In: Albert, E., Lanese, I. (eds.) Formal
  Techniques for Distributed Objects, Components, and Systems - 36th {IFIP}
  {WG} 6.1 International Conference, {FORTE} 2016, Held as Part of the 11th
  International Federated Conference on Distributed Computing Techniques,
  DisCoTec 2016, Heraklion, Crete, Greece, June 6-9, 2016, Proceedings. Lecture
  Notes in Computer Science, vol.~9688, pp. 74--95. Springer (2016).
  \doi{10.1007/978-3-319-39570-8\_6},
  \url{https://doi.org/10.1007/978-3-319-39570-8\_6}

\bibitem{DBLP:journals/mscs/CairesPT16}
Caires, L., Pfenning, F., Toninho, B.: Linear logic propositions as session
  types. Math. Struct. Comput. Sci.  \textbf{26}(3),  367--423 (2016).
  \doi{10.1017/S0960129514000218},
  \url{https://doi.org/10.1017/S0960129514000218}

\bibitem{DBLP:conf/concur/CarboneLMSW16}
Carbone, M., Lindley, S., Montesi, F., Sch{\"{u}}rmann, C., Wadler, P.:
  Coherence generalises duality: {A} logical explanation of multiparty session
  types. In: Desharnais, J., Jagadeesan, R. (eds.) 27th International
  Conference on Concurrency Theory, {CONCUR} 2016, August 23-26, 2016,
  Qu{\'{e}}bec City, Canada. LIPIcs, vol.~59, pp. 33:1--33:15. Schloss Dagstuhl
  - Leibniz-Zentrum f{\"{u}}r Informatik (2016).
  \doi{10.4230/LIPIcs.CONCUR.2016.33},
  \url{https://doi.org/10.4230/LIPIcs.CONCUR.2016.33}

\bibitem{DBLP:journals/acta/CarboneMSY17}
Carbone, M., Montesi, F., Sch{\"{u}}rmann, C., Yoshida, N.: Multiparty session
  types as coherence proofs. Acta Informatica  \textbf{54}(3),  243--269
  (2017). \doi{10.1007/s00236-016-0285-y},
  \url{https://doi.org/10.1007/s00236-016-0285-y}

\bibitem{DBLP:journals/tcs/CasalMV22}
Casal, F., Mordido, A., Vasconcelos, V.T.: Mixed sessions. Theor. Comput. Sci.
  \textbf{897},  23--48 (2022). \doi{10.1016/j.tcs.2021.08.005},
  \url{https://doi.org/10.1016/j.tcs.2021.08.005}

\bibitem{DBLP:journals/corr/abs-1203-0780}
Castagna, G., Dezani{-}Ciancaglini, M., Padovani, L.: On global types and
  multi-party session. Log. Methods Comput. Sci.  \textbf{8}(1) (2012).
  \doi{10.2168/LMCS-8(1:24)2012},
  \url{https://doi.org/10.2168/LMCS-8(1:24)2012}

\bibitem{DBLP:journals/toplas/CastagnaGP09}
Castagna, G., Gesbert, N., Padovani, L.: A theory of contracts for web
  services. {ACM} Trans. Program. Lang. Syst.  \textbf{31}(5),  19:1--19:61
  (2009). \doi{10.1145/1538917.1538920},
  \url{https://doi.org/10.1145/1538917.1538920}

\bibitem{DBLP:journals/tcs/CastellaniDGH20}
Castellani, I., Dezani{-}Ciancaglini, M., Giannini, P., Horne, R.: Global types
  with internal delegation. Theor. Comput. Sci.  \textbf{807},  128--153
  (2020). \doi{10.1016/j.tcs.2019.09.027},
  \url{https://doi.org/10.1016/j.tcs.2019.09.027}

\bibitem{DBLP:journals/scp/CharalambidesDA16}
Charalambides, M., Dinges, P., Agha, G.A.: Parameterized, concurrent session
  types for asynchronous multi-actor interactions. Sci. Comput. Program.
  \textbf{115-116},  100--126 (2016). \doi{10.1016/j.scico.2015.10.006},
  \url{https://doi.org/10.1016/j.scico.2015.10.006}

\bibitem{DBLP:conf/sfm/CoppoDPY15}
Coppo, M., Dezani{-}Ciancaglini, M., Padovani, L., Yoshida, N.: A gentle
  introduction to multiparty asynchronous session types. In: Bernardo, M.,
  Johnsen, E.B. (eds.) Formal Methods for Multicore Programming - 15th
  International School on Formal Methods for the Design of Computer,
  Communication, and Software Systems, {SFM} 2015, Bertinoro, Italy, June
  15-19, 2015, Advanced Lectures. Lecture Notes in Computer Science, vol.~9104,
  pp. 146--178. Springer (2015). \doi{10.1007/978-3-319-18941-3\_4},
  \url{https://doi.org/10.1007/978-3-319-18941-3\_4}

\bibitem{DBLP:conf/ppopp/CutnerYV22}
Cutner, Z., Yoshida, N., Vassor, M.: Deadlock-free asynchronous message
  reordering in rust with multiparty session types. In: Lee, J., Agrawal, K.,
  Spear, M.F. (eds.) PPoPP '22: 27th {ACM} {SIGPLAN} Symposium on Principles
  and Practice of Parallel Programming, Seoul, Republic of Korea, April 2 - 6,
  2022. pp. 246--261. {ACM} (2022). \doi{10.1145/3503221.3508404},
  \url{https://doi.org/10.1145/3503221.3508404}

\bibitem{DBLP:conf/esop/DenielouY12}
Deni{\'{e}}lou, P., Yoshida, N.: Multiparty session types meet communicating
  automata. In: Seidl, H. (ed.) Programming Languages and Systems - 21st
  European Symposium on Programming, {ESOP} 2012, Held as Part of the European
  Joint Conferences on Theory and Practice of Software, {ETAPS} 2012, Tallinn,
  Estonia, March 24 - April 1, 2012. Proceedings. Lecture Notes in Computer
  Science, vol.~7211, pp. 194--213. Springer (2012).
  \doi{10.1007/978-3-642-28869-2\_10},
  \url{https://doi.org/10.1007/978-3-642-28869-2\_10}

\bibitem{DBLP:journals/corr/abs-1208-6483}
Deni{\'{e}}lou, P., Yoshida, N., Bejleri, A., Hu, R.: Parameterised multiparty
  session types. Log. Methods Comput. Sci.  \textbf{8}(4) (2012).
  \doi{10.2168/LMCS-8(4:6)2012}, \url{https://doi.org/10.2168/LMCS-8(4:6)2012}

\bibitem{DBLP:journals/jalc/EllulKSW05}
Ellul, K., Krawetz, B., Shallit, J.O., Wang, M.: Regular expressions: New
  results and open problems. J. Autom. Lang. Comb.  \textbf{10}(4),  407--437
  (2005). \doi{10.25596/jalc-2005-407},
  \url{https://doi.org/10.25596/jalc-2005-407}

\bibitem{DBLP:journals/eik/EsparzaN94}
Esparza, J., Nielsen, M.: Decidability issues for petri nets - a survey. J.
  Inf. Process. Cybern.  \textbf{30}(3),  143--160 (1994)

\bibitem{DBLP:journals/acta/GayH05}
Gay, S.J., Hole, M.: Subtyping for session types in the pi calculus. Acta
  Informatica  \textbf{42}(2-3),  191--225 (2005).
  \doi{10.1007/s00236-005-0177-z},
  \url{https://doi.org/10.1007/s00236-005-0177-z}

\bibitem{DBLP:journals/jlamp/GhilezanJPSY19}
Ghilezan, S., Jakšić, S., Pantović, J., Scalas, A., Yoshida, N.: Precise
  subtyping for synchronous multiparty sessions. Journal of Logical and
  Algebraic Methods in Programming  \textbf{104},  127--173 (2019).
  \doi{https://doi.org/10.1016/j.jlamp.2018.12.002},
  \url{https://www.sciencedirect.com/science/article/pii/S2352220817302237}

\bibitem{DBLP:journals/pacmpl/GhilezanPPSY21}
Ghilezan, S., Pantovic, J., Prokic, I., Scalas, A., Yoshida, N.: Precise
  subtyping for asynchronous multiparty sessions. Proc. {ACM} Program. Lang.
  \textbf{5}({POPL}),  1--28 (2021). \doi{10.1145/3434297},
  \url{https://doi.org/10.1145/3434297}

\bibitem{DBLP:journals/tcs/Girard87}
Girard, J.: Linear logic. Theor. Comput. Sci.  \textbf{50},  1--102 (1987).
  \doi{10.1016/0304-3975(87)90045-4},
  \url{https://doi.org/10.1016/0304-3975(87)90045-4}

\bibitem{DBLP:conf/concur/Honda93}
Honda, K.: Types for dyadic interaction. In: Best, E. (ed.) {CONCUR} '93, 4th
  International Conference on Concurrency Theory, Hildesheim, Germany, August
  23-26, 1993, Proceedings. Lecture Notes in Computer Science, vol.~715, pp.
  509--523. Springer (1993). \doi{10.1007/3-540-57208-2\_35},
  \url{https://doi.org/10.1007/3-540-57208-2\_35}

\bibitem{DBLP:conf/esop/HondaVK98}
Honda, K., Vasconcelos, V.T., Kubo, M.: Language primitives and type discipline
  for structured communication-based programming. In: Hankin, C. (ed.)
  Programming Languages and Systems - ESOP'98, 7th European Symposium on
  Programming, Held as Part of the European Joint Conferences on the Theory and
  Practice of Software, ETAPS'98, Lisbon, Portugal, March 28 - April 4, 1998,
  Proceedings. Lecture Notes in Computer Science, vol.~1381, pp. 122--138.
  Springer (1998). \doi{10.1007/BFb0053567},
  \url{https://doi.org/10.1007/BFb0053567}

\bibitem{DBLP:conf/popl/HondaYC08}
Honda, K., Yoshida, N., Carbone, M.: Multiparty asynchronous session types. In:
  Necula, G.C., Wadler, P. (eds.) Proceedings of the 35th {ACM}
  {SIGPLAN-SIGACT} Symposium on Principles of Programming Languages, {POPL}
  2008, San Francisco, California, USA, January 7-12, 2008. pp. 273--284. {ACM}
  (2008). \doi{10.1145/1328438.1328472},
  \url{https://doi.org/10.1145/1328438.1328472}

\bibitem{DBLP:conf/concur/Horne20}
Horne, R.: Session subtyping and multiparty compatibility using circular
  sequents. In: Konnov, I., Kov{\'{a}}cs, L. (eds.) 31st International
  Conference on Concurrency Theory, {CONCUR} 2020, September 1-4, 2020, Vienna,
  Austria (Virtual Conference). LIPIcs, vol.~171, pp. 12:1--12:22. Schloss
  Dagstuhl - Leibniz-Zentrum f{\"{u}}r Informatik (2020).
  \doi{10.4230/LIPIcs.CONCUR.2020.12},
  \url{https://doi.org/10.4230/LIPIcs.CONCUR.2020.12}

\bibitem{DBLP:journals/cacm/Lamport78}
Lamport, L.: Time, clocks, and the ordering of events in a distributed system.
  Commun. {ACM}  \textbf{21}(7),  558--565 (1978). \doi{10.1145/359545.359563},
  \url{https://doi.org/10.1145/359545.359563}

\bibitem{DBLP:conf/tacas/LangeY16}
Lange, J., Yoshida, N.: Characteristic formulae for session types. In: Chechik,
  M., Raskin, J. (eds.) Tools and Algorithms for the Construction and Analysis
  of Systems - 22nd International Conference, {TACAS} 2016, Held as Part of the
  European Joint Conferences on Theory and Practice of Software, {ETAPS} 2016,
  Eindhoven, The Netherlands, April 2-8, 2016, Proceedings. Lecture Notes in
  Computer Science, vol.~9636, pp. 833--850. Springer (2016).
  \doi{10.1007/978-3-662-49674-9\_52},
  \url{https://doi.org/10.1007/978-3-662-49674-9\_52}

\bibitem{DBLP:conf/fossacs/LangeY17}
Lange, J., Yoshida, N.: On the undecidability of asynchronous session
  subtyping. In: Esparza, J., Murawski, A.S. (eds.) Foundations of Software
  Science and Computation Structures - 20th International Conference, {FOSSACS}
  2017, Held as Part of the European Joint Conferences on Theory and Practice
  of Software, {ETAPS} 2017, Uppsala, Sweden, April 22-29, 2017, Proceedings.
  Lecture Notes in Computer Science, vol. 10203, pp. 441--457 (2017).
  \doi{10.1007/978-3-662-54458-7\_26},
  \url{https://doi.org/10.1007/978-3-662-54458-7\_26}

\bibitem{DBLP:conf/cav/LangeY19}
Lange, J., Yoshida, N.: Verifying asynchronous interactions via communicating
  session automata. In: Dillig, I., Tasiran, S. (eds.) Computer Aided
  Verification - 31st International Conference, {CAV} 2019, New York City, NY,
  USA, July 15-18, 2019, Proceedings, Part {I}. Lecture Notes in Computer
  Science, vol. 11561, pp. 97--117. Springer (2019).
  \doi{10.1007/978-3-030-25540-4\_6},
  \url{https://doi.org/10.1007/978-3-030-25540-4\_6}

\bibitem{DBLP:conf/cav/LiSWZ23}
Li, E., Stutz, F., Wies, T., Zufferey, D.: Complete multiparty session type
  projection with automata. In: Enea, C., Lal, A. (eds.) Computer Aided
  Verification. pp. 350--373. Springer Nature Switzerland, Cham (2023)

\bibitem{DBLP:journals/toplas/LiskovW94}
Liskov, B., Wing, J.M.: A behavioral notion of subtyping. {ACM} Trans. Program.
  Lang. Syst.  \textbf{16}(6),  1811--1841 (1994). \doi{10.1145/197320.197383},
  \url{https://doi.org/10.1145/197320.197383}

\bibitem{DBLP:conf/concur/MajumdarMSZ21}
Majumdar, R., Mukund, M., Stutz, F., Zufferey, D.: Generalising projection in
  asynchronous multiparty session types. In: Haddad, S., Varacca, D. (eds.)
  32nd International Conference on Concurrency Theory, {CONCUR} 2021, August
  24-27, 2021, Virtual Conference. LIPIcs, vol.~203, pp. 35:1--35:24. Schloss
  Dagstuhl - Leibniz-Zentrum f{\"{u}}r Informatik (2021).
  \doi{10.4230/LIPIcs.CONCUR.2021.35},
  \url{https://doi.org/10.4230/LIPIcs.CONCUR.2021.35}

\bibitem{DBLP:conf/tlca/MostrousY09}
Mostrous, D., Yoshida, N.: Session-based communication optimisation for
  higher-order mobile processes. In: Curien, P. (ed.) Typed Lambda Calculi and
  Applications, 9th International Conference, {TLCA} 2009, Brasilia, Brazil,
  July 1-3, 2009. Proceedings. Lecture Notes in Computer Science, vol.~5608,
  pp. 203--218. Springer (2009). \doi{10.1007/978-3-642-02273-9\_16},
  \url{https://doi.org/10.1007/978-3-642-02273-9\_16}

\bibitem{DBLP:conf/esop/MostrousYH09}
Mostrous, D., Yoshida, N., Honda, K.: Global principal typing in partially
  commutative asynchronous sessions. In: Castagna, G. (ed.) Programming
  Languages and Systems, 18th European Symposium on Programming, {ESOP} 2009,
  Held as Part of the Joint European Conferences on Theory and Practice of
  Software, {ETAPS} 2009, York, UK, March 22-29, 2009. Proceedings. Lecture
  Notes in Computer Science, vol.~5502, pp. 316--332. Springer (2009).
  \doi{10.1007/978-3-642-00590-9\_23},
  \url{https://doi.org/10.1007/978-3-642-00590-9\_23}

\bibitem{DBLP:journals/mscs/Palamidessi03}
Palamidessi, C.: Comparing the expressive power of the synchronous and
  asynchronous pi-calculi. Math. Struct. Comput. Sci.  \textbf{13}(5),
  685--719 (2003). \doi{10.1017/S0960129503004043},
  \url{https://doi.org/10.1017/S0960129503004043}

\bibitem{DBLP:journals/corr/abs-2209-06819}
Peters, K., Yoshida, N.: On the expressiveness of mixed choice sessions. In:
  Castiglioni, V., Mezzina, C.A. (eds.) Proceedings Combined 29th International
  Workshop on Expressiveness in Concurrency and 19th Workshop on Structural
  Operational Semantics, {EXPRESS/SOS} 2022, and 19th Workshop on Structural
  Operational Semantics Warsaw, Poland, 12th September 2022. {EPTCS}, vol.~368,
  pp. 113--130 (2022). \doi{10.4204/EPTCS.368.7},
  \url{https://doi.org/10.4204/EPTCS.368.7}

\bibitem{DBLP:journals/mscs/PierceS96}
Pierce, B.C., Sangiorgi, D.: Typing and subtyping for mobile processes. Math.
  Struct. Comput. Sci.  \textbf{6}(5),  409--453 (1996).
  \doi{10.1017/s096012950007002x},
  \url{https://doi.org/10.1017/s096012950007002x}

\bibitem{DBLP:books/daglib/0086373}
Sipser, M.: Introduction to the theory of computation. {PWS} Publishing Company
  (1997)

\bibitem{DBLP:conf/ecoop/Stutz23}
Stutz, F.: Asynchronous multiparty session type implementability is decidable -
  lessons learned from message sequence charts. In: Ali, K., Salvaneschi, G.
  (eds.) 37th European Conference on Object-Oriented Programming, {ECOOP} 2023,
  July 17-21, 2023, Seattle, Washington, United States. LIPIcs, vol.~263, pp.
  32:1--32:31. Schloss Dagstuhl - Leibniz-Zentrum f{\"{u}}r Informatik (2023).
  \doi{10.4230/LIPIcs.ECOOP.2023.32},
  \url{https://doi.org/10.4230/LIPIcs.ECOOP.2023.32}

\bibitem{DBLP:conf/ppdp/ToninhoCP11}
Toninho, B., Caires, L., Pfenning, F.: Dependent session types via
  intuitionistic linear type theory. In: Schneider{-}Kamp, P., Hanus, M. (eds.)
  Proceedings of the 13th International {ACM} {SIGPLAN} Conference on
  Principles and Practice of Declarative Programming, July 20-22, 2011, Odense,
  Denmark. pp. 161--172. {ACM} (2011). \doi{10.1145/2003476.2003499},
  \url{https://doi.org/10.1145/2003476.2003499}

\bibitem{DBLP:/conf/ppdp/ToninhoCP21}
Toninho, B., Caires, L., Pfenning, F.: A decade of dependent session types. In:
  23rd International Symposium on Principles and Practice of Declarative
  Programming. PPDP 2021, Association for Computing Machinery, New York, NY,
  USA (2021). \doi{10.1145/3479394.3479398},
  \url{https://doi.org/10.1145/3479394.3479398}

\bibitem{DBLP:journals/jfp/Wadler14}
Wadler, P.: Propositions as sessions. J. Funct. Program.  \textbf{24}(2-3),
  384--418 (2014). \doi{10.1017/S095679681400001X},
  \url{https://doi.org/10.1017/S095679681400001X}

\bibitem{DBLP:journals/pacmpl/ZhouFHNY20}
Zhou, F., Ferreira, F., Hu, R., Neykova, R., Yoshida, N.: Statically verified
  refinements for multiparty protocols. Proceedings of the ACM on Programming
  Languages  \textbf{4},  1--30 (11 2020). \doi{10.1145/3428216}

\end{thebibliography}

\iftoggle{techrep}{
    \clearpage
    \appendix
    \section{Appendix}
\label{sec:appendix}

\subsection{Indistinguishability Relation \cite{DBLP:conf/concur/MajumdarMSZ21}}
\label{app:indist-rel}
We define a family of \emph{indistinguishability relations}
${\interswap_i} \subseteq \AlphAsync^* \times \AlphAsync^*$ for $i\geq 0$
as follows.
For all $w\in\Alphabet^*$, we have $w \interswap_0 w$.
For $i=1$, we define:
\vspace{-1ex}
\begin{enumerate}[label=(\arabic*)]
	\item
	If $\procA ≠ \procC$, then
	$
	w.\snd{\procA}{\procB}{\val}.\snd{\procC}{\procD}{\val'}.u
	\; \interswap_{1} \;
	w.\snd{\procC}{\procD}{\val'}.\snd{\procA}{\procB}{\val}.u
	$.
	
	\item
	If $\procB ≠ \procD$, then
	$
	w.\rcv{\procA}{\procB}{\val}.\rcv{\procC}{\procD}{\val'}.u
	\; \interswap_{1} \;
	w.\rcv{\procC}{\procD}{\val'}.\rcv{\procA}{\procB}{\val}.u
	$.
	
	\item
	If $\procA ≠ \procD \land (\procA ≠ \procC ∨ \procB ≠ \procD)
	$, then
	$
	w.\snd{\procA}{\procB}{\val}.\rcv{\procC}{\procD}{\val'}.u
	\; \interswap_{1} \;
	w.\rcv{\procC}{\procD}{\val'}.\snd{\procA}{\procB}{\val}.u
	$.
	\item
	If $\card{w \wproj_{\snd{\procA}{\procB}{\_}}} >
	\card{w \wproj_{\rcv{\procA}{\procB}{\_}}}$,
	then
	$
	w.\snd{\procA}{\procB}{\val}.\rcv{\procA}{\procB}{\val'}.u
	\; \interswap_{1} \;
	w.\rcv{\procA}{\procB}{\val'}.\snd{\procA}{\procB}{\val}.u
	$.
\end{enumerate}
Let $w, w', w''$ be sequences of events s.t.~$w \interswap_1 w'$ and $w' \interswap_i w''$ for some~$i$.
Then, $w \interswap_{i+1} w''$.
We define $w \interswap u$ if $w \interswap_n u$ for some $n$.

It is easy to see that $\interswap$ is an equivalence relation.
Define $u \preceq_\interswap v$ if there is $w\in\Sigma^*$ such that $u.w \interswap v$.
Observe that $u \interswap v$ iff
$u \preceq_\interswap v$ and $v \preceq_\interswap u$.

For infinite words $u, v\in\Sigma^\omega$, we define $u \preceq_\interswap^\omega v$ 
if for each finite prefix $u'$ of $u$, there is a finite prefix $v'$ of $v$ such that
$u' \preceq_\interswap v'$.
Define $u \interswap v$ iff $u \preceq_\interswap^\omega v$ and $v\preceq_\interswap^\omega u$.

We lift the equivalence relation $\interswap$ on $\Sigma^\infty$ to languages:
\[
\interswaplang(L) = \left\{ w' \mid \bigvee
\begin{array}{l}
	w' \in \Alphabet^* \land ∃ w ∈ \Alphabet^*. \; w \in L \text{ and } w' \interswap w \\
	w' ∈ \Alphabet^ω \land \exists w \in \Alphabet^\omega. \; w \in L \text{ and } w' \preceq_\interswap^\omega w
\end{array} \right\}
\]
For the infinite case, we take the downward closure w.r.t.~$\preceq_\interswap^\omega$.
Notice that the closure operator is asymmetric.
Consider the protocol $(\snd{\procA}{\procB}{\val}.\rcv{\procA}{\procB}{\val})^ω$.
Since we do not make any fairness assumption on scheduling, we need to include in the closure the execution where only the sender is scheduled, i.e., 
\[
	(\snd{\procA}{\procB}{\val})^ω \preceq_\interswap^\omega (\snd{\procA}{\procB}{\val}.\rcv{\procA}{\procB}{\val})^ω
	\enspace .
\]

\subsection{Proofs}
\label{app:proofs}

\impliesLocalLangIncl*
\begin{proof}
	First, we show that every trace in $\lang(\GG) \wproj_{\Alphabet_{\procA}}$ is a trace in $A_\procA$. 
	Let $u$ be a trace in $\lang(\GG) \wproj_{\Alphabet_{\procA}}$.
We proceed by induction on the length of $u$.
	In the base case, $u = \emptystring$, and $\emptystring$ is trivially a trace of every state machine.
In the induction step, let $ux$ be a prefix in $\lang(\GG) \wproj_{\Alphabet_{\procA}}$. 
From the induction hypothesis, we know that $u$ is a prefix in $\lang(A_\procA)$. 
	Let $s \in Q_\procA$ be the state reached on $u$ in $A_\procA$. 
	Because $ux$ is a prefix in $\lang(\GG) \wproj_{\Alphabet_{\procA}}$, there exists a run \mbox{$q_{0,\GG} \xrightarrow{u}\mathrel{\vphantom{\to}^*} G \xrightarrow{x}\mathrel{\vphantom{\to}^*} G'$} in the projection by erasure automaton for $\procA$.
	By the definition of state decoration, it holds that $G \in d_\GG(s)$. 
	By \transitionExhaustive, it holds that there exists a state $s' \in Q_\procA$ such that $s \xrightarrow{x} s' \in \delta_\procA$, and therefore $ux$ is also a prefix in $\lang(A_\procA)$.
	This concludes our proof by induction that every prefix in $\lang(\GG) \wproj_{\Alphabet_{\procA}}$ is a prefix in $\lang(A_\procA)$. 
	
	Let $w \in \lang(\GG) \wproj_{\Alphabet_{\procA}}$. 
	To show that $w \in \lang(A_\procA)$ for $w \in \AlphAsync^*$, it remains to show that $w$ reaches a final state in $A_\procA$. 
	Let $G'' \in F_\GG$ be the state reached on $w$ in the projection by erasure automaton, and let $s''$ be the state reached on $w$ in $A_\procA$. 
	By the state decoration function it holds that $G'' \in d_\GG(s'')$, and therefore by \finalVal, $s'' \in F_\procA$ and $w$ is a word in $\lang(A_\procA)$.
	The case for $w \in \AlphAsync^\infty$ follows from the fact that every trace of $\lang(\GG) \wproj_{\Alphabet_{\procA}}$ is a trace of $\lang(A_\procA)$ and the fact that $A_\procA$ is deterministic. 
	\proofend
\end{proof}

\aboutReceiveDecorationValidity*
\begin{proof}
	Suppose by contradiction that $x \neq \SyncToAsync(l) \wproj_{\Alphabet_{\procB}}$. 
	By the definition of unique splittings, $\procB$ is the active role in $l$. 
	We proceed by case analysis on $l$: (1) either $l$ is of the form $\procC \xrightarrow{} \procB: \val'$, with $\procC$ sending $\procB$ a different message $\val' \neq \val$, or (2) $l$ is of the form $\procD \xrightarrow{} \procB: \val$, with a different role $\procD \neq \procC$ sending $\procB$ a message, or $l$ is of the form $\procB \xrightarrow \_: \_$, with $\procB$ sending a message. We prove a contradiction in each case.

	First, we establish a claim that is used in both cases, and relies only on the fact that $\run$ is consistent with $w$ and $wx$ is a trace of $\mathcal{A}$. 
	
Let $\run_\procB$ denote the largest consistent prefix of $\run$ for $\procB$; it is clear that 
	$\run_\procB = \alpha \cdot G$. Formally,
	\[
	\run_\procB = max\{\run'~|~\run' \leq \run ~\land~ \bigl(
	\SyncToAsync(\trace(\run'))
	\bigr)
	\wproj_{\Alphabet_\procB} \preforder  w\wproj_{\Alphabet_\procB}
	\}\enspace.
	\]
	Let $\run_\procC$ be defined analogously.
	
	\noindent
	\textit{Claim:} $\run_\procB < \run_\procC$.
	Intuitively, $\procA$ is ahead of $\procB$ in $\run$ due to the half-duplex property of CSMs and the fact that $\procC$ is the sender. 
	Formally, \cite[Lemma~19]{DBLP:conf/concur/MajumdarMSZ21} implies $\xi(\procC,\procB) = u$ where
	$\mathcal{V}(w \wproj_{\snd{\procC}{\procB}{\_}}) = \MsgVals(w \wproj_{\rcv{\procC}{\procB}{\_}}).u$.
	Because $\xi(\procC, \procB)$ contains at least $\val$ by assumption,
	$|\MsgVals(w\wproj_{\snd{\procC}{\procB}{\_}})| > |\MsgVals(w \wproj_{\rcv{\procC}{\procB}{\_}})|$.
	Because 
	$\MsgVals(w \wproj_{\rcv{\procC}{\procB}{\_}}) < \MsgVals(w\wproj_{\snd{\procC}{\procB}{\_}})$ 
	and traces of CSMs are channel-compliant \cite[Lemma~19]{DBLP:conf/concur/MajumdarMSZ21}, it holds that
	$\run_\procC$ contains all 
	$\card{\MsgVals(w \wproj_{\rcv{\procC}{\procB}{\_}})}$
	transition labels of the form  
	$\procC \xrightarrow{} \procB: \_$
	that are contained in $\run_\procC$, 
	plus at least one more of the form 
	$\procC \xrightarrow{} \procB: \val$. 
	Because both $\run_\procB$ and $\run_\procC$ are prefixes of $\run$, 
	it must be the case that 
	$\run_\procB < \run_\procC$. 
	This concludes the proof of the above claim. 
	
	\noindent
	\textit{Case:} $l = \procC \xrightarrow{} \procB: \val'$ and $\val' \neq \val$. 
	We discharge this case by showing a contradiction to the fact that $\val$ is at the head of the channel between $\procC$ and $\procB$. 
	
	Because $\alpha \cdot G \leq \run_\procB$ and $\run_\procB < \run_\procC$ from the claim above, it must be the case that
	$\alpha \cdot G \xrightarrow{l} G' \leq  \run_\procC$
	and
	$\snd{\procC}{\procB}{\val'}$ is in $w \wproj_{\Alphabet_{\procC}}$.
	From \cite[Lemma 19]{DBLP:conf/concur/MajumdarMSZ21}, it follows that
	$\mathcal{V}(w \wproj_{\snd{\procC}{\procB}{\_}}) = \MsgVals(w \wproj_{\rcv{\procC}{\procB}{\_}}).\val'.u'$ and
	$\xi(\procC,\procB)  =  \val'.u'$, i.e. $\val'$ is at the head of the channel between $\procC$ and $\procB$. 
	We reach a contradiction. 
	
	\noindent
	\textit{Case:} $\forall \val'.~l \neq \procC \xrightarrow{} \procB: \val'$. 
	It follows that $\SyncToAsync(l) \wproj_{\Alphabet_{\procB}} \neq \rcv{\procC}{\procB}{\val'}$ for any $\val'$. 
We discharge this case by showing that
	\[
	\snd{\procC}{\procB}{\val} \in \semavail^{\procB}_{(G' \ldots)}\enspace.
	\]
	Recall that $\alpha \cdot G \xrightarrow{l} G' \leq \run_\procC$. 
	Then, there exists a transition labeled $\procC \xrightarrow{} \procB: \val$ that occurs in the suffix $G' \cdot \beta$. 
	Let 
	$G_0 \xrightarrow{\procC \xrightarrow{} \procB: \val} G_0'$
	be the earliest occurrence of such a transition in the suffix, then:
	\[
	\run_\procC = \alpha \cdot G \xrightarrow{l} G' \ldots G_0 \xrightarrow{\procC \xrightarrow{} \procB: \val} G_0'\dots\enspace.
	\]
	Note that $G_0$ must be a syntactic subterm of $G'$. 
	In order for $\snd{\procC}{\procB}{\val} \in \semavail^{\procB}_{(G' \ldots)}$ to hold, it suffices to show that $\procC \notin \blockedset$ in the recursive call to $\semavail^{\blockedset}_{(G' \dots)}$. 
	We argue this from the definition of $\semavail$ and the fact that $\run_\procB = \alpha \cdot G$.
	Suppose for the sake of contradiction that 
	$\procC \in \blockedset$.
	Because $\semavail$ only adds receivers of already blocked senders to $\blockedset$ and $\semavail^{\procB}_{(G' \ldots)}$ starts with $\blockedset=\{\procB\}$, there must exist a chain of message exchanges $\procD_{i+1} \xrightarrow{} \procD_i: \val_{i}$ in $G'$ with $1 \leq i < n$, $\procB=\procD_{n}$, and $\procC=\procD_1$. That is, $G' \cdot \beta$ must be of the form
	\[
	G' \dots G_{n-1} \xrightarrow{\procB \xrightarrow{} \procD_{n-1}: \val_{n-1}} G_{n-1}' \ldots G_{1} \xrightarrow{\procD_{2} \xrightarrow{} \procC: \val_{1}} G_{1}' \ldots G_{0} \xrightarrow{\procC \xrightarrow{} \procB: m} G_{0}' \ldots\enspace.
	\]
	Let $\val_0=\val$ and $\procD_0=\procB$. We show by induction over $i$ that for all $i \in [1,n]$
	\[
	\alpha \cdot G \xrightarrow{l} G' \dots G_{i} \xrightarrow{\procD_i \xrightarrow{} \procD_{i-1}: \val_{i-1}} G_{i}' ~\preforder~ \run_{\procD_{i}}\enspace.
	\]
	We then obtain the desired contradiction with the fact that $\run_{\procD_n} = \run_{\procB} =  \alpha \cdot G'$. 
	The base case of the induction follows immediately from the construction. 
	For the induction step, assume that
	\[
	\alpha \cdot G \xrightarrow{l} G' \dots G_{i} \xrightarrow{\procD_i \xrightarrow{} \procD_{i-1}: \val_{i-1}} G_{i}' ~\preforder~ \run_{\procD_{i}}\enspace.
	\]
	From the definition of $\run_{\procD_i}$ and the fact that $\procD_{i}$ is the active role in $\rcv{\procD_{i+1}}{\procD_{i}}{\val_i}$, it follows that $\rcv{\procD_{i+1}}{\procD_{i}}{\val_i} \in w$. Hence, we must also have $\snd{\procD_{i+1}}{\procD_{i}}{\val_i} \in w$. Since $\procD_{i+1}$ is the active role in $\snd{\procD_{i+1}}{\procD_{i}}{\val_i}$, we can conclude
	\[
	\alpha \cdot G \xrightarrow{l} G' \dots G_{i} \xrightarrow{\procD_{i+1} \xrightarrow{} \procD_{i}: \val_{i}} G_{i+1}' ~\preforder~ \run_{\procD_{i+1}}\enspace.
	\]
	This concludes the proof of \cref{lm:about-receive-decoration-validity}.
	\proofend
\end{proof}

\characterizationOneTransitionExhaustiveFinalValComplete*
\begin{proof}
	From the negation of \transitionExhaustive, we find a witness trace $v$ such that $v$ is a trace in $\subsetprojCSM$ but not a trace in $\mathcal{A}$, thus contradicting the fact that $\subsetprojCSM$ refines $\mathcal{A}$. 
	Let $\procA$ be a role that violates \transitionExhaustive.
	Let $s$ be a state such that there exists $G \in d_\GG(s)$ with $G \xrightarrow{x} \mathrel{\vphantom{\to}^*} G' \in \projerasuretrans$ but no transition outgoing from $s$ labeled with $x$. 
	By the definition of state decoration, there exists $u \in \Alphabet^*_\procA$ such that $A_\procA$ reaches $s$ on $u$ from its initial state, and the projection by erasure automaton for $\procA$ reaches $G$ on $u$ from its initial state. 
Because $G \xrightarrow{x} \mathrel{\vphantom{\to}^*} G' \in \projerasuretrans$, 
	it holds that $q_{0,\GG} \xrightarrow{u} \mathrel{\vphantom{\to}^*} G \xrightarrow{x} \mathrel{\vphantom{\to}^*} G' \in \projerasuretrans$ is a run in the projection by erasure automaton for $\procA$. 
	Let $\run$ denote this run, and let $w = \SyncToAsync(\trace(\run))$. 
	Then, it holds that $ux \leq w \wproj_{\Alphabet_{\procA}}$. 
Because $\subsetprojCSM$ implements $\GG$, $w$ is a trace of $\subsetprojCSM$. 
	Consequently, $w \wproj_{\Alphabet_{\procA}}$ is a prefix of $A_\procA$. 
	Because $ux$ is a prefix of $w \wproj_{\Alphabet_{\procA}}$, $ux$ is thus also a prefix of $A_\procA$. 
	Because $A_\procA$ is deterministic, $A_\procA$ reaches $s$ on $u$. 
	However, there does not exist an outgoing transition labeled with $x$ from $s$, and we reach a contradiction to the fact that $ux$ is a prefix of $A_\procA$.

	From the negation of \finalVal, we find a witness trace $v$ that is maximally terminated in $\subsetprojCSM$, but not maximally terminated in $\mathcal{A}$, thus contradicting the fact that $\subsetprojCSM$ refines $\mathcal{A}$. 
	Let $\procA$ be a role that violates \finalVal.
	Let $s$ be a state such that there exists $G \in d_\GG(s)$ with $G \in F_\GG$ but $s \notin F_\procA$. 
	Let $w \in \lang(\GG)$ such that $w \wproj_{\Alphabet_{\procA}}$ reaches $G$ in the projection by erasure automaton on $w \wproj_{\Alphabet_{\procA}}$; such a word is guaranteed to exist. 
	Because $\subsetprojCSM$ refines $\mathcal{A}$, $w \in \lang(\mathcal{A})$. 
	Because $A_\procA$ is deterministic, $A_\procA$ reaches $s$ on $w \wproj_{\Alphabet_{\procA}}$.
In other words, in the $\mathcal{A}$ configuration reached on $w$, $A_\procA$ is in state $s$. 
	However, $s \notin F_\procA$. 
	Therefore, $w$ is not terminated in $\mathcal{A}$ and $w \notin \lang(\mathcal{A})$. 
	We reach a contradiction. 
	\proofend
\end{proof}

\characterizationOneSendReceiveDecValComplete*
\begin{proof}
	Because $\GG$ is implementable, $\subsetproj{\GG}{\procA}$ satisfies Send Validity and Receive Validity~\cite[Theorem 7.1]{DBLP:conf/cav/LiSWZ23}. 
	For each condition, we assume the violation of the condition and the fact that $\mathcal{A}$ and $\subsetprojCSM$ are equivalent, and show a contradiction to Send Validity and Receive Validity in turn. 
	
	Let $\procA$ be a role that violates \sendDecVal. 
	Let $s$ be a state and $s \xrightarrow{\snd{\procA}{\procB}{\val}} s'$ be a transition in $A_\procA$ such that
	\[
	\transAnnoFunc(d(s)  \xrightarrow{\snd{\procA}{\procB}{\val}} d(s')) \neq d(s) \enspace.
	\]
	Let $G$ be a state in $d(s) \setminus \transAnnoFunc(d(s)  \xrightarrow{\snd{\procA}{\procB}{\val}} d(s'))$.
	Such a $G$ exists by the negation of \sendDecVal.
	Let $\alpha \cdot G$ be a run in $\semglobalsync(\GG)$; such a run must exist by the fact that $G$ is a syntactic subterm of $\GG$.
	Let $w = \SyncToAsync(\trace(\alpha \cdot G))$. 
	Because $w \in \text{pref}(\lang(\GG))$, it holds that $w$ is a trace of $\subsetprojCSM$. 
	Because $\subsetprojCSM$ refines $\mathcal{A}$ by assumption, $w$ is a trace in $\mathcal{A}$, and there exists an $\mathcal{A}$ configuration reached on $w$ in which $A_\procA$ is in state $s$. 
	Because send actions are always enabled, $wx$ is a trace in $\mathcal{A}$. 
	Now because $\mathcal{A}$ refines $\subsetprojCSM$, $wx$ is also a trace in $\subsetprojCSM$. 
	By definition, let $t$ be the state of $\procA$ in the $\subsetprojCSM$ configuration reached on~$w$.
	Because $w = \SyncToAsync(\trace(\alpha \cdot G))$, it holds that $w \wproj_{\Alphabet_{\procA}} \in \text{pref}(\lang(\subsetproj{\GG}{\procA}))$, and by~\cref{def:subset-construction}, it holds that $G \in t$. 
	Then, there exists a $t'$ such that $t \xrightarrow{x} t'$ is a transition in $\subsetproj{\GG}{\procA}$. 
	We find a contradiction to Send Validity for this transition by using $G$ as a witness.

	Let $\procA$ be a role that violates \receiveDecVal. 
	Let $s$ be a state and let 
	$s \xrightarrow{\rcv{\procB_1}{\procA}{\val_1}} s_1$, 
	$s \xrightarrow{x} s_2$ be two transitions in $A_\procA$, with 
	$G_2 \in \transAnnoFunDest(d(s) \xrightarrow{x} d(s_2))$ 
	such that 
	\[
	x \neq \rcv{\procB_1}{\procA}{\_}
	\quad \land \quad
	\snd{\procB_1}{\procA}{\val_1} \in \semavail^{\procA}_{(G_2 \ldots)} \enspace.
	\]
	Following the construction in\cite[Theorem 7.1]{DBLP:conf/cav/LiSWZ23}, we can construct a witness trace~$w$ in $\mathcal{A}$ such that both $w \cdot \rcv{\procB_1}{\procA}{\val_1}$ and $w \cdot x$ are traces in $\mathcal{A}$.
	Because $\mathcal{A}$ refines $\subsetprojCSM$ by assumption, both $w \cdot \rcv{\procB_1}{\procA}{\val_1}$ and $w \cdot x$ are also traces in $\subsetprojCSM$. 
	Let $t$ be the state reached by $\subsetprojCSM$ on $w$. 
	Then, there must exist two transitions $t \xrightarrow{\rcv{\procB_1}{\procA}{\val_1}} t'$ and $t \xrightarrow{x} t''$ in $\subsetproj{\GG}{\procA}$. 
	Either $x \in \Alphabet_{\procA,!}$ and No Mixed Choice \cite[Corollary 5.5]{DBLP:conf/cav/LiSWZ23} is violated, or $x \in \Alphabet_{\procA,?}$ and Receive Validity is violated. 
\proofend
\end{proof}

\monolithicRefinementHardness*
\begin{proof}
We show the PSPACE-hardness of the monolithic refinement problem by a reduction from the PSPACE-hard problem of deciding deadlock freedom for 1-safe Petri nets~\cite{DBLP:journals/eik/EsparzaN94}. 
Let $(N, M_0)$ be a 1-safe Petri net, with $N = (S, T, F)$.

We construct a CSM $\mathcal{A}_N$ and a global type $\GG_N$ such that $\mathcal{A}_N$ refines $\GG_N$ if and only if the Petri net is deadlock-free. 

We first describe the construction of $\mathcal{A}_N$. 
$\mathcal{A}_N$ consists of one state machine per place in $S$, one state machine per transition in $T$, and one special coordinator role, which we denote $\procA$. 
Each place state machine tracks whether its place is marked by $0$ or $1$, and responds to messages to increment or decrement its marking.
Each transition state machine communicates with its input and output place state machines to check whether its transition is enabled, and to update place markings. 
The coordinator $\procA$ first asks each transition state machine whether its transition is enabled. 
This querying can be performed in an arbitrary fixed order on $T$. 
If at least one transition is enabled, $\procA$ then non-deterministically picks a transition to fire. 
Depending on whether the picked transition is enabled, the input and output place state machines update the configuration, and the transition state machine returns the control flow to $\procA$, which repeats this process with the new configuration. 
If no transition is enabled, $\procA$ enters a sink state with no outgoing transitions, thus causing a deadlock in $\mathcal{A}_N$.

Each message exchange between roles is echoed with an acknowledgement, and the CSM thus constructed is 1-bounded: there is at most one message in flight at any point during its execution. 
Intuitively, $\mathcal{A}_N$ simulates the firing of transitions in the Petri nets via message exchanges, and represents all valid execution traces of the Petri net as CSM traces.

Correspondingly, we construct a global type $\GG_N$ whose language includes not only all execution traces of $\mathcal{A}_N$, but also traces that do not correspond to valid execution traces in the Petri net. 
$\GG_N$ achieves this by mimicing the control flow of the $\mathcal{A}_N$, but decoupling the message contents from the underlying Petri net configuration: at each control flow point, roles non-deterministically choose a message to send. 

If the Petri net is deadlock-free, then $\mathcal{A}_N$ is also deadlock-free and $\lang(\mathcal{A}_N)$ includes only infinite words: because each configuration has at least one enabled transition, $\procA$'s sink state will never be reached.
Because $\lang(\mathcal{A}_N) \subseteq \lang(\GG_N)$ by construction, it holds that $\mathcal{A}_N$ refines $\GG_N$. 
On the contrary, if $\mathcal{A}_N$ refines $\GG_N$ and is thus deadlock-free, then the Petri net is also deadlock-free, as $\mathcal{A}_N$ can simulate all valid execution traces of the Petri net.

\end{proof}

\soundnessCharacterizationTwo*
\begin{proof}
	First, we prove that any trace in \subCSM is a trace in \superCSM: 
	
	\noindent
	\textit{Claim 1: } $\forall~w \in \AlphAsync^*$.~$w$ is a trace in \subCSM $\implies$ $w$ is a trace in \superCSM.

	We prove the claim by induction on $w$. 
	The base case, where $w = \emptystring$, is trivially discharged by the fact that $\emptystring$ is a trace of all CSMs.
	In the inductive step, assume that $w$ is a trace of \subCSM. 
	Let $x \in \AlphAsync$ such that $wx$ is a trace of \subCSM. 
	We want to show that $wx$ is also a trace of \superCSM. 
	
	From the induction hypothesis, we know that $w$ is also a trace of \superCSM. 
	Let $\xi$ be the channel configuration uniquely determined by $w$.
	Let $(\vec{s},\xi)$ be the \subCSM configuration reached on $w$, and let $(\vec{t},\xi)$ be the \superCSM configuration reached on $w$. 
	
	Let $\procB$ be the role such that $x \in \Alphabet_\procB$, and let $s$, $t$ denote $\vec{s}_\procB$, $\vec{t}_\procB$ from the respective CSM configurations reached on $w$ for \subCSM and \superCSM.  
	
	To show that $wx$ is a trace of \superCSM, it suffices to show that there exists a state $t'$ and a transition $t \xrightarrow{x} t'$ in $B$.  
	
	By the definition of state decoration (\cref{def:state-decoration-supertype}), it follows that $t \in d_B(s)$.
	Because \superCSM refines $\GG$ and is deadlock-free, it holds that all traces of \superCSM are prefixes of $\lang(\GG)$.
	In other words, $w \in \text{pref}(\lang(\GG))$. 
	Let $\run$ be a run such that $\run \in I(w)$; such a run must exist from \cite[Theorem 6.1]{DBLP:conf/cav/LiSWZ23} and \cite[Lemma 6.3]{DBLP:conf/cav/LiSWZ23}. 
	Let $\alpha \cdot G \xrightarrow{l} G' \cdot \beta$ be the unique splitting of $\run$ for $\procB$ matching $w$. 
	From \cref{def:state-decoration}, it holds that $G \in d(t)$.

	We proceed by case analysis on whether $x$ is a send or receive event. 
	\begin{itemize}
		\item Case $x \in \Alphabet_{\procA,!}$. 
		Let $x = \snd{\procA}{\procB}{\val}$. 
By assumption, there exists $s \xrightarrow{\snd{\procA}{\procB}{\val}} s'$ in $\delta_A$. 
		We instantiate \sendDecSubVal from \CharacterizationTwo with this transition to obtain: 
		\[
		\transAnnoFunc_B(d_B(s) \xrightarrow{\snd{\procA}{\procB}{\val}} d_B(s')) = d_B(s)
		\enspace .
		\]
		From $t \in d_B(s)$, it follows immediately that there exists $t'$ such that $t \xrightarrow{x} t'$ is a transition in $B$. 
		\item Case $x \in \Alphabet_{\procA,?}$. 
		Let $x = \rcv{\procB}{\procA}{\val}$. 
		
		We proceed by case analysis on $\SyncToAsync(l) \wproj_{\Alphabet_{\procA}}$. 
		In the case that $\SyncToAsync(l) \wproj_{\Alphabet_{\procA}} \in \Alphabet_{\procA,?}$, from~\cref{lm:stepping-stone-completeness-rest} there exists a transition $t \xrightarrow{\SyncToAsync(l) \wproj_{\Alphabet_{\procA}}} t'$ in $\delta_B$, and from \receiveExhaustiveSub there exists a transition $s \xrightarrow{\SyncToAsync(l) \wproj_{\Alphabet_{\procB}}} s''$ in $\delta_A$. 
		We can apply \cref{lm:about-receive-decoration-validity} with $\run$ to conclude that $\SyncToAsync(l) \wproj_{\Alphabet_{\procA}} = x$: we satisfy the assumption that $\snd{\procB}{\procA}{\val} \notin \semavail^{\procA}_{(G' \ldots)}$ by instantiating \receiveDecSubVal with $s \xrightarrow{x} s'$, $s \xrightarrow{\SyncToAsync(l) \wproj_{\Alphabet_\procB}} s''$, and $G'$. 
The fact that $t' \in \transAnnoFunDest_B(d_B(s) \xrightarrow{\SyncToAsync(l) \wproj_{\Alphabet_\procA}} d_B(s''))$ follows from the existence of $t \xrightarrow{\SyncToAsync(l) \wproj_{\Alphabet_{\procA}}} t'$ in $\delta_B$ and the definition of state decoration (\cref{def:state-decoration-supertype}).
		The fact that $G' \in \transAnnoFunDest_B(d(t) \xrightarrow{\SyncToAsync(l) \wproj_{\Alphabet_\procA}} d(t'))$ follows from the fact that $\alpha \cdot G \xrightarrow{l} G' \cdot \beta$ is a run in $\GG$ and \cref{def:state-decoration}.

		In the case that $\SyncToAsync(l) \wproj_{\Alphabet_{\procA}} \in \Alphabet_{\procA,!}$, we again prove a contradiction. 
		Because $G$ is a send-originating global state, \sendPreservationSub guarantees that there exists a transition $s \xrightarrow{x'} s''$ in $A$ such that $x' \in \Alphabet_{\procA,!}$. 
		By \sendDecVal, $x'$ originates from $G$ in the projection by erasure, and we can find another run $\run'$ such that 
		$\alpha' \cdot G \xrightarrow{l'} G'' \cdot \beta'$ is the unique splitting for $\procA$ matching $w$, and $\SyncToAsync(l') \wproj_{\Alphabet_{\procA}} = x'$. 
		
		We can instantiate \cref{lm:about-receive-decoration-validity} with $\run'$ and $\snd{\procB}{\procA}{\val} \notin \semavail^{\procA}_{(G'' \ldots)}$ as above to obtain $\SyncToAsync(l') \wproj_{\Alphabet_{\procA}} = x$, which is a contradiction: $x$ is a receive event and $\SyncToAsync(l') \wproj_{\Alphabet_{\procA}}$ is a send event.
\end{itemize}
	This concludes our proof of Claim 1. 
	
	Next, we show that any trace that terminates in \subCSM also terminates in \superCSM and is maximal in \subCSM.
	
	\noindent
	\textit{Claim 2: } $\forall~w \in \AlphAsync^*$.~$w$ is terminated in \subCSM $\implies w$ is terminated in \superCSM and $w$ is maximal in \subCSM. 
	
	Let $w$ be a terminated trace in \subCSM. 
	Let $\xi$ be the channel configuration uniquely determined by $w$.
	Let $(\vec{s},\xi)$ be the \subCSM configuration reached on $w$, and let $(\vec{t},\xi)$ be the \superCSM configuration reached on $w$. 
	Let $s$, $t$ denote $\vec{s}_\procA$, $\vec{t}_\procA$.
	First suppose by contradiction that $w$ is not terminated in \superCSM.
	Because the state machines for all non-$\procA$ roles are identical between the two CSMs, and because \superCSM is deadlock-free by assumption, it must be the case that $\procA$ witnesses the non-termination of $w$, in other words, $B$ can perform an action that $A$ cannot. 
	Let $x$ be the action that $\procA$ can perform from~$t$.
	Let $G$ be a state in $d(t)$, such a state is guaranteed to exist by Claim 1 and the fact that no reachable states in $B$ have empty decorating sets. 
	Then, $w \wproj_{\Alphabet_{\procA}}$ reaches $G$ from the initial state in the projection by erasure automaton. 
	By the fact that $w$ is a trace of \subCSM, it holds that there exists a run with trace $w \wproj_{\Alphabet_{\procA}}$ in $A$. 
	By the definition of state decoration, $t \in d_B(s)$. 
	\begin{itemize}
		\item If $x \in \Alphabet_!$, it follows that $G$ is a send-originating global state. 
		By \sendPreservationSub, for any state in $A$ that is decorated by a state in $B$ that itself is decorated by at least one send-originating global state, of which $t$ is one, there exists a transition $s \xrightarrow{x'} s'$ such that $x' \in \Alphabet_{\procA,!}$. 
		Because send transitions in a CSM are always enabled, role $\procA$ can take this transition in \subCSM. 
		We reach a contradiction to the fact that $w$ is terminated in \subCSM. 
		\item If $x \in \Alphabet_?$, it follows that $G$ is a receive-originating global state. 
		From \receiveExhaustiveSub, any receive action that originates from any global state in $d(t)$ for any state $t \in d_B(s)$ must also originate from $s$. 
		Therefore, there must exist $s'$ such that $s \xrightarrow{x} s'$ is a transition in $A$. 
		Thus, role $\procA$ can take this transition in \subCSM.
		We again reach a contradiction to the fact that $w$ is terminated in \subCSM. 
	\end{itemize}
	To see that every terminated trace in \subCSM in maximal, from the above we know that $w$ is terminated in \superCSM. From the fact that \superCSM is deadlock-free, $w$ is maximal in \superCSM: all states in $\vec{t}$ are final and all channels in $\xi$ are empty. 
	Because $t$ is a final state, by that fact that \superCSM refines $\GG$ there exists a global state $G \in t$ such that the projection erasure automaton reaches $G$ on $w \wproj_{\Alphabet_{\procA}}$ and $G$ is a final state. 
	Because $A$ reaches $s$ on $w \wproj_{\Alphabet_{\procA}}$, by~the definitions of state decorations (\cref{def:state-decoration,def:state-decoration-supertype}), it holds that $G \in \underset{t \in d_{B}(s)} \bigcup d(t)$.
	By \finalVal, it holds that $s$ is a final state in $A$.
	This concludes our proof that any terminated trace in \subCSM is also a terminated trace in \superCSM, and is maximal in \subCSM. 
	
	Together, Claim 1 and 2 establish that \subCSM satisfies language inclusion with respect to \superCSM (\cref{def:refinement-language-inclusion}), and deadlock freedom (\cref{def:refinement-deadlock-freedom}). 
	It remains to show that \subCSM also satisfies subprotocol fidelity (\cref{def:refinement-subprotocol-fidelity}). 
	This follows immediately from \cite[Lemma 22]{DBLP:conf/concur/MajumdarMSZ21}, which states that all CSM languages are closed under~$\interswap$. 
	\proofend
\end{proof}

 }

\end{document}